\documentclass[preprint,12pt]{elsarticle}

\usepackage{amsmath}
\usepackage{amsthm}
\usepackage{amssymb}
\usepackage[all]{xy}

\journal{\ }

\begin{document}

\newtheorem{theorem}{Theorem}
\newtheorem{lemma}{Lemma}
\newtheorem{corollary}[theorem]{Corollary}
\newtheorem{definition}{Definition}
\newtheorem{proposition}{Proposition}
\begin{frontmatter}



\title{On the Capacity of  Noisy Computations}



\author{Fran\c{c}ois Simon \fnref{affiliation}}

\address{Telecom SudParis, 9 rue Charles Fourier, 91011 Evry, FRANCE}
\fntext[affiliation]{Institut Mines-Telecom ; Telecom SudParis ; CITI}
\ead{Francois.Simon@telecom-sudparis.eu}

\begin{abstract}
This paper presents an analysis of the concept of {\em capacity}  for  {\em noisy computations}, i.e.  functions implemented by unreliable or random  devices.  An information theoretic model of noisy computation of a perfect  function $f$ (measurable function between sequence spaces) thanks to an unreliable device (random channel) $F$ is given: a noisy computation is a product $f\times F$ of channels.  A model of reliable computation based on input encoding and output decoding is also proposed. These models extend those of noisy communication channel and of reliable communication through a noisy channel. The capacity of a noisy computation  is defined and justified by a coding theorem and a converse. Under some constraints on the encoding process, capacity is the upper bound of  {\em input rates} allowing reliable computation, i.e. decodability of noisy outputs into expected outputs.  

These results hold when the one-sided random processes under concern are  asymptotic mean stationary (AMS)  and ergodic. Conditions on $f$ and $F$ for the noisy computation $f\times F$ be AMS and ergodic are also given. 

In addition, some characterizations of AMS and ergodic noisy computations are given based on stability properties of the perfect function $f$ and of the random channel $F$. These results are derived from the more general framework of channel products.
Finally, a way to apply the noisy and reliable computation models to cases where the perfect function $f$ is a defined according to a formal computational model is proposed.
\end{abstract}


\end{frontmatter}



\section{Introduction \label{SectionIntroduction}}

Reliable computation with unreliable devices, or in the presence of noise to use a terminology close to the Information Theory field, has been the subject of numerous works within the vast field of fault-tolerant computing but much less in Information Theory.

A major difference between reliable communication and reliable computation in the presence of noise is that unreliable computing devices can be made reliable not only using information redundancy but also component redundancy.  For example,  self-checking circuits may rely on information redundancy but also on gate redundancy, see  \cite{Nicolaidis93}. Among theoretical analyses, some aim at identifying theoretical boundaries on the amount of necessary and/or sufficient redundancy to achieve reliability.   

Recent references continue to extend the stream opened by Von Neumann's seminal paper \cite{VonNeumann56}, for example \cite{Spielman96}, \cite{Gacs05}, \cite{Hadjicostisverghese05}, \cite{Romashchenko06} and \cite{RachlinSavage08}. These works identify theoretical limits or bounds (e.g., depth and size of circuits) but also propose frameworks to design reliable computations mainly thanks  to gate redundancy. Nevertheless these papers  do not address the question  of a capacity for noisy computations.

When it exists, the capacity of a channel sets a boundary between codes which allow reliable communication and codes which cannot. Channel capacity can be viewed as the upper limit of the ratio $k/n$ of encoding $k$-length input blocks in $n$-length blocks of binary symbols to allow almost perfect error correction for arbitrarily large $k$ and $n$.

The question of whether a noisy computation possesses a capacity like communication channels do (or equivalently whether some coding theorems for noisy computation hold and in which cases) has been raised by \cite{Elias58}. There is a practical consequence in answering positively this theoretical question. This would  mean that, given a noisy implementation of an expected function, it is possible to find families of efficient  input codes which asymptotically allow an almost perfect computation. Efficiency means having an input encoding rate which could remain strictly positive or arbitrarily close to a capacity when the length of the code tends to infinity. 

While the concept of capacity has been thoroughly studied, in Information Theory, for data communication, it has not been the case for computation and very few positive results have been obtained. 

In a strongly constrained context (independent encoding of operands for bit-by-bit boolean operations), \cite{Elias58} obtained negative first results on the existence of a noisy computation capacity . This work was deepened by  \cite{PetersonRabin59} and by  \cite{Winograd62}. One of the major conclusions of these studies was that reliable computation with asymptotic positive rate (the ratio $k/n$ of encoding $k$-length input block in $n$-length blocks of binary symbols) in the presence of noise  is not possible for some boolean operations (e.g., AND) under some assumptions (independent coding of operands, bijective decoding and bit-by-bit operation). This led to the conclusion that, under these assumptions, there is no capacity for such noisy operations. It is worth noting that the assumptions were made to forbid the reliable encoder and decoder to ``participate'' to the computation of the expected operation.

\cite{Ahlswede84} went into the subject in greater depth and made an important contribution  in characterizing  contexts  in which  a capacity for noisy computations cannot exist. It appears that the characteristics of the decoding function play a fundamental role. If the inverse of the decoding function is injective and monotonic  then weak converse theorems hold for the average and maximal error probabilities. If, in addition, the inverse of the decoding function preserves the logical AND (this implies monotonicity), then strong converses hold\footnote{Weak converse states that the rate of encoding tends to $0$ when the block code length tends to infinity and for an error probability which tends to zero. Strong converse states the same result for an error probability arbitrarily close to 1.}. The hypotheses made by \cite{Elias58}, \cite{PetersonRabin59} and \cite{Winograd62}, i.e., independent encoding of operands and bijective decoding, imply monotonicity of the inverse of the decoding function. On this aspect, \cite{Ahlswede84} supersedes \cite{Elias58}, \cite{PetersonRabin59} and \cite{Winograd62}.

Nevertheless, these negative results do not imply the absolute impossibility to identify a capacity for noisy computation. They characterize codes and the companion encoding and decoding processes which cannot open this ability. To relevantly define a capacity for noisy computations, the  assumptions must be relaxed.

To the author's best knowledge, the first positive answer given through a definition of a capacity of a noisy computation (in fact similar to the one for a noisy channel) and a companion coding theorem came from   \cite{WinogradCowan63}.   \cite{WinogradCowan63} considers the special case of noisy functions called decomposable modules. Decomposable modules are noisy functions  which  can be modeled by a perfect function followed by a noisy transmission channel: the error probability depends on the desired output value not on the input value. The quantity $max_X(I(f(X);Z))$ is shown to be a capacity, where $X$ is the input process, $f$ is the expected perfect function,  $Z$ the decomposable module output and $I(f(X);Z)$ is the mutual information between the expected output process and the noisy one. These peculiar  noisy functions make the context equivalent to  that where the reliable encoder computes and encodes the expected function result before transmission through a noisy channel. Due to the restriction of considering decomposable modules, \cite{WinogradCowan63} did not completely succeed in proposing a noisy computation capacity  in a general scope (see  theorem 6.3, pages 47-48 in \cite{WinogradCowan63}). In fact the key property of decomposable modules is that $X\rightarrow f(X) \rightarrow Z$ is a Markov chain (see \ref{SubsectionDecomposableModules} below). 

Noisy computation capacity is also considered in reliable reconstruction of a function of sources over a multiple access channel. Much more recently, a definition of noisy computation capacity is established by  \cite{NazerGastpar07} and is totally consistent with the one proposed here.  Nazer and Gastpar demonstrate the possible advantages of joint source-channel coding of multiple sources over a separation-based scheme, allowing a decoder to retrieve a value which is a function of input sources. This context makes relevant the proposed distributed encoding process which perfectly performs a computation equivalent to the desired function. The encoder outputs are then transmitted through a noisy multiple access channel (MAC) to a decoder, see proofs of Theorems 1 and 2 of \cite{NazerGastpar07}. This models a noisy computation as a perfect computation followed by a noisy transmission of the result like \cite{WinogradCowan63} and, thus, does not cover in full generality noisy computation of functions. It can be noticed that \cite{WinogradCowan63} and \cite{NazerGastpar07} relax the assumptions of \cite{Elias58}, \cite{PetersonRabin59}, \cite{Winograd62} and \cite{Ahlswede84} in a similar way: all goes as if the operands are jointly coded into an encoded form of the expected function result before being handled by the noisy device which transmits with possible noise. Joint coding means that, for example, a pair of operands $(x_1,x_2)$  is encoded in a pair $(u_1,u_2)$ where $u_1$ is obtained from $(x_1,x_2)$ and not only from $x_1$ (the same for $u_2$).

The question close to the fundamental one raised by Elias is under what conditions,  for an algorithm and a noisy implementation of it, does there exist a family of codes which permits almost perfectly reliable computation and for what cost ? To the author's best knowledge, no general answer to this question has been given.

The present paper extends and makes more rigorous ideas sketched in \cite{Simon10} and \cite{Simon11}. Noisy computations are addressed following an information theoretic approach and not a computability theoretic one.

Part \ref{PartSourcesAndChannels} is devoted to preliminaries on sources and channels. Section \ref{SectionSources} and Section \ref{SectionChannels} review the definitions and key properties of information sources, which are random processes, and random channels, which are probability kernels. Some technical lemmas and some properties of one-sided channels used latter are proved. Channel products are defined.

Part \ref{PartCapacity} deals with the main subject of this paper.
Section~\ref{SectionModelForNoisyComputations} proposes a model for noisy computations which is in fact an extension of the classical approach for noisy channels: a noisy computation is a channel $X\rightarrow (f(X),Z)$ where $f$ is the perfect or expected function and $Z$ the (possibly noisy) output of a random channel. An important remark is that functions under consideration are not, at this stage, assumed to be computable according to a formal computational model (e.g., Turing computable functions). They are measurable functions between sequence spaces.  Section~\ref{SectionTypicalInputRateTypicalInputCapacity}  motivates and defines the notions of typical input rate and typical input capacity of a noisy computation. The adjective ``typical'' comes from the fact that the encoding process assumed in  section~\ref{SectionJointSourceComputationCodingTheorem} will be constrained in using only typical sequences.  Section~\ref{SectionFeinsteinTheorem} extends Feinstein theorem and the so-called Feinstein channel codes to noisy computations.  Section~\ref{SectionJointSourceComputationCodingTheorem}  proposes a complete model of reliable computation of a function thanks to a sequence ``encoder-noisy computation-decoder'' and then states and proves  a "joint source-computation coding theorem" and its converse. These results assume that the functions under concern are unary. $m$-ary functions can be modeled as unary ones by concatenating $m$ input values in one ``meta''-input and thus modeling a {\em joint coding} of operands. This relaxes the assumption of independent coding of operands made by  \cite{Elias58}, \cite{PetersonRabin59}, \cite{Winograd62} and \cite{Ahlswede84}. To some extent, computation induces a generalization of the problem (and so some results) set by communication. The concept of code applies well: a code for computation is a couple of an input code (linked to a function $f$) and an "associated" output code. The capacity of a noisy computation generalizes well the existing concept for communication as the latter appears to be a special case of the former. Finally, as an example,  Section \ref{SectionExampleComputationWithNoisyInputs} establishes that correcting results of a computation on noisy input offers a better input rate that correcting noisy inputs before computation.

The results established in part \ref{PartCapacity} hold for AMS and ergodic noisy computations. Then the characterization of classes of AMS and ergodic noisy computations is an important question in this context. Part \ref{PartErgodicAMSNoisyComputations} is a first analysis of AMS and ergodic noisy computations. Some results are obtained for channel products and straightforwardly applied to noisy computations. In Section \ref{SectionChannelProduct}, it is proved that a channel product is resp. stationary, AMS (w.r.t. a stationary source), recurrent or output weakly mixing  if both channels are resp. stationary, AMS (w.r.t. a stationary source), recurrent or output weakly mixing. Based on these results, Section \ref{SectionSufficientConditionsForErgodicity} gives sufficient condition for ergodicity of channel products and thus noisy computations. Section \ref{SectionNoisyComputationsAndComputableFunctions} proposes a way to   apply the noisy computation model (and thus the reliable computation model, coding theorem \ldots) to cases where the perfect function is defined in a formal computational model and its noisy realization in the corresponding formal noisy computation model. In fact, this is obtained through the more general case where the perfect function is a function acting on finite length sequences and producing finite length sequences. Two very simple examples are given.

\pagebreak 

\part{Preliminaries: information sources and channels} \label{PartSourcesAndChannels}
\pagebreak

\section{Information Sources}\label{SectionSources}

This section reviews the definition of information sources and  some properties of sources.

 Information sources can be described following two equivalent models: random processes and dynamical systems (see \cite{Gray11}). Both models will be used below. Notations follow partly \cite{Gray11} and partly \cite{Kakihara99}.

\begin{definition}
Let $\mathcal{I}$ be a countable index set. Let $(A,\mathcal{B}_A)$ be a measurable space, called the alphabet. Let $(A^\mathcal{I}, \mathcal{B}_{A^\mathcal{I}})$ be the measurable sequence space where  $\mathcal{B}_{A^\mathcal{I}}$ is the $\sigma$-field generated by the sets of rectangles $\{x=(x_i)_{i\in \mathcal{I}} \in A^\mathcal{I} / x_i\in B_i,  B_i\in \mathcal{B}_A \text{ for any }i\in \mathcal{J}\}$, $\mathcal{J}$  finite subsets of $\mathcal{I}$. A source $[A,X]$ is a random process $X= \{ X_i ; i\in \mathcal{I} \}$ with values in  $(A^\mathcal{I}, \mathcal{B}_{A^\mathcal{I}})$. The distribution of the source $[A,X]$ is denoted by $P_X$.
\end{definition}
If the index set $\mathcal{I}$ is $\mathbb{N}$ then the process $X$ is said to be one-sided. If the index set $\mathcal{I}$ is $\mathbb{Z}$ then the process $X$ is said to be two-sided. Let $T_A: A^\mathcal{I} \rightarrow A^\mathcal{I}$ be the shift transform on $A^\mathcal{I}$. For a one-sided process:
$$
T_A(x_0 x_1 \ldots x_i \ldots)=x_1 x_2 \ldots x_{i+1} \ldots
$$
For a two-sided process:
$$
T_A(\ldots x_{-j}\ldots x_0 x_1 \ldots x_i \ldots)=\ldots x_{-j+1}\ldots x_1 x_2 \ldots x_{i+1} \ldots
$$
In the latter case, the shift $T_A$ is invertible. 

$T_A$ is  $(\mathcal{B}_{A^\mathcal{I}},\mathcal{B}_{A^\mathcal{I}})$-measurable. If $\mu$ is a probability on $(A^\mathcal{I},\mathcal{B}_{A^\mathcal{I}})$ then $(A^\mathcal{I},\mathcal{B}_{A^\mathcal{I}}, T_A,\mu)$ is a dynamical system.
\begin{definition}
A source on the alphabet $(A, \mathcal{B}_A)$ is a dynamical system \\ $(A^\mathcal{I},\mathcal{B}_{A^\mathcal{I}}, T_A,\mu)$.
\end{definition}
Let $\Pi_0$ denote the ``zero-time sampling function'': $\Pi_0(x)=x_0$ for any $x=x_0x_1\ldots x_i\ldots$ (or $x=\ldots x_{-j}\ldots x_0 x_1 \ldots x_i \ldots$ in the two-sided case). The two definitions of a source are equivalent. A dynamical system $(A^\mathcal{I},\mathcal{B}_{A^\mathcal{I}}, T_A,\mu)$ determines a random process $X= \{ X_i ; i\in \mathcal{I} \}$ where $X_i(\omega)=\Pi_0(T_A^i(x))=x_i$ and the distribution of $X$ is $\mu$: $P_X=\mu$.

In the sequel, each model will be used when the most relevant. Moreover, the word source will be used to name either the random process, the dynamical system, the distribution $P_X$ or the probability $\mu$. If $x=x_0 \ldots x_n \ldots \in A^\mathcal{I}$, $x^n$ stands for $x_0\cdots x_{n-1}$.

\begin{definition}\label{DefinitionNthExtensionOfASource}
Let $[A,X]$ be a source with distribution $P_X$. The $n^{th}$ extension of the source $[A,X]$ is the random value $X^n=(X_0,\ldots,X_{n-1})$. 
\end{definition}

The distribution $P_{X^n}$ of the $n^{th}$ extension of the source $[A,X]$ is given by:
$$
\forall E \in \mathcal{B}_{A^n},  P_{X^n}(E)=P_X(c(E))
$$
where $c(E)=\{ x\in A^\mathcal{I} / x^n \in E\}$.

\begin{definition}\label{DefinitionSourceStabilityProperties}
A source $[A,X]$ with distribution $\mu$ is {\em stationary} if $\forall E \in \mathcal{B}_{A^\mathcal{I}}$, $\mu(T_A^{-1}E)=\mu(E)$. 

 A source is {\em asymptotically mean stationary} (AMS) if $\forall E \in \mathcal{B}_{A^\mathcal{I}}$, the limit $\lim_{n \to \infty}\frac{1}{n} \sum_{k=0}^{n-1}\mu(T^{-k}E)$ exists. If $\overline{\mu}(E)$ is this limit, then $\overline{\mu}$ is a stationary probability called the stationary mean of $\mu$. 

A source is {\em recurrent} if,  $\forall E \in \mathcal{B}_{A^\mathcal{I}}$, $\mu(E \setminus \cup_{k\geq 1} T_{A}^{-k}E)=0$.

A recurrent and AMS source is called R-AMS. 

A source with distribution $\mu$ is {\em ergodic} if for any $T_A$-invariant $E \in \mathcal{B}_{A^\mathcal{I}}$ (i.e., $T_A^{-1}E=E$), $\mu(E)=0$ or $\mu(E)=1$. 
\end{definition}

If $\mu$ is AMS, then its stationary mean $\overline{\mu}$ asymptotically dominates $\mu$ ($\mu \ll^a \overline{\mu}$): $\overline{\mu}(E) =0 \Rightarrow \lim_{n \to \infty} \mu(T^{-n}E)=0$ (\cite{Gray09}).  $\mu$ is AMS if and only if  there exists a stationary distribution $\eta$ such that $\mu \ll^a \eta$.

A stationary source is recurrent (\cite{Gray09}).  An AMS source distribution $\mu$ is dominated by its stationary mean ($\mu \ll \overline{\mu}$: $\overline{\mu}(E)=0 \Rightarrow \mu(E)=0$) if and only if the source is recurrent, see \cite{Gray09}, theorem 7.4.

The next lemma is a rephrasing of Theorem 3 of \cite{GrayKieffer80}.
\begin{lemma}\label{LemmaDominanceOnTailSigmaField}
Let $[A,X]$ be a source with distribution $\mu$,  let $[A,Y]$ be a stationary source with distribution $\eta$. Let $(\mathcal{B}_{A^\mathcal{I}})_{\infty}$ denote the tail $\sigma$-field $\cap_{n\geq 0}T_A^{-n}\mathcal{B}_{A^\mathcal{I}}$,  $\mu_{\infty}$ and $\eta_{\infty}$ denote the restrictions of $\mu$ and $\eta$ to $(\mathcal{B}_{A^\mathcal{I}})_{\infty}$. Then $\eta$ asymptotically dominates $\mu$ if and only if $\eta_\infty$ dominates $\mu_\infty$: 
$$
\mu \ll^a \eta \Leftrightarrow \mu_\infty \ll \eta_\infty
$$
\end{lemma}

The following lemma asserts that asymptotic mean stationarity and ergodicity jointly imply recurrence.
\begin{lemma}\label{LemmaErgodicAMSIsRecurrent}
An ergodic and AMS process is recurrent.
\end{lemma}
\begin{proof}
Let $\mu$ be an ergodic AMS process distribution on the space $(\Omega, \mathcal{B})$. $\mu$ is AMS then $\mu \ll^a \overline{\mu}$ and $\overline{\mu}$ is ergodic. 

From the Lebesgue decomposition theorem there exist  two probabilities $\mu_1$, $\mu_2$ on $(\Omega, \mathcal{B})$ and $\lambda \in [0,1]$ such that $\mu=\lambda \mu_1 + (1-\lambda)\mu_2$, $\mu_1 \ll \overline{\mu}$, $\mu_2$ and $\overline{\mu}$ are mutually singular.

$\overline{\mu}$ is stationary (thus recurrent) ergodic and dominates $\mu_1$ then $\mu_1$ is AMS ergodic and recurrent.  This implies (by theorem 7.4, page 220, of \cite{Gray09}) that $\mu_1 \ll \overline{\mu}_1$. $\mu_1$ being ergodic so is $\overline{\mu}_1$. $\overline{\mu}$ and $\overline{\mu}_1$ are stationary ergodic probabilities on the same space, they are thus equal or mutually singular (\cite{Kakihara99}, Lemma 1 p. 75). Since they both dominate the same probability ($\mu_1$), they are  equal: $\overline{\mu}=\overline{\mu}_1$

$\mu=\lambda \mu_1 + (1-\lambda)\mu_2$ is AMS  and $\mu_1$ is AMS  then, necessarily,  $\mu_2$ is AMS. This also implies that: $\overline{\mu}=\lambda \overline{\mu_1}+ (1-\lambda)\overline{\mu_2}$

Since $\overline{\mu}=\overline{\mu_1}$, necessarily $\lambda=1$, hence $\mu=\mu_1$ which is recurrent.
\end{proof}

\begin{definition}
Let $[A,X]$ be a source. The entropy rate of the source is defined by
$$
\overline{H}(X)= \limsup_{n \to \infty} \frac{H(X_0,\ldots, X_{n-1})}{n}
$$
where $H(X_0,\ldots, X_{n-1})=-\sum_{x^n \in A^n} P_{X^n}(\{x^n\}) ln(P_{X^n}(\{x^n\}))$ is the entropy of the random value $(X_0,\ldots, X_{n-1})$ .
\end{definition}

If a source $[A,X]$ is AMS and ergodic then
$$
\overline{H}(X)= \lim_{n \to \infty} \frac{H(X_0,\ldots, X_{n-1})}{n}
$$
A consequence of the Shannon-Mc Millan-Breimann theorem is the Asymptotic Equipartition Property (\cite{Gray11}):
\begin{theorem} [Asymptotic Equipartition Property]\label{TheoremAEP}
Let $[A,X]$ be an ergodic AMS source. For $\epsilon>0$, let $A^n(\epsilon)$ be the set of $\epsilon$-typical $n$-sequences 
$$
A^n(\epsilon) = \left\{ x^n \in A^n / \left|  \frac{-log(P_{X^n}(x^n))}{n} - \overline{H}(X) \right| < \epsilon \right\}
$$
Then:
\begin{enumerate}
	\item $\lim_{n\rightarrow \infty} P_{X^n}(A^n(\epsilon)) = 1$
	\item for any $\delta>0$ and any $\epsilon >0$, there exists $N$ such that  $\forall n \geq N$:
	$$
	(1-\delta) e^{n(\overline{H}(X)-\epsilon)} \leq |A^n(\epsilon)| \leq e^{n(\overline{H}(X)+\epsilon)}
	$$
where $|A^n(\epsilon)|$ is the cardinality of $A^n(\epsilon)$.
\end{enumerate}
\end{theorem}

\pagebreak 

\section{Channels}\label{SectionChannels}

This section reviews the definition of noisy communication channels and some properties of  channels, used latter, are recalled (references are \cite{FontanaGrayKieffer81} and \cite{Kakihara99}). Some technical lemmas are given for further use. In what follows the shifts under consideration are not assumed to be invertible, channels under consideration may be one-sided.

\subsection{Noisy or random channels}

Let $(A,\mathcal{B}_A)$ and $(B,\mathcal{B}_B)$  be alphabets,  $(A^\mathcal{I}, \mathcal{B}_{A^\mathcal{I}})$ and $(B^\mathcal{I}, \mathcal{B}_{B^\mathcal{I}})$ be the two corresponding sequence spaces. The shifts (assumed non-invertible) on $(A^\mathcal{I}, \mathcal{B}_{A^\mathcal{I}})$ and $(B^\mathcal{I}, \mathcal{B}_{B^\mathcal{I}})$ are respectively denoted $T_A$ and $T_B$.  

Let $\mathcal{B}_{A^\mathcal{I}\times B^\mathcal{I}}$ be the $\sigma$-field generated by the rectangles $\{ E\times G / E \in A^\mathcal{I} , G \in B^\mathcal{I}\}$. $T_A$ and $T_B$ define a measurable shift $T_{AB}$ on the space $(A^\mathcal{I} \times B^\mathcal{I}, \mathcal{B}_{A^\mathcal{I}\times B^\mathcal{I}})$ where $T_{AB}(x,y)=(T_Ax, T_By)$.

\begin{definition}
A noisy communication channel $[A,\nu,B]$  is a function $\nu: A^\mathcal{I} \times \mathcal{B}_{B^\mathcal{I}} \to [0,1]$ such that:
\begin{itemize}
	\item for any $x\in A^\mathcal{I}$, the set function $G \mapsto \nu(x,G)$ is a probability on the space $(B^\mathcal{I}, \mathcal{B}_{B^\mathcal{I}})$
	\item for any $G\in \mathcal{B}_{B^\mathcal{I}}$, the function $x \mapsto \nu(x,G)$ is measurable.
\end{itemize}
If the noisy channel  $[A,\nu,B]$ takes a source $[A,X]$ as input, it produces as an output an information source  $[B,Y]$.  Let $P_{XY}$ be the distribution of the joint process $(X,Y)$ induced by the source $X$ fed into the noisy channel. $P_{XY}$ is defined on rectangles by: 
$$
\forall E\in \mathcal{B}_{A^\mathcal{I}}, \forall G\in \mathcal{B}_{B^\mathcal{I}}, P_{XY}(E\times G) =\int_E \nu(x,G) dP_X
$$
$P_{XY}$ will also be denoted by $\mu\nu$ (the hookup of $\mu$ and $\nu$ where $\mu=P_X$). 
\end{definition}
In fact a noisy communication channel is a probability kernel. In the sequel, the alphabets $(A,\mathcal{B}_A)$ and $(B,\mathcal{B}_B)$  will  be assumed standard. Then the sequence spaces $(A^\mathcal{I}, \mathcal{B}_{A^\mathcal{I}})$ and $(B^\mathcal{I}, \mathcal{B}_{B^\mathcal{I}})$ are also standard. This assumption is made to ensure that  conditional probabilities defined  are regular (see \cite{Gray09} or \cite{Faden85}). Thus, on such spaces, given a joint random process $(X,Y)$, it will always be possible to define a channel $\nu$, unique $P_X$-a.s., taking $X$ as an input and inducing the joint process $(X,Y)$ (\cite{Gray11}). 

Given a channel $[A,\nu,B]$ and a source $[A,X]$, and denoting $[B,Y]$ the corresponding output source, the probabilities $\nu(x,.)$ will also be denoted by $P_{Y|X}(.|x)$.

\begin{definition}
Let $[A,X]$ be a source and $[A,\nu,B]$ be a channel. Let $(X,Y)$ be the generated joint process and $[B,Y]$ be the  output of the channel when $[A,X]$ is the input. The set of probabilities $\{ P_{Y^n|X^n}(.|x^n), x^n \in A^n\}$ such that, $\forall a^n \in A^n$  such that $P_{X^n}(\{a^n\}) \neq 0$, $\forall G \in \mathcal{B}_{ B^n}$: 
$$
 P_{Y^n|X^n}(G|a^n)= \frac{1}{P_X(c(a^n))} \int_{c(a^n)} P_{Y|X}(c(G)|x) dP_X
$$
is  called the induced channel (of order $n$).
\end{definition}
Obviously, since $P_{XY}(c(\{ a^n \}) \times c(G)) =  \int_{c(a^n)} P_{Y|X}(c(G)|x) dP_X$
$$
P_{XY}(c(\{ a^n \}) \times c(G)) = P_{X}(c(\{a^n\})).P_{Y^n|X^n}(G|a^n)
$$
thus
$$
P_{X^nY^n}(\{ a^n \} \times G) = P_{X^n}(\{a^n\}).P_{Y^n|X^n}(G|a^n)
$$
and, $\forall E \in \mathcal{B}_{A^n}$:
$$
P_{X^nY^n}(E \times G) = P_{X^n}(E).P_{Y^n|X^n}(G|E)
$$
Given an input source and a channel, the induced channel is the (finite) channel linking the $n^{th}$ extensions of the input and output sources.

\begin{definition}
A channel $[A,\nu,B]$ is {\em stationary} with respect to a stationary source distribution $\mu$ if the hookup $\mu\nu$ is stationary.

A channel $[A,\nu,B]$ is {\em asymptotically mean stationary (AMS)} with respect to an AMS source distribution $\mu$ if the hookup $\mu\nu$ is AMS.

A stationary (resp. AMS) channel $[A,\nu,B]$ is {\em ergodic} with respect to a stationary (resp.  AMS) ergodic source distribution $\mu$ if the hookup $\mu\nu$ is ergodic. 

A channel $[A,\nu,B]$ is {\em output weakly mixing}  w.r.t.  $\mu$ if $\forall G,G' \in \mathcal{B}_{B^\mathcal{I}}$ $\lim_{n \to \infty} \frac{1}{n} \sum_{i=0}^{n-1}  | \nu(x, T_B^{-i}G \cap G') - \nu(x, T_B^{-i}G).\nu(x, G')| =0$ $\mu$-a.e.

A channel $[A,\nu,B]$ is {\em recurrent}  with respect to a recurrent source distribution $\mu$ if the hookup $\mu\nu$ is recurrent. 

A channel $[A,\nu,B]$ is {\em R-AMS}  with respect to an R-AMS source distribution $\mu$ if the hookup $\mu\nu$ is R-AMS. 
\end{definition}
The following lemma is immediate (c.f. \cite{Kakihara99}).
\begin{lemma}\label{LemmaStationaryChannel} 
A channel $[A,\nu,B]$ is {\em stationary} with respect to a stationary source $\mu$ if and only if $\forall G\in \mathcal{B}_{B^\mathcal{I}}$, $\nu(x,T_B^{-1}G) = \nu(T_A x,G) \text{  }\mu\text{.a.e}$.  
\end{lemma}

\begin{definition}
Let $A$ and $B$ be two finite alphabets. Let  $[A,\nu,B]$ be a channel, $[A,X]$  an input source and $[B,Y]$ the corresponding output source. The information rate of $(X,Y)$ is defined by
$$
\overline{I}(X,Y)= \limsup_{n \to \infty} \frac{I(X^n,Y^n)}{n}
$$
where 
$$
I(X^n,Y^n)=\sum_{x^n \in A^n, y^n \in B^n} P_{X^n Y^n}(\{x^n\} \times \{y^n\}) ln\left(\frac{P_{X^n Y^n}(\{x^n\} \times \{y^n\})}{P_{X^n}(\{x^n\}).P_{Y^n}(\{y^n\})}\right)
$$
 is the mutual information of  $(X^n,Y^n)$.
\end{definition}

If the source $[A,X]$ and the channel $[A,\nu,B]$ are AMS and ergodic then (Lemma 8.1, \cite{Gray11}):
$$
\overline{I}(X,Y)= \lim_{n \to \infty} \frac{I(X^n,Y^n)}{n}
$$
In this case, the conditional entropy $\overline{H}(X|Y)=\overline{H}(X)-\overline{I}(X,Y)$ is well defined and a "conditional" Asymptotic Equipartition Property holds.
\begin{theorem} [Conditional AEP]\label{TheoremConditionalAEP}
Let $A$ and $B$ be two finite alphabets. Let  $[A,\nu,B]$ be an ergodic AMS  channel, $[A,X]$  an ergodic AMS input source and $[B,Y]$ the corresponding output source.  For $\epsilon>0$ and $y\in B^\mathcal{I}$, let $A^n(\epsilon,y^n)$ be the set of $\epsilon$-typical sequences given $y^n$: 
$$
A^n(\epsilon,y^n) = \left\{ x^n \in A^n / \left|  \frac{-log(P_{X^n|Y^n}(x^n|y^n))}{n} - \overline{H}(X|Y) \right| < \epsilon \right\}
$$
Then:
\begin{enumerate}
	\item $\lim_{n\rightarrow \infty} P_{X^n|Y^n}(A^n(\epsilon,y^n)|y^n) = 1\text{   }P_Y -a.e.$
	\item for any $\delta>0$ and any $\epsilon >0$, for $n$ large enough:
	$$
	(1-\delta) e^{n(\overline{H}(X|Y)-\epsilon)} \leq |A^n(\epsilon,y^n)| \leq e^{n(\overline{H}(X|Y)+\epsilon)}  \text{  }P^n_Y-a.e.
	$$
\end{enumerate}
\end{theorem}

\subsection{Restriction of a channel to the tail $\sigma$-field}

 The following lemma asserts that it is possible to consistently define the restriction $[A,\nu_\infty,B]$ of the channel $[A,\nu,B]$ meaning a probability kernel such that $\mu_\infty \nu_\infty = (\mu\nu)_\infty$.
\begin{lemma}\label{LemmaChannelRestriction}
Let $\mu$ be the distribution of a source $[A,X]$ and let $\mu_\infty$ denote its restriction to the tail  $\sigma$-field $(\mathcal{B}_{A^\mathcal{I}})_{\infty} = \cap_{n\geq 0}T_A^{-n}\mathcal{B}_{A^\mathcal{I}}$. Let $[A,\nu,B]$ be a channel. Then
\begin{itemize}
	\item $\left( \sigma\left( \mathcal{B}_{A^\mathcal{I}} \times \mathcal{B}_{B^\mathcal{I}} \right) \right)_\infty = \sigma\left( (\mathcal{B}_{A^\mathcal{I}})_\infty \times (\mathcal{B}_{B^\mathcal{I}})_\infty \right)$
	\item $(\mu\nu)_\infty = \mu_\infty \nu_\infty$ where $\nu_\infty(x, .)$ is the restriction of  $\nu(x, .)$ to the tail $\sigma$-field $(\mathcal{B}_{B^\mathcal{I}})_\infty$ $\mu_\infty$-a.e.
\end{itemize}
\end{lemma} 
\begin{proof}
Let $\mathcal{G}=\mathcal{B}_{A^\mathcal{I}} \times \mathcal{B}_{B^\mathcal{I}}$. Since $T_A$ and $T_B$ are respectively  $(\mathcal{B}_{A^\mathcal{I}},\mathcal{B}_{A^\mathcal{I}})$-measurable and  $(\mathcal{B}_{B^\mathcal{I}},\mathcal{B}_{B^\mathcal{I}})$-measurable , $\forall i \in \mathbb{N}$, $T_{AB}^{-(i+1)}\mathcal{G} \subset T_{AB}^{-(i)}\mathcal{G}$ and, thus,  $\sigma(T_{AB}^{-(i+1)}\mathcal{G}) \subset \sigma(T_{AB}^{-(i)}\mathcal{G})$. Then $\forall n \in \mathbb{N}$
$$
\sigma(T_{AB}^{-n}\mathcal{G})= \cap_{i=0}^{n}  \sigma(T_{AB}^{-i}\mathcal{G}) = \sigma(\cap_{i=0}^{n}  T_{AB}^{-i}\mathcal{G})
$$
This implies that $\cap_{i\geq 0}  \sigma(T_{AB}^{-i}\mathcal{G}) = \sigma(\cap_{i\geq 0}  T_{AB}^{-i}\mathcal{G})$, thus 
$$
\left( \sigma\left( \mathcal{B}_{A^\mathcal{I}} \times \mathcal{B}_{B^\mathcal{I}} \right) \right)_\infty = \sigma\left( (\mathcal{B}_{A^\mathcal{I}})_\infty \times (\mathcal{B}_{B^\mathcal{I}})_\infty \right)
$$
Defining $\nu_\infty(x, .)$ as the restriction of  $\nu(x, .)$ to the tail $\sigma$-field $(\mathcal{B}_{B^\mathcal{I}})_\infty$ for $\mu_\infty$-almost $x\in A^\mathcal{I}$, $(\mu\nu)_\infty$ and $\mu_\infty \nu_\infty$ are probabilities defined on the same $\sigma$-field $\left( \sigma\left( \mathcal{B}_{A^\mathcal{I}} \times \mathcal{B}_{B^\mathcal{I}} \right) \right)_\infty$ which coincide on the generating semi-algebra $(\mathcal{B}_{A^\mathcal{I}})_\infty \times (\mathcal{B}_{B^\mathcal{I}})_\infty$. They are thus equal.
\end{proof}

\subsection{Recurrent channels}

The following lemma gives necessary and sufficient conditions for a channel to be recurrent.
\begin{lemma}\label{LemmaRecurrentChannel2}
Let $[A,X]$ be a source with distribution $\mu$ and $[A,\nu,B]$ a channel.  Let $\mathcal{G}=\{ F \times G / F\in \mathcal{B}_{A^\mathcal{I}}, G\in \mathcal{B}_{B^\mathcal{I}} \}$. Then
\begin{itemize}
	\item the dynamical system $(A^\mathcal{I}\times B^\mathcal{I}, \mathcal{B}_{A^\mathcal{I}\times B^\mathcal{I}}, T_{AB} , \mu\nu)$ is recurrent if and only if for any $F\times G \in \mathcal{G}$, $\mu\nu(F\times G \setminus \cup_{k\geq 1}T_{AB}^{-k}F\times G)=0$.

	\item $[A,\nu,B]$ is recurrent with respect to  $\mu$ if and only if $\mu$ is recurrent and $\nu(x,.)$ is recurrent $\mu$-a.e.
\end{itemize}

\end{lemma}
\begin{proof}
Assume that  for any $F\times G \in \mathcal{G}$, $\mu\nu(F\times G \setminus \cup_{k\geq 1}T_{AB}^{-k}(F\times G) )=0$. Let $\mathcal{B}'$ be the set $\{ O \in \mathcal{B}_{A^\mathcal{I} \times B^\mathcal{I}} / \mu\nu(O \setminus \cup_{k\geq 1} T_{AB}^{-k}O) = 0 \}$.  Let $(O_i)_i$ be a countable family of elements of $\mathcal{B'}$. Then 
$$
\mu\nu( (\cup_{i\geq 0} O_i) \setminus (\cup_{k\geq 1} T_{AB}^{-k} (\cup_{i\geq 0} O_i))  = \mu\nu( (\cup_{i\geq 0} O_i) \setminus (\cup_{i\geq 0} \cup_{k\geq 1}   T_{AB}^{-k} O_i))
$$
Given four sets $\Theta$, $\Lambda$, $\Gamma$ and $\Delta$, $(\Theta \cup \Lambda) \setminus (\Gamma \cup \Delta) \subset (\Theta \setminus \Gamma) \cup (\Lambda \setminus \Delta)$, then:
$$
 \mu\nu( (\cup_{i\geq 0} O_i) \setminus (\cup_{k\geq 1} T_{AB}^{-k} (\cup_{i\geq 0} O_i)) \leq \mu\nu( \cup_{i\geq 0} (O_i \setminus \cup_{k\geq 1}   T_{AB}^{-k} O_i) )
$$
$$
\mu\nu( (\cup_{i\geq 0} O_i) \setminus (\cup_{k\geq 1} T_{AB}^{-k} (\cup_{i\geq 0} O_i))   \leq \sum_{i\geq 0} \mu\nu( O_i \setminus \cup_{k\geq 1}   T_{AB}^{-k} O_i) = 0
$$
Thus $\cup_{i\geq 0} O_i \in \mathcal{B}'$. Then $\mathcal{B}'$ is stable by countable union.

Let $\mathcal{F}$ be the field generated by $\mathcal{G}$. Any element of the field $\mathcal{F}$ is a finite union of rectangles. Then any countable union of field elements is a countable union of rectangles. Then, from above, for any element $R$ of $\mathcal{F}$, $\mu\nu(R \setminus \cup_{k\geq 1}T_{AB}^{-k}R)=0$.

Let $O\in \mathcal{B_{A^{\mathcal{I}}\times B^{\mathcal{I}}}}$. 
 The probability $\mu\nu$ on  $( A^{\mathcal{I}}\times B^{\mathcal{I}} , \mathcal{B_{A^{\mathcal{I}}\times B^{\mathcal{I}}}})$ is the extension of  the set function $\mu\nu$ on the field generated by the rectangles and verifies (\cite{Gray09}):
$$
\mu\nu(O)=\inf_{(R_i)_{i\geq 0}/ O\subset \cup_{i\geq 0}R_i}\mu\nu(\cup_{i\geq 0}R_i)
$$
where the families  $(R_i)_{i\geq 0}$ are countable covers of $O$ made of elements of the field generated by the rectangles (see \cite{Gray09}). Let $\epsilon >0$, then there exist  countable families of field elements $(R_i)_{i\geq 0}$ and $(R'_i)_{i\geq 0}$ respectively covering $O$ and $O^c$ such that:
$$
\mu\nu(\cup_{i\geq 0}R_i)-\frac{\epsilon}{2} < \mu\nu(O) \leq \mu\nu(\cup_{i\geq 0}R_i) 
$$
$$
\mu\nu(\cup_{i\geq 0}R'_i)-\frac{\epsilon}{2} < \mu\nu(O^c) \leq \mu\nu(\cup_{i\geq 0}R'_i) 
$$Let $\alpha=\cup_{i\geq 0}R_i$ and $\beta=\cup_{i\geq 0}R'_i$.  Obviously $T_{AB}^{-k}\beta^c \subset T_{AB}^{-k}O \subset T_{AB}^{-k}\alpha$ for any $k$. Then:
$$
O\setminus \cup_{k\geq 1} T_{AB}^{-k}O \subset   \alpha \setminus \cup_{k\geq 1} T_{AB}^{-k}\beta^c \subset (\beta^c \setminus \cup_{k\geq 1} T_{AB}^{-k}\beta^c) \cup (\alpha \setminus \beta^c)
$$
$\beta^c$ is a countable union of  elements of the field generated by rectangles then $\mu\nu(\beta^c \setminus \cup_{k\geq 1} T_{AB}^{-k}\beta^c)=0$. Moreover, $\mu\nu(\alpha \setminus \beta^c) < \epsilon$. Then $\mu\nu(O\setminus \cup_{k\geq 1} T_{AB}^{-k}O)<\epsilon$. This holds for any $\epsilon>0$, then 
$$
\mu\nu(O\setminus \cup_{k\geq 1} T_{AB}^{-k}O)=0
$$
Then the dynamical system $(A^\mathcal{I}\times B^\mathcal{I}, \mathcal{B}_{A^\mathcal{I}\times B^\mathcal{I}}, T_{AB} , \mu\nu)$  is recurrent. The reciprocal is obvious.

Assume that $\mu$ is recurrent and $\nu(x,.)$ is recurrent $\mu$-a.e. Let $F \times G \in \mathcal{G}$.
\begin{eqnarray}
\mu\nu\left( F\times G \setminus \cup_{k\geq 1}T_{AB}^{-k}F\times G \right) & = &  \mu\nu\left( F\times G \cap  \bigcap_{k\geq 1}(T_{AB}^{-k}(F\times G) )^c \right) \nonumber \\
										& = & \mu\nu\left( F\times G \cap \bigcap_{k\geq 1}(T_{A}^{-k}F^c \times T_{B}^{-k} G) \right) \nonumber \\
										&     & +\mu\nu\left( F\times G \cap \bigcap_{k\geq 1}(T_{A}^{-k}F \times T_{B}^{-k} G^c \right) \nonumber \\
										&     &  + \mu\nu\left( F\times G \cap \bigcap_{k\geq 1}(T_{A}^{-k}F^c \times T_{B}^{-k} G^c) \right) \nonumber
\end{eqnarray}

Since $\mu$ is recurrent:
$$
\mu\nu\left( F\times G \cap \bigcap_{k\geq 1}(T_{A}^{-k}F^c \times T_{B}^{-k} G) \right) \leq \mu( F \cap \bigcap_{k\geq 1} T_A^{-k}F^c) = 0
$$
Since $\nu(x,.)$ is recurrent $\mu$-a.e. :
\begin{multline}
\mu\nu\left( F\times G \cap \bigcap_{k\geq 1}(T_{A}^{-k}F \times T_{B}^{-k} G^c \right) \leq \nonumber \\
		\mu\nu(A^\mathcal{I}\times (G \cap_{k\geq 1} T_{B}^{-k} G^c))= \int \nu(x, G \cap_{k\geq 1} T_{B}^{-k} G^c) d\mu = 0 \nonumber
\end{multline}
and
\begin{multline}
\mu\nu\left( F\times G \cap \bigcap_{k\geq 1}(T_{A}^{-k}F^c \times T_{B}^{-k} G^c) \right) = \nonumber \\
		\mu\nu( (F \cap_{k\geq 1} T_{A}^{-k} F^c) \times  (G \cap_{k\geq 1} T_{B}^{-k} G^c))= \nonumber \\
			\int_{F \cap_{k\geq 1} T_{A}^{-k} F^c } \nu(x, G \cap_{k\geq 1} T_{B}^{-k} G^c) d\mu=0 \nonumber
\end{multline}
then $\mu\nu\left( F\times G \setminus \cup_{k\geq 1}T_{AB}^{-k}F\times G \right)=0$.
From above, $[A,\nu,B]$ is recurrent w.r.t. $\mu$. The reciprocal is obvious.
\end{proof}

\subsection{AMS Channels}

Let $\mu$ be the distribution of an AMS source with stationary mean $\overline{\mu}$: $\mu\ll^a \overline{\mu}$. Let $\nu$ be an AMS channel. Then $\mu\nu$ is AMS: $\mu\nu \ll^a \overline{\mu\nu}$. The "input marginal" of $\overline{\mu\nu}$ is $\overline{\mu}$, hence there exists a (unique modulo $\overline{\mu}$) channel $\overline{\nu}_\mu$ stationary w.r.t $\overline{\mu}$ such that $\overline{\mu\nu}=\overline{\mu}\ \overline{\nu}_\mu$. 
 The channel $[A,\overline{\nu}_{\mu},B]$ such that $\overline{\mu\nu} = \overline{\mu}\ \overline{\nu}_\mu$ is the {\em stationary mean} of the AMS channel $\nu$ with respect to the AMS source distribution $\mu$.

The following lemma shows that the stationary mean of a channel is made of AMS probabilities and that the channel made by the stationary means of these AMS probabilities is stationary. The lemma gives also relationships, based on asymptotic dominance, between the corresponding hookups.

\begin{lemma}\label{LemmaChannelOfStationaryMeans}
Let $\mu$ be the distribution of an AMS source and  $[A,\nu,B]$ be a channel AMS w.r.t. $\mu$. Let $\overline{\nu}_\mu$ be the stationary mean of $\nu$ w.r.t. $\mu$.
\begin{enumerate}
	\item The probability $\overline{\nu}_\mu(x,.)$ is AMS $\overline{\mu}$-a.e.
	\item Let $S\overline{\nu}_\mu$ be the channel defined by 
		$$
		 S\overline{\nu}_\mu(x,G)=\lim_{n\to \infty} \frac{1}{n} \sum_{i=0}^{n-1} \overline{\nu}_\mu(x,T_B^{-i}G) \text{  }\overline{\mu}\text{-a.e.}
		$$
	$S\overline{\nu}_\mu$ is stationary w.r.t. $\overline{\mu}$ and
	\begin{enumerate}
			\item $\mu \nu \ll^a \overline{\mu}\ \overline{\nu}_\mu$
			\item $\mu \nu \ll^a \overline{\mu}S\overline{\nu}_\mu$
			\item $ \overline{\mu}\ \overline{\nu}_\mu \ll \overline{\mu}S\overline{\nu}_\mu$
		\end{enumerate}
\end{enumerate}
\end{lemma}
Proof of Lemma \ref{LemmaChannelOfStationaryMeans} relies on the following lemma.
\begin{lemma}\label{LemmaStationaryChannelAMSkernel}
Let $\mu$ be the distribution of a stationary source and  $[A,\nu,B]$ a channel stationary w.r.t. $\mu$. Then
\begin{enumerate}
	\item $\nu(x,.)$ is an AMS probability $\mu$-a.e.
	\item let $S\nu$ be the channel defined by $S\nu(x,G) =  \lim_{n \to \infty}\frac{1}{n}\sum_{i=0}^{n-1}\nu(x,T_B^{-i}G)$ when $\nu(x,.)$ is AMS and $S\nu(x,G) = \nu(x,G)$ else. Then $[A,S\nu,B]$ is stationary w.r.t. $\mu$.
	\item $\mu\nu \ll \mu S\nu$
\end{enumerate}
\end{lemma}
\begin{proof}
\begin{enumerate}
	\item Since $\nu$ is stationary w.r.t. $\mu$, by Lemma \ref{LemmaStationaryChannel} (p. \pageref{LemmaStationaryChannel}):
$$
\forall n\in \mathbb{N}, \frac{1}{n} \sum_{i=0}^{n-1} \nu(x,T_B^{-i}G) = \frac{1}{n} \sum_{i=0}^{n-1} \nu(T_A^{i}x, G) \text{   }\mu\text{-a.e.}
$$
$\mu$ is stationary thus, thanks to the pointwise ergodic theorem, $S\nu(x,G)= \lim_{n \to \infty}\frac{1}{n} \sum_{i=0}^{n-1} \nu(T_A^{i}x,G)$ exists and $S\nu(x,G) = S\nu(T_A x,G)$ $\mu$-a.e. 

Then $S\nu(x,G)=\lim_{n \to \infty}\frac{1}{n} \sum_{i=0}^{n-1} \nu( x,T_B^{-i}G)$ exists and, by the Vitali-Hahn-Saks theorem, is a (stationary) probability. 

	\item From 1.): 
$$
S\nu(T_A x,G)=S\nu( x,G)=S\nu( x,T_B^{-1}G) \text{  }\mu\text{-a.e.}
$$
Then, thanks to Lemma \ref{LemmaStationaryChannel} (p. \pageref{LemmaStationaryChannel}), $S\nu$ is stationary w.r.t. $\mu$
	\item Let $E \in \left( \sigma\left( \mathcal{B}_{A^\mathcal{I}} \times \mathcal{B}_{B^\mathcal{I}} \right) \right)_\infty$ such that $(\mu S\nu)_\infty (E)=0$. 

Then $S\nu(x,E_x)_\infty=0$ $\mu$-a.e. where $E_x \in (\mathcal{B}_{B^\mathcal{I}})_\infty$ is the section of $E$ at $x$. Since $\nu(x,.) \ll^a S\nu(x,.)$ and since $S\nu(x,.)$ is stationary, by Lemma \ref{LemmaDominanceOnTailSigmaField} (p. \pageref{LemmaDominanceOnTailSigmaField}), $\nu(x,.)_\infty \ll S\nu(x,.)_\infty$. Then
$$
\nu(x,E_x)_\infty=0 \text{  }\mu\text{-a.e.}
$$
This implies that $(\mu\nu)_\infty(E)=0$. Then $(\mu\nu)_\infty \ll (\mu S\nu)_\infty$. From 2.), $\mu S\nu$ is stationary then, by  Lemma \ref{LemmaDominanceOnTailSigmaField} (p. \pageref{LemmaDominanceOnTailSigmaField}), $\mu\nu \ll^a \mu S\nu$. $\mu\nu$ is stationary then $\mu\nu \ll \mu S\nu$.
\end{enumerate}
\end{proof}

\begin{proof}[Proof of Lemma \ref{LemmaChannelOfStationaryMeans}]
By definition $\overline{\nu}_\mu$ is stationary w.r.t. $\overline{\mu}$. Then, thanks to Lemma \ref{LemmaStationaryChannelAMSkernel} (p. \pageref{LemmaStationaryChannelAMSkernel}):
\begin{itemize}
	\item  $\overline{\nu}_\mu(x,.)$ is AMS $\overline{\mu}$-a.e.
	\item $S\overline{\nu}_\mu$ is stationary w.r.t. $\overline{\mu}$
	\item $ \overline{\mu}\ \overline{\nu}_\mu \ll \overline{\mu}S\overline{\nu}_\mu$
\end{itemize}
$\nu$ is AMS w.r.t. the AMS probability $\mu$ thus  $\mu \nu \ll^a \overline{\mu}\ \overline{\nu}_\mu$.

Since $\mu \nu \ll^a \overline{\mu}\ \overline{\nu}_\mu$ and $ \overline{\mu}\ \overline{\nu}_\mu \ll \overline{\mu}S\overline{\nu}_\mu$:
$$ 
\mu \nu \ll^a \overline{\mu}S\overline{\nu}_\mu
$$

\end{proof}

The following lemma  is Lemma 2 of \cite{FontanaGrayKieffer81}. It also holds for  restrictions of sources and channels to the tail $\sigma$-fields.
\begin{lemma}\label{Lemma2Fontana}
Let $\mu$ and $\eta$ be the distributions of two sources on the same alphabet $A$. Then,  for any channel $[A,\nu,B]$ 
$$
\mu \ll \eta \Rightarrow \mu\nu \ll \eta\nu
$$
\end{lemma}
The following  lemma gathers Lemma 3 and Lemma 4 of \cite{FontanaGrayKieffer81}. It states that "hookup dominance" is equivalent to "channel dominance". It also holds for  restrictions of sources and channels to the tail $\sigma$-fields.
\begin{lemma}\label{LemmaChannelDominanceEquivHookupDominance}
Let $[A,X]$ be a source with distribution $\mu$. Let  $[A,\nu,B]$ and $[A,\nu',B]$ be two arbitrary channels. Then 
$$
\nu(x,.) \ll \nu'(x,.) \text{ }\mu\text{-a.e. } \Leftrightarrow \mu\nu \ll \mu\nu'
$$
\end{lemma}

The following lemma gives a necessary and sufficient condition for a channel to be AMS w.r.t a stationary source.
\begin{lemma}\label{LemmaNecessarySufficientConditionAMSChannel}
Let $[A,X]$ be a stationary source with distribution $\mu$. A channel $[A,\nu,B]$ is AMS w.r.t. to $\mu$ if and only if there exists a  channel $\overline{\nu}_\mu$ stationary w.r.t. $\mu$ such that $\nu(x,.) \ll^a S\overline{\nu}_\mu(x,.)$ $\mu$-a.e. where $S\overline{\nu}_\mu(x,.)$ is the stationary mean of $\overline{\nu}_\mu(x,.)$. 
\end{lemma}
\begin{proof}
Assume that $\nu$ is AMS w.r.t. $\mu$. Then there exists a stationary channel $\overline{\nu}_\mu$ such that $\mu\nu \ll^a \mu\overline{\nu}_\mu$ and (thanks to Lemma \ref{LemmaChannelOfStationaryMeans}, p. \pageref{LemmaChannelOfStationaryMeans}) such that $\mu\nu \ll^a \mu S\overline{\nu}_\mu$. Then, since $\mu S\overline{\nu}_\mu$ is stationary, thanks to Lemma \ref{LemmaDominanceOnTailSigmaField} (p. \pageref{LemmaDominanceOnTailSigmaField}), $\mu_\infty\nu_\infty \ll \mu_\infty(S\overline{\nu}_\mu)_\infty$. By Lemma \ref{LemmaChannelDominanceEquivHookupDominance} (p. \pageref{LemmaChannelDominanceEquivHookupDominance}), $(\nu(x,.))_\infty \ll (S\overline{\nu}_\mu(x,.))_\infty$ $\mu$-a.e. Since $S\overline{\nu}_\mu(x,.)$ is stationary, by Lemma \ref{LemmaDominanceOnTailSigmaField} (p. \pageref{LemmaDominanceOnTailSigmaField}), $\nu(x,.) \ll^a S\overline{\nu}_\mu(x,.)$ $\mu$-a.e.

Assume that $\nu(x,.) \ll^a S\overline{\nu}_\mu(x,.)$ $\mu$-a.e.. Then,  by Lemma \ref{LemmaDominanceOnTailSigmaField} (p. \pageref{LemmaDominanceOnTailSigmaField}), $(\nu(x,.))_\infty \ll (S\overline{\nu}_\mu(x,.))_\infty$ $\mu$-a.e. By Lemma \ref{LemmaChannelDominanceEquivHookupDominance} (p. \pageref{LemmaChannelDominanceEquivHookupDominance}), $\mu_\infty\nu_\infty \ll \mu_\infty(S\overline{\nu}_\mu)_\infty$. Since $\mu S\overline{\nu}_\mu$ is stationary (thanks to Lemma \ref{LemmaStationaryChannelAMSkernel}, p. \pageref{LemmaStationaryChannelAMSkernel}), by Lemma \ref{LemmaDominanceOnTailSigmaField} (p. \pageref{LemmaDominanceOnTailSigmaField}), $\mu\nu \ll^a \mu S\overline{\nu}_\mu$.  Then $\mu\nu$ is AMS.
\end{proof}

Lemma \ref{LemmaNecessarySufficientConditionAMSChannel} can also be stated as:
\begin{lemma}\label{LemmaNecessarySufficientConditionAMSChannel2}
Let $[A,X]$ be a stationary source with distribution $\mu$. A channel $[A,\nu,B]$ is AMS w.r.t. to $\mu$ if and only if $\nu(x,.)$ is AMS $\mu$-a.e.
\end{lemma}

\subsection{Channel product}

Given two channels with the same input alphabet, it is possible to define a product of channels as follows.
\begin{definition}\label{DefinitionChannelProduct}
Let $[A,\nu_1,B]$ and $[A,\nu_2,C]$ be two channels. The {\em product} of $\nu_1$ and $\nu_2$ is the channel $[A,\nu_1\times \nu_2,B\times C]$ defined by, $\forall x \in A^\mathcal{I}$, $\forall G\in \mathcal{B}_{B^\mathcal{I}}$ and $\forall H\in \mathcal{B}_{C^\mathcal{I}}$:
$$
\nu_1\times \nu_2(x, G \times H)=\nu_1(x,G).\nu_2(x,H)
$$
\end{definition}
Remark: if $[A,X]$ is an input source to the channel product $[A,\nu_1\times \nu_2,B\times C]$ and $[B \times C , (Y,Z)]$ the corresponding output source, then $Y \to X \to Z$ is a Markov chain.

The following lemma states that the product of  channel restrictions to the tail $\sigma$-fields is the restriction of the channel product.

\begin{lemma}\label{LemmaProductOfTailIsTailOfProduct}
Let $[A,\nu_1,B]$ and $[A,\nu_2,C]$ be two channels. Then, for any $x\in A^\mathcal{I}$, $(\nu_1\times \nu_2(x,.))_\infty = (\nu_1(x,.))_\infty \times (\nu_2(x,.))_\infty$ i.e. the restriction of $[A,\nu_1\times \nu_2,B\times C]$ to the tail $\sigma$-field $\left( \sigma\left( \mathcal{B}_{B^\mathcal{I}} \times \mathcal{B}_{C^\mathcal{I}} \right) \right)_\infty$ is the probability kernel  $[A,(\nu_1)_\infty \times (\nu_2)_\infty,B\times C]$:
$$
(\nu_1\times \nu_2)_\infty = (\nu_1)_\infty \times (\nu_2)_\infty
$$
\end{lemma}
\begin{proof}
For any $x\in A^\mathcal{I}$, $(\nu_1\times \nu_2(x,.))_\infty$ and $(\nu_1(x,.))_\infty \times (\nu_2(x,.))_\infty$ coincide on the semi-algebra of rectangles $G\times H \in \left( \mathcal{B}_{B^\mathcal{I}}\right)_\infty \times \left( \mathcal{B}_{C^\mathcal{I}} \right)_\infty$. Then they coincide on $\sigma \left( \left( \mathcal{B}_{B^\mathcal{I}}\right)_\infty \times \left( \mathcal{B}_{C^\mathcal{I}} \right)_\infty \right)$. 

Thanks to Lemma \ref{LemmaChannelRestriction} (p. \pageref{LemmaChannelRestriction}), $\sigma \left( \left( \mathcal{B}_{B^\mathcal{I}}\right)_\infty \times \left( \mathcal{B}_{C^\mathcal{I}} \right)_\infty \right) = \left( \sigma \left( \mathcal{B}_{B^\mathcal{I}} \times  \mathcal{B}_{C^\mathcal{I}} \right) \right)_\infty $
\end{proof}

\pagebreak 

\part{Capacity of a noisy computation}\label{PartCapacity}

\pagebreak

\section{Model for noisy computations \label{SectionModelForNoisyComputations}}

The model of noisy computation proposed in this section is not linked to a noisy version of a peculiar formal computational model (e.g., boolean circuits made of noisy boolean gates, \cite{VonNeumann56}). Functions considered here are "general" measurable functions from a sequence space to a sequence space and their noisy realizations are viewed as random channels. Section \ref{SectionNoisyComputationsAndComputableFunctions} is devoted to the question of the instantiation of the noisy computation model to formal computational models and to their noisy versions.

The  input  and output sequences of a computation will be modeled as random processes with values in sequence spaces built from measurable alphabets. Since the purpose is to handle computations which could be processes evolving from an initial state,  the time shifts will not be assumed invertible. The processes considered are thus one-sided. The alphabets $(A,\mathcal{B}_A)$, $(B,\mathcal{B}_B)$ and $(C,\mathcal{B}_C)$, considered below, are assumed to be standard measurable spaces.
Let $(A^\mathcal{I}, \mathcal{B}_{A^\mathcal{I}})$, $(B^\mathcal{I}, \mathcal{B}_{B^\mathcal{I}})$ and $(C^\mathcal{I}, \mathcal{B}_{C^\mathcal{I}})$ be the corresponding  measurable sequence spaces. Let $[A,X]$ be a source. 

A noisy computing device is modeled as a noisy channel $[A,\nu,C]$: given an input sequence $x$, the output sequence $z$ is not deterministically determined but, due to possible errors, randomly produced. The noisy computing device  takes $X$ as input and produces as an output a random process $Z$ on $(C^\mathcal{I}, \mathcal{B}_{C^\mathcal{I}})$. The hookup $P_{XZ}= P_X\nu$ represents the actual noisy calculation of the device operating on a given input flow represented by $X$.

 In the sequel, the noisy channel/computing device will  be denoted by $F$ and, with an abuse of notation, the output $Z$ by $F(X)$. Thus $\nu(x,.)$ will be denoted by $P_{F(X)|X}(.|x)$.

The perfect computation is represented by a $(\mathcal{B}_{A^\mathcal{I}},\mathcal{B}_{B^\mathcal{I}})$-measurable function  $f :  A^\mathcal{I}   \rightarrow  B^\mathcal{I}$ (the expected function).  $f$ defines a deterministic channel  $\left\{   P_{f(X)|X}(. |x), x\in A^\mathcal{I}\right\}$ such that, $\forall G \in \mathcal{B}_{B^\mathcal{I}}$,  $P_{f(X)|X}(G|x) = 1_{f^{-1}(G)}(x)$ $P_X$-a.e. (\cite{Gray11}) where $1_E$ denotes the characteristic function of a set $E$ .

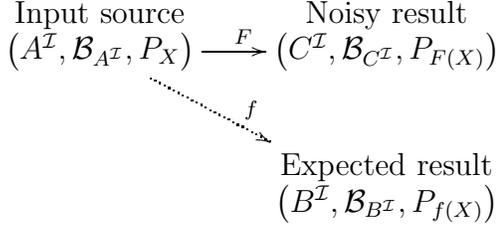
\begin{figure}[htbp]
	\begin{center}
	\xymatrix{
    		\overset{\txt{Input source }}{ \left(  A^{\mathcal{I}} , \mathcal{B}_{A^\mathcal{I}}, P_X  \right) }	\ar  [r] ^F 	\ar@{.>} [dr] ^f &   \overset{\txt{Noisy result  }}{\left( C^{\mathcal{I}} , \mathcal{B}_{C^\mathcal{I}}, P_{F(X)} \right)} \\
																		&   \overset{\txt{Expected result}}{\left( B^{\mathcal{I}} , \mathcal{B}_{B^\mathcal{I}}, P_{f(X)} \right)} 
 	 }
	\end{center}
	\caption{Model for Noisy Computation} \label{FigNoisyComputationModel}
\end{figure}

Since $f$ is a function, $F(X) \rightarrow X \rightarrow f(X)$ is a Markov chain. Then the measurable  functions $x \mapsto P_{f(X)F(X)}(E|x)$, $E  \in \mathcal{B}_{B^\mathcal{I}\times C^\mathcal{I}}$, are such that, $\forall G  \in \mathcal{B}_{B^\mathcal{I}}$ and $\forall H  \in \mathcal{B}_{C^\mathcal{I}}$:
$$
P_{f(X)F(X)|X}(G\times H|x)= P_{f(X)|X}(G|x).P_{F(X)|X}(H|x)
$$
The collection of probabilities $\left\{ P_{f(X)F(X)}(.|x), x \in A^\mathcal{I}\right\}$ defines a channel product which will be denoted $[A,f\times F,B\times C]$ and called the noisy computation of $f$ by $F$.  

\begin{definition}\label{DefinitionNoisyComputation}
 Let $\mathcal{I}$ be a countable index set. Let $(A^\mathcal{I}, \mathcal{B}_{A^\mathcal{I}})$, $(B^\mathcal{I}, \mathcal{B}_{B^\mathcal{I}})$ and $(C^\mathcal{I}, \mathcal{B}_{C^\mathcal{I}})$  be standard sequence spaces. Let $f: A^\mathcal{I} \rightarrow B^\mathcal{I}$ be a $(\mathcal{B}_{A^\mathcal{I}},\mathcal{B}_{B^\mathcal{I}})$-measurable function and $[A,F,C]$ a channel. The channel product $[A,f\times F,B\times C]$ is called the {\em noisy computation} of $f$ by $F$. $[[A,X];[A,f\times F,B\times C]]$ is called the {\em  hookup} of the source $[A,X]$ and of the noisy computation $[A,f\times F,B\times C]$. It determines the process $(X,f(X),F(X))$.
\end{definition}

Asymptotic mean stationarity (or stationarity) and ergodicity of processes are key properties for the Asymptotic Equipartition Property (AEP) to hold  (\cite{Girardin05}). Thus coding theorems generally rely on these properties. Asymptotic mean stationarity and ergodicity will be also fundamental properties for noisy computations.

\begin{definition}
A noisy computation $[A,f\times F,B\times C]$ is stationary (respectively AMS) if, for any stationary (respectively AMS) source $[A,X]$, the processes $(X,f(X))$, $(X,F(X))$ and $(f(X),F(X))$ are stationary (respectively AMS).
\end{definition}

\begin{definition}
A stationary (resp. AMS) noisy computation $[A,f\times F,B\times C]$ is ergodic if, for any stationary (resp. AMS) and ergodic source $[A,X]$, the processes $(X,f(X))$, $(X,F(X))$ and $(f(X),F(X))$ are ergodic.
\end{definition}

The following lemma is a straightforward application of the conditional AEP to the special case where $Y=f(X)$.
\begin{lemma}\label{LemmaConditionAEPForFunctions}
Let $A$ and $B$ be two finite alphabets. Let $f: A^\mathcal{I} \rightarrow B^\mathcal{I}$ be a $(\mathcal{B}_{A^\mathcal{I}},\mathcal{B}_{B^\mathcal{I}})$-measurable function such that $[A,f,B]$ is an ergodic AMS channel. Let $[A,X]$  be an ergodic AMS input source and $[B,Y]$ the corresponding output source.  For $\epsilon>0$ and $y\in B^\mathcal{I}$, let $A^n(\epsilon,y^n)$ be the set of $\epsilon$-typical sequences given $y^n$ such that: 
$$
A^n(\epsilon,y^n)= \left\{ x^n \in (f^{-1}(y))^n / \left|  \frac{-log(P_{X^n|Y^n}(x^n|y^n))}{n} - \overline{H}(X^n|Y^n) \right| < \epsilon \right\}
$$
where $(f^{-1}(y))^n$ is the set of prefixes (of length $n$) of inverse images of $y$ by $f$.

Then:
\begin{enumerate}
	\item $\lim_{n\rightarrow \infty} P_{X^n|Y^n}(A^n(\epsilon,y^n)|y^n) = 1\text{   }P_Y -a.e.$
	\item for any $\delta>0$ and any $\epsilon >0$, for $n$ large enough:
	$$
	(1-\delta) e^{n(\overline{H}(X|Y)-\epsilon)} \leq |A^n(\epsilon,y^n)| \leq e^{n(\overline{H}(X|Y)+\epsilon)}  \text{  }P^n_Y-a.e.
	$$
\end{enumerate}
\end{lemma}
\begin{proof}
The lemma follows immediately from the conditional AEP and from the fact that
$$
\lim_{n\to \infty} P_{X^n|Y^n}(\{ x^n \in  (f^{-1}(y))^n\}|y^n)=1
$$
\end{proof}

It should be noticed that  $\overline{H}(f(X)|X)=0$ and then
$$
\overline{H}(X)=\overline{H}(X|f(X))+\overline{H}(f(X))
$$
and
$$
\overline{I}(X,f(X)) = \overline{H}(f(X))
$$

Let $f$ be a measurable function  such that $f^n : x^n \mapsto (f(x))^n$ is a $(\mathcal{B}_{A^n},\mathcal{B}_{B^n} )$-measurable function for any large $n$. This means that $\forall x, x' \in A^\mathcal{I}$ such that $x^n=x'^n$,  $(f(x))^n=(f(x'))^n$. In other words $(f^{-1}(y))^n=(f^n)^{-1}(y^n)$, for large $n$. Such functions will be called weakly causal functions.

\begin{definition}
A $(\mathcal{B}_{A^\mathcal{I}},\mathcal{B}_{B^\mathcal{I}})$-measurable function $f:A^\mathcal{I} \rightarrow B^\mathcal{I}$ is weakly causal if there exists $N_f \in \mathbb{N}$ such that $\forall n\geq N_f$, $\forall x, x' \in A^\mathcal{I}$ such that $x^n=x'^n$, then $(f(x))^n=(f(x'))^n$ and $f^n : x^n \mapsto (f(x))^n$ is a $(\mathcal{B}_{A^n},\mathcal{B}_{B^n} )$-measurable function. 
\end{definition}

\pagebreak 

\section{Typical input Rate of a Hookup and Typical Input Capacity of a Noisy Computation \label{SectionTypicalInputRateTypicalInputCapacity}}

In this section, the definition of the typical  input rate of a hookup and the definition of the typical input capacity of a noisy computation are given. In a next section, this rate will be shown to characterize the rate at which a source should produce typical inputs for a random channel in order to allow decodability of the output. Decodability means ability to recover the desired function results by decoding. Assumptions and notations used throughout this section are:
 \begin{itemize}
	\item the alphabets $A$, $B$ and $C$ are assumed finite, the $\sigma$-fields $\mathcal{B}_A$, $\mathcal{B}_B$  and $\mathcal{B}_C$ are the sets of subsets of the alphabets $A$, $B$ and $C$
	\item the noisy computation $[A,f \times F,B\times C]$ and the source $[A,X]$ are assumed AMS and ergodic
	\item the function $f$ is assumed weakly causal.
\end{itemize}

Defining a notion of capacity for noisy computations, following the classical approach for communication channels,  boils down to characterizing the maximum number of {\em input} $n$-sequences that can be selected in order to allow an asymptotically perfect correction/decoding process. The definition of capacity proposed below can be informally justified as follows.  

Assume that only sequences $x^n$ conditionally typical given $y^n$, typical sequence of $f^n(X^n)$, are allowed as inputs ($n$ large). Assume also that only typical sequences $z^n$ of $F^n(X^n)$ which are conditionally typical given any allowed input $x^n$ ($x^n$ conditionally typical given a typical $y^n$) are retained for decoding, the others are simply ignored. Then such a $z^n$ is conditionally typical given $y^n$. 

Hence, the set of possible output values given $x^n\in (f^n)^{-1}(y^n)$ submitted to decoding has ``almost''  $e^{n\overline{H}(F(X)|f(X))}$ elements (thanks to the conditional AEP). Thus the number of non-overlapping such sets built from the set of  $e^{n\overline{H}(F(X))}$ typical sequences $z^n$ is at most $e^{n\overline{H}(F(X))}/e^{n(\overline{H}(F(X)|f(X)))}$. 

Each inverse image $(f^n)^{-1}(y^n)$ offers $e^{n\overline{H}(X|f(X))}$  different possible input $n$-sequences. So among the $e^{n\overline{H}(X)}$ possible input $n$ sequences,  at most $e^{n(\overline{H}(X|f(X))+\overline{H}(F(X))-\overline{H}(F(X)|f(X)))}$ of them can be used in order to reach unambiguous decodability. 

Since $\overline{H}(F(X))-\overline{H}(F(X)|f(X)))=\overline{I}(F(X),f(X))$, this leads to define  the capacity of the noisy computation $[A,f \times F,B\times C]$ as:
$$
C_f(F)  =  \sup_X [ \overline{H}(X|f(X)) + \overline{I}(F(X),f(X))] \nonumber
$$
$\overline{I}(F(X),f(X))=\overline{H}(f(X))-\overline{H}(f(X)|F(X)))$ and since $f$ is a function, $\overline{H}(X|f(X))+\overline{H}(f(X))=\overline{H}(X)$. Then:
$$
C_f(F) = \sup_X \left( \overline{H}(X) - \overline{H}(f(X)|F(X))\right)
$$

\begin{definition} [Typical input rate]
Let $[A,f \times F,B\times C]$ be an AMS and ergodic noisy computation on finite alphabets and let $[A,X]$ be an AMS and ergodic source. Assume that $f$ is a $(\mathcal{B}_{A^\mathcal{I}},\mathcal{B}_{B^\mathcal{I}} )$-measurable weakly causal function. $\forall n \geq N_f$, the {\em $n^{th}$ order typical input rate} of the hookup $[[A,X] ; [A, f \times F , B\times C] ]$ is:
$$
B^{n}(X^n,f^n,F^n)=\frac{1}{n}\left[ H(X^n) - H(f^n(X^n)|F^n(X^n))\right]
$$
The typical input rate of the hookup $[[A,X] ; [A, f \times F , B\times C]]$ is:
$$
\overline{B}(X,f,F)=\overline{H}(X) - \overline{H}(f(X)|F(X))
$$
\end{definition}

Let $\mathfrak{P}_{ae}(A^\mathcal{I})$ denote the set of ergodic and AMS probability measures on $(A^\mathcal{I}, \mathcal{B}_{A^\mathcal{I}})$. 

\begin{definition} [Typical input capacity]
Let $[A,f \times F,B\times C]$ be an AMS and ergodic noisy computation on finite alphabets. Assume that $f$ is a $(\mathcal{B}_{A^\mathcal{I}},\mathcal{B}_{B^\mathcal{I}} )$-measurable weakly causal function. The typical input capacity (for AMS and ergodic sources) of the noisy computation $[A,f \times F,B\times C]$ is 
$$
C_f(F) = \sup_{P_X\in \mathfrak{P}_{ae}(A^\mathcal{I}) } \overline{B}(X,f,F)
$$
\end{definition}

The equivalent expression 
$$
C_f(F) =\sup_{P_X\in \mathfrak{P}_{ae}(A^\mathcal{I}) }\left( \overline{H}(X) - \overline{H}(f(X)|F(X)) \right)
$$
 immediately shows that this capacity boils down to the usual channel capacity when $f$ is a bijection (in which case $\overline{H}(f(X)|F(X))=\overline{H}(X|F(X))$. 

The adjective typical is motivated by the forthcoming joint source-channel coding theorem: the encoding process will be constrained to encode typical sequences of a source by typical sequences of a source ``achieving'' capacity. This constraint will be discussed in section~\ref{SectionJointSourceComputationCodingTheorem}.

$f$ being a function, $\overline{H}(f(X)|F(X)) \leq \overline{H}(X|F(X))$ thus $C_f(F)$ is greater than the {\em channel capacity} of the noisy function $F$ (considered as a {\em communication channel}). This is totally consistent with the intuitive interpretation of equivocation: as computation ``burns'' information (i.e. $\overline{H}(f(X))\leq \overline{H}(X)$), the amount of information that  must be added to a noisy result to retrieve the input of the function is likely to be larger than the amount of information that must added to a noisy result to obtain the correct one.

From the definition of $C_f(F)$, the limit cases can be briefly discussed:
\begin{itemize}
	\item if the computation by $F$ is perfect almost surely (i.e. there is a deterministic way of retrieving $f(x)$ from $F(x)$ almost everywhere) then $\overline{H}(f(X)|F(X))=0$  and thus $C_f(F)=\sup_{P_X\in \mathfrak{P}_{ae}(A^\mathcal{I}) } \overline{H}(X)$. In other words, almost any possible input sequence allows  decodability
	\item if the computation is totally noisy meaning that the noisy output $F(X)$ and the expected output $f(X)$ are independent for any $X$, then $\overline{H}(f(X)|F(X))=\overline{H}(f(X))$ and $C_f(F) = \sup_{P_X\in \mathfrak{P}_{ae}(A^\mathcal{I}) } \overline{H}(X|f(X)) $ ; the joint source-noisy computation coding theorem will show that this means that to allow decodability it is necessary to restrict inputs to sequences belonging to a unique $f^{-1}(y_0)$ ;  only one output sequence $y_0$  can be reliably computed i.e. only a constant function can be reliably computed with such a totally noisy apparatus
	\item if $C_f(F)=0$ then, for any $X$, $\overline{H}(X|f(X)) + \overline{I}(f(X) ; F(X)) =0$ ; this implies that both $\overline{H}(X|f(X))=0$ and $\overline{I}(f(X) ; F(X)) =0$ ; the latter condition is equivalent to the preceding case (totally noisy computation) and the former condition means that $f$ is injective almost everywhere.
\end{itemize}

\pagebreak 

\section{Extension of Feinstein Theorem to Noisy Computations\label{SectionFeinsteinTheorem}}

In this section,  extensions of Feinstein lemma and theorem covering noisy computations are given.  Based on the model given above, the existence of Feinstein  codes for noisy computations is derived from classical results dealing with communication channels. While a Feinstein code for a channel is a set of pairs $(x_i,\Gamma_i)$ where $x_i$ is an input sequence and $\Gamma_i$ a set of output sequences (the decodability region), a Feinstein code for a noisy computation will be defined as a set of pairs $(A_i,\Gamma_i)$ where $A_i$ is a set of input sequences and $\Gamma_i$ the associated decodability region.

\subsection{Feinstein lemma for noisy computations}

If $B$ is countable, for any $y\in B^{\mathcal{I}}$, the set $\left\{y \right \}$ is measurable (countable intersection of cylinders $c(y_i), i=1,\cdots$ if $y=y_1y_2\ldots$) and  $f^{-1}(y)$ is measurable as $f$ is measurable. This justifies the assumptions on the alphabets made by the following lemmas and corollaries. 

Feinstein lemma for noisy computation will be stated in a form consistent with the one of Feinstein lemma for communication channels given in  \cite{Gray11} (Lemma 14.1) and assumptions close to \cite{Gray11} : alphabets $A$ and $C$ are standard, $B$ is countable and standard.

In the framework of the noisy computation model, from now on, the joint probabilities $P_{XF(X)}$ and $P_{f(X)F(X)}$ are assumed to be dominated by (i.e., absolutely continuous w.r.t) the corresponding product probabilities: $P_{XF(X)} \ll P_{X} \times P_{F(X)}$ and $P_{f(X)F(X)}\ll P_{f(X)} \times P_{F(X)}$.

\begin{lemma} [Feinstein lemma for noisy computation] \label{LemmaFeinsteinLemmaForNoisyComputation}
Let $A$ and $C$ be standard alphabets and $B$ be a countable standard alphabet. Let $[A,f \times F,B\times C]$ be a noisy computation. Let $[A,X]$ be a source. Assume that $P_{f(X)F(X)} \ll P_{f(X)}\times P_{F(X)}$. Let $\psi_{f(X)F(X)}$ be the Radon-Nikodym derivative $\frac{dP_{f(X)F(X)}}{d(P_{f(X)}\times P_{F(X)})} $ and $i_{f(X)F(X)}=ln(\psi_{f(X)F(X)})$. Then $\forall M \in \mathbb{N}$,  $\forall a>0$ such that 
$$
M e^{-a} + P_{f(X)F(X)}(i_{f(X)F(X)}\leq a) < \frac{1}{4}
$$
 there exist a set $\left\{ y_k \in B^\mathcal{I} , k=1\cdots M \right\}$, a collection of measurable disjoint sets $\Gamma_k, k=1,...,M$ members of $\mathcal{B}_{C^{\mathcal{I}}}$ and a collection of measurable sets $A_k \subset f^{-1}(y_k), k=1,\cdots, M$ members of $\mathcal{B}_{A^{\mathcal{I}}}$ such that:
$$
\forall k=1, \cdots, M \text{  } P_{X|f(X)}(A_k|y_k)>1-\lambda 
$$
and
$$
\forall x\in A_k \text{  }  P_{F(X)|X}\left( \Gamma_k^c | x  \right) \leq \epsilon 
$$
for any $\epsilon \in ]0,\frac{1}{2}[$ and $\lambda \in ]0,\frac{1}{2}[$ such that $\epsilon\lambda = M e^{-a} + P_{f(X)F(X)}(i_{f(X)F(X)}\leq a)$
\end{lemma}

Before proving the extension of  Feinstein lemma for noisy computation,  the following lemma is needed. Its proof is given in \ref{SectionProofsOfLemmas}:

\begin{lemma}\label{LemmaMaximalFeinsteinFamily}
Let $A$ and $C$ be standard alphabets and $B$ be a countable standard alphabet. Let $[A,f \times F,B\times C]$ be a noisy computation. Let $[A,X]$ be a source. Let $\epsilon>0$ and $\lambda>0$. Let $\tilde{B}$ be a measurable subset of $B^{\mathcal{I}}$. Assume there exist a set $\left\{ y_k \in \tilde{B} , k=1\cdots M \right\}$, a collection of measurable disjoint sets $\Gamma_k, k=1,...,M$ members of $\mathcal{B}_{C^{\mathcal{I}}}$ and a collection of measurable sets $A_k \subset f^{-1}(y_k), k=1,\cdots, M$ members of $\mathcal{B}_{A^{\mathcal{I}}}$ such that:
$$
\forall k=1, \cdots, M \text{  } P_{X|f(X)}(A_k|y_k)>1-\lambda \text{  and  } \forall x\in A_k \text{  }  P_{F(X)|X}\left( \Gamma_k^c | x  \right) \leq \epsilon 
$$
If the set $\left\{ y_k \in \tilde{B} , k=1\cdots M \right\}$ is maximal, meaning there do not  exist any other $y_{M+1}\in \tilde{B}$, $\Gamma_{M+1}$ and $A_{M+1}$ complying with the given properties, then:
\begin{enumerate}
	\item $\forall y_0\in \tilde{B}\setminus \{y_1,\cdots,y_M\}$, $ P_{F(X)|f(X)}\left( \bigcup_{k=1}^M \Gamma_k | y_0  \right)  > \lambda\epsilon$ \label{Statement1LemmaMaximalFeinsteinFamily}
	\item $P_{F(X)}\left( \bigcup_{k=1}^M \Gamma_k  \right)  > min( (1-\lambda)(1-\epsilon) , \lambda\epsilon).P_{f(X)}(\tilde{B})$ \label{Statement2LemmaMaximalFeinsteinFamily}
	\item $\tilde{B}\neq \emptyset \Rightarrow M\geq 1$ \label{Statement3LemmaMaximalFeinsteinFamily}
\end{enumerate}
\end{lemma}

\begin{proof}[Proof of Lemma~\ref{LemmaFeinsteinLemmaForNoisyComputation}]
This proof is derived from the proof of  Feinstein lemma given in \cite{Gray11} (Lemma 14.1).

Let $a>0$, $M$ a positive integer, $\epsilon \in ]0,\frac{1}{2}[$ and $\lambda=\in ]0,\frac{1}{2}[$ be such that 
$$
\epsilon\lambda = M e^{-a} + P_{f(X)F(X)}(i_{f(X)F(X)}\leq a)
$$

Let $G=\{(y,z)\in B^\mathcal{I}\times C^\mathcal{I} / i_{f(X)F(X)}(y,z) > a \}$.  $P_{f(X)F(X)}(G^c)\leq\epsilon\lambda <1$ then  $P_{f(X)F(X)}(G) \geq 1-\epsilon\lambda>0$.

Pick up a $y_1\in B^\mathcal{I}$ and determine a $\mathcal{B_{A^\mathcal{I}}}$-measurable $A_1 \subset f^{-1}(y_1)$ and a $\mathcal{B_{C^\mathcal{I}}}$-measurable $\Gamma_1$ such that:
$$
 P_{X|f(X)}(A_1|y_1)>1-\lambda \text{  and  }\forall x\in A_1 \text{  }  P_{F(X)|X}\left( \Gamma_1^c | x  \right) \leq \epsilon 
$$
Such a  $(y_1,A_1,\Gamma_1)$ exists, for instance $(y_1, f^{-1}(y_1), C^\mathcal{I})$. Recursively, pick up (if any) a $y_i$ and determine a $\mathcal{B_{A^\mathcal{I}}}$-measurable $A_i \subset f^{-1}(y_i)$ and a $\mathcal{B_{C^\mathcal{I}}}$-measurable $\Gamma_i$ such that:
\begin{eqnarray}
P_{X|f(X)}(A_i|y_i) & > & 1-\lambda \nonumber \\
\text{and  }\forall x\in A_i \text{  }  P_{F(X)|X}\left( \Gamma_i^c | x  \right) & \leq&  \epsilon \nonumber\\
\text{and  }\forall j<i,  \Gamma_i \bigcap \Gamma_j & = & \emptyset  \nonumber
\end{eqnarray}
until $B^\mathcal{I}$ is exhausted. Let $n$ be the largest possible number of selected $y_i$ and assume that $n<M$. 

The  selected set $\left\{ (y_k, A_k,\Gamma_k), k=1,\cdots,n \right \}$ is maximal. By Lemma~\ref{LemmaMaximalFeinsteinFamily}:
\begin{eqnarray}
\forall y \notin\{y_1,\cdots,y_n\},  P_{F(X)|f(X)}\left( \bigcup_{k=1}^M \Gamma_k  | y \right) &  > & \lambda\epsilon \nonumber \\
\Rightarrow P_{F(X)|f(X)}\left( (\bigcup_{k=1}^M \Gamma_k)^c  | y \right)   <   1 - \lambda\epsilon & &  \label{EqUpperBoundFLemmaNoisyComputation}
\end{eqnarray}

$f(X)\rightarrow X\rightarrow F(X)$ is a Markov Chain, then $\forall k=1,\cdots,n$:
\begin{eqnarray}
P_{F(X)|f(X)}(\Gamma_k  | y_k ) & = & \int P_{F(X)|X}(\Gamma_k  | x) dP_{X|f(X)}(x|y_k) \nonumber \\
				     & \geq  &  \int_{A_k} P_{F(X)|X}(\Gamma_k  | x) dP_{X|f(X)}(x|y_k)  >  (1-\lambda)(1-\epsilon) \nonumber
\end{eqnarray}
which implies $P_{F(X)|f(X)}\left(  (\bigcup_{k=1}^M \Gamma_k)  | y_k \right)> (1-\lambda)(1-\epsilon)$ and thus 
\begin{equation}
P_{F(X)|f(X)}\left(  (\bigcup_{k=1}^M \Gamma_k)^c  | y_k \right) \leq 1-(1-\lambda)(1-\epsilon) \label{EqLowerBoundFLemmaNoisyComputation}
\end{equation}

$P_{f(X)F(X)}(G)=\Delta + \Theta$ where $\Delta=P_{f(X)F(X)}\left(G \bigcap (B^\mathcal{I} \times \bigcup_{k=1}^M \Gamma_k) \right)$ and $\Theta= P_{f(X)F(X)}\left(G \bigcap (B^\mathcal{I} \times \bigcup_{k=1}^M \Gamma_k)^c \right)$.

It holds that:
$$
\Theta  \leq  P_{f(X)F(X)}\left(B^\mathcal{I} \times \bigcup_{k=1}^M \Gamma_k)^c \right) = \int P_{F(X)|f(X)}\left(  (\bigcup_{k=1}^M \Gamma_k)^c  | y \right) dP_{f(X)}
$$
Thus
\begin{multline}
\Theta	 \leq  \int_{\{y_1,\cdots,y_n \}} P_{F(X)|f(X)}\left(  (\bigcup_{k=1}^M \Gamma_k)^c  | y \right) dP_{f(X)} +  \nonumber  \\
						 					\int_{\{y_1,\cdots,y_n \}^c} P_{F(X)|f(X)}\left(  (\bigcup_{k=1}^M \Gamma_k)^c  | y \right) dP_{f(X)} \nonumber
\end{multline}
By (\ref{EqUpperBoundFLemmaNoisyComputation}),  (\ref{EqLowerBoundFLemmaNoisyComputation}) and the fact that $\epsilon \in ]0,\frac{1}{2}[$ and $\lambda \in ]0,\frac{1}{2}[$: 
$$
\Theta	 < \int_{\{y_1,\cdots,y_n \}} (1-(1-\epsilon)(1-\lambda)) dP_{f(X)} + \int_{\{y_1,\cdots,y_n \}^c} (1-\epsilon\lambda) dP_{f(X)}
$$
\begin{equation}
\Rightarrow \Theta <  1-\epsilon\lambda \label{EqThetaUpperBound}
\end{equation}
$$
\Delta = P_{f(X)F(X)}\left(G \bigcap (B^\mathcal{I} \times \bigcup_{k=1}^M \Gamma_k) \right) \leq P_{f(X)F(X)}(G) =  \int_G dP_{f(X)F(X)}
$$
Since, on G, $i_{f(X)F(X)}(y,z)>a$ or equivalently $\frac{\psi_{f(X)F(X)}}{e^a}>1$:
$$
\Delta	 \leq  \int_G \frac{\psi_{f(X)F(X)}(y,z)}{e^a} dP_{f(X)F(X)}  \leq  e^{-a} \int \psi_{f(X)F(X)} dP_{f(X)F(X)}
$$
But $\Delta = \int \Delta d(P_{f(X)}\times P_{F(X)})$
thus, since $\psi_{f(X)F(X)} = \frac{dP_{f(X)F(X)}}{d(P_{f(X)}\times P_{F(X)})} $:
\begin{multline}
\Delta  \leq  e^{-a} \int \int \psi_{f(X)F(X)} dP_{f(X)F(X)}d(P_{f(X)}\times P_{F(X)}) \nonumber \\
						 \leq  e^{-a} \int \int  dP_{f(X)F(X)} dP_{f(X)F(X)} =  e^{-a}  \nonumber
\end{multline}

\begin{equation}
\Rightarrow \Delta \leq  n.e^{-a} \label{EqDeltaUpperBound}
\end{equation}

From (\ref{EqThetaUpperBound}) and (\ref{EqDeltaUpperBound}), $P_{f(X)F(X)}(G)=\Delta+\Theta < 1-\epsilon\lambda + ne^{-a}$. But $\epsilon\lambda = Me^{-a}+P_{f(X)F(X)}(G^c)=Me^{-a}+1-P_{f(X)F(X)}(G)$ thus $P_{f(X)F(X)}(G)=1-\epsilon\lambda + Me^{-a}$. This leads to $n>M$ contradicting the assumption $n<M$. Then $n = M$.
\end{proof}

{\bf Remark:} since $Me^{-a}<1$, $M$ is necessarily strictly lower than $e^{a}$.

\subsection{Decomposable modules}\label{SubsectionDecomposableModules}

In the peculiar case where $X\rightarrow f(X) \rightarrow F(X)$ is a Markov chain (decomposable modules defined in \cite{WinogradCowan63} fulfill this assumption), a straightforward lemma covering noisy computations can be directly derived from the generalized Feinstein lemma, \cite{Gray11} (Lemma 14.1, page 362) or \cite{Kadota70}. This assumption allows to decompose the channel $[A,F,C]$ as a cascade of the channels $[A,f,B]$ and $[B,f_X^{-1}F,C]$ where $[B,f_X^{-1},A]$ is the reverse channel linking $f(X)$ to $X$ for a fixed source $[A,X]$.

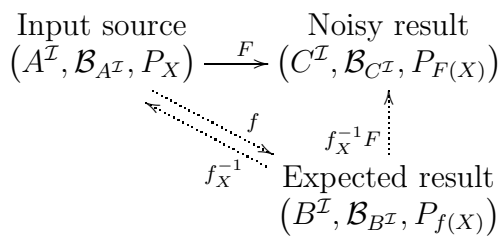
\begin{figure}[htbp]
	\begin{center}
	\xymatrix{
    		\overset{\txt{Input source }}{ \left(  A^{\mathcal{I}} , \mathcal{B}_{A^\mathcal{I}}, P_X  \right) }	\ar  [r] ^F 	\ar@{.>} [dr] ^f &   \overset{\txt{Noisy result  }}{\left( C^{\mathcal{I}} , \mathcal{B}_{C^\mathcal{I}}, P_{F(X)} \right)} \\
																		&   \overset{\txt{Expected result}}{\left( B^{\mathcal{I}} , \mathcal{B}_{B^\mathcal{I}}, P_{f(X)} \right)} \ar@{.>} [u]^{f_X^{-1}F} \ar@<1ex>@{.>} [ul] ^{f_X^{-1}}
 	 }
	\end{center}
	\caption{Model for Winograd and Cowan 's Decomposable Module} \label{FigMarkovNoisyComputationModel}
\end{figure}

\begin{lemma} \label{LemmaFeinsteinLemmaForNoisyComputationMarkovChain}
Let $A$ and $C$ be standard alphabets and  $B$ be a countable standard alphabet. Let $[A,f \times F,B\times C]$ be a noisy computation and $[A,X]$ a source such that $X\rightarrow f(X) \rightarrow F(X)$ is a Markov Chain. Let $\psi_{f(X)F(X)}$ be the Radon-Nikodym derivative $\frac{dP_{f(X)F(X)}}{d(P_{f(X)}\times P_{F(X)})} $ and $i_{f(X)F(X)}=ln(\psi_{f(X)F(X)})$. Then $\forall M \in \mathbb{N}$,  $\forall a>0$, there exist a set $\left\{ y_k \in B^{\mathcal{I}} , k=1\cdots M \right\}$ and a collection of measurable disjoint sets $\Gamma_k, k=1,...,M$ members of $\mathcal{B}_{C^{\mathcal{I}}}$ such that $\forall k=1, \cdots, M$, $\forall x\in f^{-1}(y_k)$:
$$
P_{F(X)|X}\left( \Gamma_k^c | x  \right) \leq M e^{-a} + P_{f(X)F(X)}(i_{f(X)F(X)}\leq a)
$$ 
\end{lemma}

\begin{proof}
$f$ is a function then $ f(X) \rightarrow X \rightarrow F(X)$ is a Markov chain so the channel $f(X)\rightarrow F(X)$ is  a cascade of the channels $f(X)\rightarrow X$, denoted $[B,f_X^{-1},A]$, and $X\rightarrow F(X)$ i.e. $[A,F,C]$. This cascade is the channel $[B,f_X^{-1}F, C]$.

Let $M \in \mathbb{N}$ and $a>0$. By Feinstein lemma (\cite{Kadota70}) applied to the channel $[B,f_X^{-1}F, C]$, there exist  a set $\left\{ y_k \in B^{\mathcal{I}} , k=1\cdots M \right\}$ and a collection of measurable disjoint sets $\Gamma_k, k=1,...,M$ members of $\mathcal{B}_{C^{\mathcal{I}}}$ such that, $\forall k=1, \cdots, M$:
$$
 P_{F(X)|f(X)}\left( \Gamma_k^c | y_k  \right) \leq M e^{-a} + P_{f(X)F(X)}(i_{f(X)F(X)}\leq a)
$$ 
By assumption $X\rightarrow f(X) \rightarrow F(X)$ is also a Markov chain so the channel $F$ is also a cascade of the channels ($[A,f,B]$ and $[B,f_X^{-1}F,C]$). Then:
$$
\forall H\in \mathcal{B}_{C^{\mathcal{I}}} \text{  } P_{F(X)|X}(H|x) = \int_{B^{\mathcal{I}}} P_{F(X)|f(X)}(H|y) dP_{f(X)|X}(y|x)
$$
Thus, for any $k=1,\cdots,M$:
$$
 P_{F(X)|X}(\Gamma_k^{c}|x) = \int_{B^{\mathcal{I}}} P_{F(X)|f(X)}(\Gamma_k^{c}|y) dP_{f(X)|X}(y|x)
$$
$\left\{y_k \right \}$ is measurable then:
\begin{multline}
 P_{F(X)|X}(\Gamma_k^{c}|x)  = \int_{\left\{y_k \right \}} P_{F(X)|f(X)}(\Gamma_k^{c}|y) dP_{f(X)|X}(y|x) \nonumber \\
				 + \int_{B^{\mathcal{I}}-\left\{y_k \right \}} P_{F(X)|f(X)}(\Gamma_k^{c}|y) dP_{f(X)|X}(y|x) \nonumber
\end{multline}
\begin{multline}
 P_{F(X)|X}(\Gamma_k^{c}|x) \leq \nonumber \\
	  P_{f(X)|X}\left(  \left\{y_k \right \} |x \right).\left(  M e^{-a} + P_{f(X)F(X)}(i_{f(X)F(X)}\leq a)  \right) \nonumber \\
				 + P_{f(X)|X}\left(  B^{\mathcal{I}} - \left\{y_k \right \} |x \right) \nonumber
\end{multline}
If $x\in f^{-1}(y_k)$ then $P_{f(X)|X}\left(  \left\{y_k \right \} |x \right)=1$ and $P_{f(X)|X}\left(  B^{\mathcal{I}} - \left\{y_k \right \} |x \right)=0$, hence $\forall x\in f^{-1}(y_k)$:
$$
P_{F(X)|X}\left( \Gamma_k^c | x_k  \right) \leq M e^{-a} + P_{f(X)F(X)}(i_{f(X)F(X)}\leq a)
$$
\end{proof}

\subsection{Feinstein codes  and Extension of Feinstein theorem to noisy computations}

The alphabets $A$, $B$ and $C$ are now assumed {\em finite} and the noisy computation $[A,f \times F,B\times C]$ AMS and ergodic. This implies that the input rate $\overline{B}(X,f,F)$ is well defined for any AMS and ergodic source $X$. The assumption about alphabets is consistent, for example, with the case of digital computation i.e. if  $F$ represents a unreliable digital computer and $f$ a computable function.

\begin{definition} [Feinstein code for noisy computation] \label{DefinitionFeinsteinCodeForNoisyComputation}
Let $A$, $B$ and $C$  be  finite alphabets. Let $[A,f \times F,B\times C]$ be a noisy computation such that $f$ is a measurable weakly causal function. Let $n\geq N_f$. 

A $[M,n,\epsilon,\lambda]$-Feinstein code for the hookup  $[[A,X^n];[A,f^n \times F^n,B\times C]]$ is a set  $\left\{(A^n_i, \Gamma^n_i)\in \mathcal{B}_{A^n}\times \mathcal{B}_{C^n}, i=1,\cdots,M \right\}$ such that:
\begin{enumerate}
	\item $\Gamma_i \cap \Gamma_j = \emptyset$, $i\neq j$
	\item $\forall i=1,\cdots,M$, $P_{F^n(X^n)|X^n}(\Gamma_i^c|x^n) \leq \epsilon$ for any $x^n\in A^n_i$
	\item $\forall i=1,\cdots,M$, there exists $y_i \in B^n$ such that $A^n_i \subset  (f^n)^{-1}(y_i)$ and $P_{X^n|f^n(X^n)}(A^n_i|y_i)>1-\lambda$ 
\end{enumerate}
$F^n$ $(\epsilon,\lambda)$-reliably computes $f^n$ on the code $\left\{(A^n_i, \Gamma^n_i), i=1,\cdots,M \right\}$ for the $n^{th}$ extension of the source $[A,X]$. $\left\{A^n_i, i=1,\cdots,M \right\}$ is the input code and $\left\{ \Gamma^n_i, i=1,\cdots,M \right\}$ the output code.
\end{definition}

\begin{proposition} [Extension of Feinstein theorem to Noisy Computation] \label{PropositionNoisyComputation}
Let $A$, $B$ and $C$ be finite alphabets. Let $[A,f \times F,B\times C]$ be an AMS and ergodic noisy computation  such that $f$ is a measurable weakly causal function. Let $[A,X]$ be an AMS and ergodic source. 

Then, for any $R<\overline{B}(X,f,F)$, for any $\epsilon>0$ and any $\lambda>0$, there exists  $n(\epsilon,\lambda) \geq N_f$ such that $\forall n > n(\epsilon,\lambda)$, there exists a $[\lfloor e^{n(R-\overline{H}(X|f(X)))} \rfloor,n,\epsilon,\lambda]$-Feinstein code for the hookup $[[A,X^n];[A,f^n \times F^n,B\times C]]$.
\end{proposition}

\begin{proof}
The proof is derived from that of  Feinstein theorem given in \cite{Gray11}, chapter 14.
$$
R<\overline{B}(X,f,F) \Rightarrow R'=R-\overline{H}(X|f(X)) < \overline{I}(f(X),F(X))
$$
Let $\delta=\frac{\overline{I}(f(X),F(X))-R'}{2}>0$, $M=\lfloor e^{nR'}\rfloor$ and $a=n(R'+\delta)$. Let $E_n$ be the measurable set $\{ (y,z) \in B^n\times C^n / i_{f^n(X^n)F^n(X^n)}(y,z) \leq a \}$. Then $E_n$ is the set:
$$
\{ (y,z) \in B^n\times C^n / \frac{i_{f^n(X^n)F^n(X^n)}(y,z)}{n}\leq (\overline{I}(f(X),F(X))-\delta)  \} 
$$
Since $X$ and $[A,f \times F,B\times C]$   are AMS and ergodic then $(f(X),F(X))$ is AMS and ergodic. Thus, from the mutual information ergodic theorem (\cite{Gray11}, Theorem 8.1), $\frac{i_{f^n(X^n)F^n(X^n)}}{n}$ converges in $L^1$ to $\overline{I}(f(X),F(X))$. This implies that 
$$
\lim_{n \to \infty}  P_{f^n(X^n)F^n(X^n)} (E_n) = 0
$$
because $\delta>0$. As a consequence, for any $\epsilon > 0$ and $\lambda>0$, a $n \geq N_f$ large enough ensures that:
$$
Me^{-a}+P_{f^n(X^n)F^n(X^n)}(i_{f^n(X^n)F^n(X^n)}\leq a)  < min(\frac{1}{4}, \epsilon \lambda)
$$
$e^{n\delta}\leq Me^{-a}$, thus for any $\epsilon' \in ]0,\frac{1}{2}[$ and  $\lambda' \in ]0,\frac{1}{2}[$ such that $\epsilon' \leq \epsilon$ and $\lambda' \leq \lambda$ and
$$
\lambda'\epsilon' =  e^{-n\delta}+P_{f^n(X^n)F^n(X^n)}(E_n) < min(\frac{1}{4}, \epsilon \lambda)
$$
Lemma~\ref{LemmaFeinsteinLemmaForNoisyComputation} applied to $[[A,X^n];[A,f^n \times F^n,B\times C]]$ implies that there exists a $[M,n,\epsilon',\lambda']$-Feinstein code   $\left\{(A^n_i, \Gamma^n_i)\in \mathcal{B}_{A^n}\times \mathcal{B}_{C^n}, i=1,\cdots,M \right\}$ such that:
\begin{enumerate}
	\item $\Gamma_i \cap \Gamma_j = \emptyset$, $i\neq j$
	\item $P_{F^n(X^n)|X^n}(\Gamma_i^c|x^n) \leq \epsilon' \leq \epsilon $ for any $x^n\in A^n_i$, $i=1,\cdots,M$
	\item $\forall i=1,\cdots,M$, there exists $y_i \in B^n$ such that $A^n_i \subset  (f^n)^{-1}(y_i)$ and $P_{X^n|f^n(X^n)}(A^n_i|y_i)>1-\lambda'>1-\lambda$ 
\end{enumerate} 
\end{proof}

This proposition is a first justification of the definition of the input capacity $C_f(F)$ of a noisy computation $[A,f \times F,B\times C]$.





\pagebreak 

\section{Reliable computation in the presence of noise \label{SectionJointSourceComputationCodingTheorem}}

In this section, an abstract model of reliable computation is proposed. A consequence of Proposition \ref{PropositionNoisyComputation} is that, given a noisy computation $[A, f \times F,B]$, $\epsilon >0$ and $\delta>0$, $F^n$ $(\epsilon, \delta)$-reliably computes $f^n$ ($n$ large enough) only on a strict subset of the domain of $f^n$ (i.e., the input code). Thus, to achieve reliable computation of $g^k$ where $g$ is a weakly causal function on the entire domain of $g^k$ requires to rely on an ancillary weakly causal function $f$ and on encoding and decoding:
\begin{itemize}
	\item an input $k$-sequence of $g^k$ is encoded into an input $n$-sequence of $f^n \times F^n$ where $F$ is the noisy implementation of $f$
	\item this input $n$-sequence is fed into $F^n$ (the induced channel) and this produces an output $n$-sequence
	\item the output $n$-sequence is decoded to give a $k$-sequence which is an estimate of the expected result from $g^k$.
\end{itemize}
It is shown that this approach allows to achieve almost perfect  reliability in computing  $g^k(x^k)$ for large $k$ provided that the encoding rate is strictly lower than the capacity of the noisy computation\footnote{As in the classical reliable communication model, the encoder and decoder will be assumed perfectly reliable. This assumption will be discussed below.}. This is very similar to the case of reliable communication through noisy channels. 

A converse is also proved: if the rate of encoding is strictly greater than the capacity then no code exists to ensure reliable computation.

Moreover the reliable computation model  formally captures  practical approaches: e.g., a regular arithmetic addition $g$ obtained from the noisy actual circuit $F$ implementing $f$ which is an addition acting on residue encoded operands, see \cite{RaoFujiwara89}.

The weakly causal perfect functions $f$ and $g$ are assumed  unary. $m$-ary functions can be modeled as unary ones by concatenating $m$ input values in one ``meta''-input and thus modeling a {\em joint coding} of operands. This relaxes the assumption of independent coding of operands made by  \cite{Elias58}, \cite{PetersonRabin59}, \cite{Winograd62} and \cite{Ahlswede84}.  

\subsection{Model of reliable computation \label{SubsectionReliableComputation}}

 There is a need to constrain the encoding and decoding stages to forbid the model to reduce to (c.f. \cite{Elias58} and \cite{Winograd62}):
\begin{itemize}
	\item either an (assumed perfect) encoder which is equivalent to the expected function followed by an encoding of the result ($g^k(x'^k)$) suitable to transmission through the noisy device (then considered as a noisy transmission channel)
	\item or an encoder which encodes input values for reliable transmission through the noisy device (considered here also as a noisy transmission channel) and a decoder (assumed also reliable) which is equivalent to an (almost perfect) decoding followed by  the perfect expected function.
\end{itemize}
In both cases, the channel coding theorem sets limits on rates within which reliable transmission (and thus computation in these peculiar cases) is possible. Assume that $g$ is non-injective. Then, if $X'$ is a source, $H(g(X')) < H(X')$ (equivalently $H(X'|g(X'))>0$). The first case cannot be under the constraint that the encoding is an injection: a injective encoding cannot compress $X'$ into $g(X')$. The second case can also be avoided if the decoding is also based on an injection: the decoding will injectively associate each decodability region to a given output sequence $g^k(x'^k)$. With this constraint, the decoding cannot be equivalent to $f$ followed by an injective transcoding: the non-injective "part of the job" has to be done by the noisy device associating an input sequence to a decodability region\footnote{\cite{WinogradCowan63} argues that, from an Information Theory point of view, a frontier between computation and transmission can be drawn considering that computation is a ``true'' computation if the entropy is reduced ($H(g(X')) < H(X')$ i.e. $g$ is non injective) and that a bijection is equivalent to perfect transmission (preserving entropy).}.

The proposed model is thus mainly adapted to {\em non-injective} functions. Obviously, it comes down to the classical reliable communication model when $g$ is an injection.

The  model of the complete process to reliably compute prefixes of a weakly causal function $g :A'^{\mathcal{I}} \rightarrow B'^{\mathcal{I}}$ acting on a source $X'$, thanks to a noisy implementation $F$ of a weakly causal function $f: A^{\mathcal{I}} \rightarrow B^{\mathcal{I}}$ is the following:
\begin{itemize}
	\item {\bf encoding}: let $X^n$ be the $n^{th}$ extension of a  source for which  a  Feinstein code $(A^n_i,\Gamma_i)_{i=1,\cdots,M}$ allowing to $(\epsilon,\lambda)$-reliably compute $f^n(X^n)$ by $F^n(X^n)$ (cf Proposition~\ref{PropositionNoisyComputation} and Definition~\ref{DefinitionFeinsteinCodeForNoisyComputation}) is given ; a typical $k$-sequence $x'$ of $X'^k$ is encoded as a {\em typical}\footnote{This important assumption is motivated below.} $n$-sequence of $X^n$ by a  function, say $\mathcal{U}$, such that $\mathcal{U}(x') \in A_i^n$ for some $i=1,\cdots,M$.

If $\mathcal{U}$ is an injection then the encoding is said {\em injective}.

	\item {\bf computation of the noisy function}: $F^n$ is applied to $\mathcal{U}(x')$ producing a typical $n$-sequence $F^n(\mathcal{U}(x'))$ of $F^n(X^n)$ where $F^n(\mathcal{U}(x'))$ belongs to a given $\Gamma_i$ (with probability greater than $1-\epsilon$)
	\item {\bf decoding}: the first step is to associate to $F^n(\mathcal{U}(x'))$ the typical $n$-sequence $y_i$ of $f^n(X^n)$ corresponding to $\Gamma_i$, the second step is to apply to $y_i$ a function $\mathcal{V} : \{\mathbf{y_1},\ldots,\mathbf{y_M}\} \rightarrow \{\text{typical k-sequences of }g^k(X'^k)\} $ such that $\mathcal{V}(y_i)=g^k(x')$. The first step of decoding considered as a function from $\{\Gamma_i, i=1,\cdots,M\}$ to $\{y_i, i=1,\cdots,M\}$ is  injective. The characteristics of the second step are discussed below.
\end{itemize}
A decoding error occurs when one obtains a $n$-sequence $y_j$ (or equivalently a  $\Gamma_j$) such that $\widehat{g^k(x')} =  \mathcal{V}(y_j) \neq g^k(x')$.

To be able to define a decoding function $\mathcal{V}$ (i.e, a {\em deterministic} decoding), the encoding function $\mathcal{U}$ has to be such that the typical $n$-sequences of one $A_i^n \subset (f^n)^{-1} (y_i)$ ($ y_i \in \{y_1,\ldots,y_M\}$) are used for encoding typical $k$-sequences of {\em only one} $(g^k)^{-1} (z)$, $z$ typical $k$-sequence of $g^k(X'^k)$. Thus a decoding will be {\em deterministic} if:
$$
f^n(\mathcal{U}(x'_1)) = f^n(\mathcal{U}(x'_2))  \Rightarrow g^k(x'_1) = g^k(x'_2)
$$

\begin{figure}[htbp]
	$$
	(g^k)^{-1}(z)= \left \{ \begin{array}{c}  x'_{1} \\ \vdots \\ \vdots\\ \vdots \\  x'_{m}\end{array}\right \} \overset{\mathcal{U}}{\rightarrow}
	\left\{
              \begin{array}{l}
	          A_i^n= \left \{ \begin{array}{c}  x_{i_1} \\ \vdots \\ x_{i_N}\end{array}\right \} \overset{F^n}{\underset{1-\epsilon}{\rightarrow}} \Gamma_i  \rightarrow y_i  \\
	          \vdots \\
	          A_j^n= \left \{ \begin{array}{c}  x_{j_1} \\ \vdots \\ x_{j_M}\end{array}\right \} \overset{F^n}{\underset{1-\epsilon}{\rightarrow}} \Gamma_j \rightarrow y_j 
              \end{array}
            \right \}
            \overset{\mathcal{V}}{\rightarrow} z
	$$
	\caption{Deterministic decoding} \label{FigDeterministicDecoding}
\end{figure}

It can also be  required that $\mathcal{V}$ be an injection in order to forbid the decoder to be able to reliably compute the (non-injective) function $g^k$. If the function $\mathcal{V}$ is injective then the typical $k$-sequences of a $(g^k)^{-1} (z)$, $z$ typical $k$-sequence of $g^k(X'^k)$, are encoded in typical $n$-sequences of one and only one $A_i^n$ . Thus, if $x'_1$ and $x'_2$ are two typical $k$-sequences of $X'^k$: 
$$
f^n(\mathcal{U}(x'_1)) = f^n(\mathcal{U}(x'_2)) \Leftrightarrow g^k(x'_1) = g^k(x'_2)
$$
In this case the deterministic decoding is said to be {\em injective}, figure \ref{FigInjectiveDecoding}.

\begin{figure}[htbp]
	$$
	(g^k)^{-1}(z)= \left \{ \begin{array}{c}  x'_{1} \\ \vdots  \\ x'_{m}\end{array}\right \} \overset{\mathcal{U}}{\rightarrow}
              \begin{array}{l}
	          A_i^n= \left \{ \begin{array}{c}  x_{i_1} \\ \vdots \\ \vdots \\ x_{i_N}\end{array}\right \} \overset{F^n}{\underset{1-\epsilon}{\rightarrow}} \Gamma_i  \rightarrow y_i  
              \end{array}
            \overset{\mathcal{V}}{\rightarrow} z
	$$
	\caption{Injective decoding} \label{FigInjectiveDecoding}
\end{figure}

The model fulfills the constraints identified above. The encoder implements an injection and thus cannot be the desired function $g$ nor $f$ followed by a encoding for transmission (if $f$ and $g$ are not injective). The same comment applies to an injective decoding step  $\mathcal{V}$.

\begin{definition}\label{DefinitionEncodingDecoding}
 Let $g :A'^{\mathcal{I}} \rightarrow B'^{\mathcal{I}}$ and $f: A^{\mathcal{I}} \rightarrow C^{\mathcal{I}}$ be measurable weakly causal functions. Let $k\geq N_g$ and $n\geq N_f$. Let $A'^k_{\epsilon}$ and $B'^k_{\epsilon}$ be respectively the sets of $\epsilon$-typical $k$-sequences of $A'^k$ and of $B'^k$.
An {\em encoding-decoding $(1-\epsilon)$-compatible with} ($g^k$,$f^n$) is a pair $(\mathcal{U},\mathcal{V})$ where:
\begin{itemize}
	\item $\mathcal{U}: {A'}^k_{\epsilon} \rightarrow A^n$ is a function called the {\em encoding function}
	\item $\mathcal{V}: B^n \rightarrow {B'}^k_{\epsilon}$ is a function called the {\em decoding function} or the {\em deterministic decoding}
	\item $\forall (x'_1,x'_2)\in A'^k_{\epsilon} \times A'^k_{\epsilon}, f^n(\mathcal{U}(x'_1)) = f^n(\mathcal{U}(x'_2))  \Rightarrow g^k(x'_1) = g^k(x'_2)$
\end{itemize}
If $\forall (x'_1,x'_2)\in A'^k_{\epsilon} \times A'^k_{\epsilon}, f^n(\mathcal{U}(x'_1)) = f^n(\mathcal{U}(x'_2))  \Leftrightarrow g^k(x'_1) = g^k(x'_2)$, then the decoding is {\em injective}.
\end{definition}

\subsection{Achievability of an encoding rate}

\begin{definition}   \label{DefinitionEncodingRate}
Let $[A',X']$ be a source on a finite alphabet $A'$. Let $g :A'^{\mathcal{I}} \rightarrow B'^{\mathcal{I}}$ be a measurable weakly causal function and $[A,f \times F,B\times C]$ be a noisy computation such that $f$ is a weakly causal function. 
Let  $R>0$. The encoding rate $R$ is said {\em achievable} if there exists a source $[A,X_0]$ and a sequence  $\left(  \left\{ (A^{n_j}_i , \Gamma^{n_j}_i) , i=1,\cdots, M_j \right\} \right)_{j \geq 1}$ of $(M_j,n_j,\epsilon_j,\lambda_j)$-Feinstein codes for the hookups $[[A^{n_j} , X_0^{n_j}];[A^{n_j},f^{n_j} \times F^{n_j},B^{n_j}\times C^{n_j}]]$, $j\geq 1$, $n_j \geq N_f$,  such that:
\begin{enumerate}
	\item $\lim_{j \rightarrow \infty} \epsilon_j = 0$
	\item $\lim_{j \rightarrow \infty} \lambda_j = 0$
	\item $M_j = \lfloor e^{n_j(R-\overline{H}(X_0|f(X_0))} \rfloor$
	\item there exists an encoding-decoding $(\mathcal{U}_j,\mathcal{V}_j)$ $(1-\epsilon)$-compatible with $(g^{k_j},$ $f^{n_j})$   such that $\mathcal{U}_j(A'^{k_j}_{\epsilon}) \subset \bigcup_{i=1}^{M_j} A^{n_j}_i$ and $R.n_j=\overline{H}(X').k_j$
\end{enumerate}
\end{definition}

{\em Remark:} $\mathbb{Q}$ being dense in $\mathbb{R}$,  $\frac{R}{H(X')}$ is assumed rational.


\begin{proposition} \label{PropositionCoding}
Let $g :A'^{\mathcal{I}} \rightarrow B'^{\mathcal{I}}$ be a measurable weakly causal function defining an AMS and ergodic deterministic channel. Let $[A,f \times F,B\times C]$ be an AMS and ergodic  noisy computation such that $f$ is a weakly causal function.  Let  $[A',X']$ be an AMS and ergodic source. Let $R>0$. Then
$$
R < C_f(F) \Longrightarrow \text{ R is achievable}
$$
Assume that  $\overline{H}(X'|g(X')) \neq 0$. If
$$
R < \sup_{P_X\in \mathfrak{P}_{ae}(A^\mathcal{I})} \left[ min \left( \frac{\overline{H}(X|f(X))}{\overline{H}(X'|g(X'))}.\overline{H}(X') , \overline{B}(X,f,F) \right) \right]
$$
then  R is achievable for an injective decoding.
\end{proposition}

\begin{proof}

Let $(\epsilon_j)_{j\in \mathbb{N}}$ and  $(\lambda_j)_{j\in \mathbb{N}}$ be two sequences of positive real numbers such that $\lim_{j \rightarrow \infty} \epsilon_j = 0$ and $\lim_{j \rightarrow \infty} \lambda_j = 0$. 

Let $R>0$ be such that $R < C_f(F)$. Then there exists a source $[A,X_0]$, $P_{X_0} \in \mathfrak{P}_{ae}(A^\mathcal{I})$, and two positive real numbers $R'$ and $\delta$ such that $R < R'-\delta < R'< \overline{B}(X_0,f,F) \leq C_f(F)$. 

For a fixed $j$, let $\delta'>0$ be such that $\lambda_j-\delta'>0$. From Proposition~\ref{PropositionNoisyComputation}, there exists an integer $n(\epsilon_j,\lambda_j) \geq N_f$ such that, for any $n_j\geq n(\epsilon_j,\lambda_j)$, there exists a $(M'_j,n_j,\epsilon_j,\lambda_j-\delta')$-Feinstein code for the hookup $[[A^{n_j} , X_0^{n_j}];[A^{n_j},f^{n_j} \times F^{n_j},B^{n_j}\times C^{n_j}]]$ where $M'_j=\lfloor e^{n_j(R'-\overline{H}(X_0|f(X_0))} \rfloor$. $n_j$ can be chosen large enough to also have $k_j\geq N_g$. Let $(\tilde{A}^{n_j}_i, \Gamma^{n_j}_i)_{i=1,\cdots,M'_j}$ be this code for a given $n_j$. Then:
$$
P_{X_0^{n_j}|f^{n_j}(X_0^{n_j})}(\tilde{A}^{n_j}_i | y_i) > 1-(\lambda_j-\delta')
$$

Let $\epsilon'>0$. For $n_j$ large enough, the subset $[(f^{n_j})^{-1}(y_i)]_{typ}$ of $n_j$-sequences  $\epsilon'$-conditionally typical given $y_j$ has a conditional probability given $y_i$:
$$
P_{X_0^{n_j}|f^{n_j}(X_0^{n_j})}([(f^{n_j})^{-1}(y_i)]_{typ}| y_i) > 1-\delta'
$$
Thus, setting $A^{n_j}_i = \tilde{A}^{n_j}_i \cap [(f^{n_j})^{-1}(y_i)]_{typ}$:
$$
P_{X_0^{n_j}|f^{n_j}(X_0^{n_j})}(A^{n_j}_i | y_i) > 1-\lambda_j
$$
Obviously, for any $x\in A^{n_j}_i$
$$
P_{F^{n_j}(X_0^{n_j})|X_0^{n_j}}((\Gamma^{n_j}_i)^c | x) \leq \epsilon_j
$$
This implies that $(A^{n_j}_i, \Gamma^{n_j}_i)_{i=1,\cdots,M'_j}$ is a $(M'_j,n_j,\epsilon_j,\lambda_j)$-Feinstein code for the hookup $[[A^{n_j} , X_0^{n_j}];[A^{n_j},f^{n_j} \times F^{n_j},B^{n_j}\times C^{n_j}]]$ and each $A^{n_j}_i$ contains only $\epsilon'$-conditionally typical sequences given $y_i$.

Let $\alpha_1$ the number of $\epsilon'$-typical (given $z$) sequences within $(g^{k_j})^{-1}(z)$ ($z$ typical $k$-sequence). Let $\alpha_2$ the number of $\epsilon'$-typical (given $y_i$) sequences within $A^{n_j}_i$. Thanks to conditional AEP (theorem~\ref{TheoremConditionalAEP}), for $n_j$ (and thus $k_j=\frac{R}{\overline{H}(X')}n_j$) large enough:
$$
(1-\lambda_j)e^{k_j(\overline{H}(X'|g(X'))-\epsilon')}\leq \alpha_1 \leq e^{k_j(\overline{H}(X'|g(X'))+\epsilon')}
$$
and
$$
(1-\lambda_j)e^{n_j(\overline{H}(X_0|f(X_0))-\epsilon')}\leq \alpha_2 \leq e^{n_j(\overline{H}(X_0|f(X_0))+\epsilon')}
$$
An encoding-decoding $(1-\epsilon')$-compatible with $(g^{k_j},f^{n_j})$ using the $(A^{n_j}_i)_{i=1,\ldots,M'_j}$  will require that  the encoding of the conditionally typical sequences of one inverse image $(g^{k_j})^{-1}(z)$ consume a number  (i.e. $\lceil \alpha_1 / \alpha_2 \rceil$) of $A^{n_j}_i$  which is upper bounded by\footnote{The constraint is due to the determinism of the decoding. The encoding may or may not be injective: the constraint allows injective encoding but does not make it mandatory.}:
$$
\left\lceil \frac{e^{{k_j}(\overline{H}(X'|g(X'))+\epsilon')}}{(1-\lambda_j)e^{n_j(\overline{H}(X_0|f(X_0))-\epsilon')}} \right\rceil
$$

If $N$ is the total number of $A^{n_j}_i$ needed to encode all the inverse images $(g^{k_j})^{-1}(z)$ then:
$$
N \leq \left\lceil \frac{e^{k_j(\overline{H}(X'|g(X'))+\epsilon')}}{(1-\lambda_j)e^{n(\overline{H}(X_0|f(X_0))-\epsilon')}} \right\rceil . e^{k_j(\overline{H}(g(X')+\epsilon')}
$$
This leads to:
\begin{eqnarray}
N 	&  \leq & \left[  \frac{e^{k_j(\overline{H}(X'|g(X'))+\epsilon')}}{(1-\lambda_j)e^{n_j(\overline{H}(X_0|f(X_0))-\epsilon')}} + 1 \right] . e^{k_j(\overline{H}(g(X')+\epsilon')} \nonumber \\
	& \leq &  \left[  \frac{e^{k_j(\overline{H}(X'|g(X'))+ \epsilon' + \epsilon'')}}{e^{n_j(\overline{H}(X_0|f(X_0))-\epsilon')}} + 1 \right] . e^{k_j(\overline{H}(g(X')+\epsilon')} \nonumber
\end{eqnarray}
where $e^{k\epsilon''} = 1/(1-\lambda_j)$. For any $\epsilon'''>0$, for $k_j$ and $n_j$ large enough:
$$
 \frac{e^{k_j(\overline{H}(X'|g(X'))+ \epsilon' + \epsilon'')}}{e^{n_j(\overline{H}(X_0|f(X_0))-\epsilon')}} + 1  \leq   \frac{e^{k_j(\overline{H}(X'|g(X'))+ \epsilon' + \epsilon'')}}{e^{n_j(\overline{H}(X_0|f(X_0))-\epsilon')}} . e^{k_j\epsilon'''}
$$
giving:
\begin{eqnarray}
N 	& \leq &  \frac{e^{k_j(\overline{H}(X'|g(X'))+ \epsilon' + \epsilon''+\epsilon''')}}{e^{n_j(\overline{H}(X_0|f(X_0))-\epsilon')}} . e^{k_j(\overline{H}(g(X')+\epsilon')} \nonumber \\
	& \leq &  \frac{e^{k_j(\overline{H}(X')+2\epsilon' + \epsilon''+\epsilon''')}}{e^{n_j(\overline{H}(X_0|f(X_0))-\epsilon')}} \nonumber 
\end{eqnarray}
Since $k_j.\overline{H}(X')=n_j.R$:
\begin{eqnarray}
N	& \leq &  \frac{e^{n_j(R+\frac{\epsilon' + 2\epsilon''+\epsilon'''}{\overline{H}(X')})}}{e^{n_j(\overline{H}(X_0|f(X_0))-\epsilon')}} \nonumber \\
	& \leq &  e^{n_j(R-\overline{H}(X_0|f(X_0))+(2\epsilon'+\epsilon''+\epsilon''')/\overline{H}(X') - \epsilon)} \nonumber
\end{eqnarray}
for suitably chosen $\epsilon'$, $\epsilon''$ and $\epsilon'''$ and large enough $k_j$ and $n_j$. Hence:
\begin{eqnarray}
N 	& \leq &  e^{n_j(R'-\overline{H}(X_0|f(X_0)))} \nonumber \\
	& \leq & \lfloor e^{n_j(R'-\overline{H}(X_0|f(X_0)))} \rfloor \nonumber
\end{eqnarray}
since $N$ is an integer. Thus
$$
N \leq M'_j
$$
In words, there are enough $A^{n_j}_i$  in the Feinstein code to encode the inverse images $(g^{k_j})^{-1}(z^{k_j})$. This closes the first step of the proof.

An injective decoding will be possible if one inverse image $(g^{k_j})^{-1}(z)$ can be encoded using only one $A^{n_j}_i$. For $n_j$ and $k_j$ large enough, a sufficient condition is:
$$
e^{k_j\overline{H}(X'|g(X'))}<e^{n_j\overline{H}(X_0|f(X_0))}
$$
which is equivalent to:
$$
\frac{k_j}{n_j} < \frac{\overline{H}(X_0|f(X_0))}{\overline{H}(X'|g(X'))} 
$$
Since $\frac{k_j}{n_j} = \frac{R}{\overline{H}(X')}$:
$$
R < \frac{\overline{H}(X_0|f(X_0))}{\overline{H}(X'|g(X'))}.\overline{H}(X')
$$
The condition 
$$
R <  \sup_{P_X\in \mathfrak{P}_{ae}(A^\mathcal{I})} \left[ min \left( \frac{\overline{H}(X|f(X))}{\overline{H}(X'|g(X'))}.\overline{H}(X') , \overline{B}(X,f,F) \right) \right] 
$$
 will thus ensure that an injective encoding of $X'$ on a Feinstein code and allowing an injective decoding is possible.
\end{proof}
Proposition \ref{PropositionCoding} does not rely on the assumption of an injective encoding: it also holds for an injective encoding. Thus such an assumption does not narrow the context of Proposition \ref{PropositionCoding}, it makes the model compliant with the constraints expressed above.

The use of conditional typical sequences and of an expected function $f$ is worth some comments. Designing a code is finding a set of pairs $\{(A^n_i,\Gamma_i),i=1,\cdots M\}$ such that the ``image'' of $A^n_i$ (set of typical sequences) by the ``noisy function'' $F^n$ falls within $\Gamma_i$ with high probability. The approach followed here has been to determine this code thanks to a function $f$ called the expected function. A dual approach could be to choose a code and then to determine a function $f$  such that  $A^n_i \subset f^{-1}(y_i)$ for a given collection of typical $n$-sequences $y_i$ and such that the typical $n$-sequences of $X^n$ belonging to $A^n_i$ are conditionally typical given $y_i$. 

Seeking a relevant concept of capacity requires to associate to a code a cost measure  which  is a rate of encoding. For channel coding and fixed block encoding, the rate determines the number of usable input blocks for a given block size. For computation coding, two cardinalities are needed: the number of sets $A^n_i$  and the number of elements in each $A^n_i$.  A convenient consequence of determining the $A^n_i$'s by a function $f^n$ is that, for large $n$, the $A^n_i$'s have almost the same number of well identified elements, that is the $n$-sequences conditionally typical given $y_i$. A unique number, i.e. $\overline{H}(X|f(X))$, gives the common size (i.e. $\simeq e^{n\overline{H}(X|f(X))}$) of  precisely determined subsets  $A^n_i$'s. Moreover, such a subset $A^n_i$ has almost the same probability measure as the entire $(f^n)^{-1}(y_i)$. The method gives a balanced code which is well characterized by a  unique number (the typical input rate) which represents the two cardinalities mentioned above.

\subsection{The converse \label{SubsectionConverseTheorem}}

\begin{proposition} \label{PropositionConverse}
Let $g :A'^{\mathcal{I}} \rightarrow B'^{\mathcal{I}}$ be a measurable weakly causal function defining an AMS and ergodic deterministic channel. Let $[A,f \times F,B\times C]$ be an AMS and ergodic  noisy computation. Assume that $f$ is a weakly causal function. Let  $[A',X']$ be an AMS and ergodic source. Let $R>0$. If $R>C_f(F)$,  $R$ is not achievable.
\end{proposition}

\begin{proof}
$R=\frac{k}{n} \overline{H}(X') > C_f(F)$ then 
$$
\forall X \in \mathfrak{P}_{ae}(A^\mathcal{I}), R=\frac{k}{n} \overline{H}(X') > \overline{H}(X|f(X)) + \overline{I}(f(X);F(X))
$$
Let $ X_0 \in \mathfrak{P}_{ae}(A^\mathcal{I})$. 

Using the same notations (dropping the subscript $j$) as in the proof of Proposition~\ref{PropositionCoding}, if the decoding is deterministic then the encoding of the conditionally typical sequences of an inverse image $(g^k)^{-1}(z^k)$ uses a number of  $A^n_i$ (or equivalently of $(f^n)^{-1}(y_i)$) which is lower bounded by:
$$
\left\lceil \frac{(1-\lambda) e^{k(\overline{H}(X'|g(X'))-\epsilon')}}{e^{n(\overline{H}(X_0|f(X_0))+\epsilon')}} \right\rceil
$$

If $N$ is the total number  $A^n_i$ (or equivalently of inverse images $(f^n)^{-1}(y_i)$ or equivalently the number of typical sequences of $f^n(\mathcal{U}(X'^k))$ ) needed to encode all the inverse images $(g^k)^{-1}(z)$ ($z$ typical $k$-sequences) then:
\begin{eqnarray}
N & \geq &  \left\lceil \frac{(1-\lambda) e^{k(\overline{H}(X'|g(X'))-\epsilon')}}{e^{n(\overline{H}(X_0|f(X_0))+\epsilon')}} \right\rceil .(1-\lambda) e^{k(\overline{H}(g(X'))-\epsilon')} \nonumber \\
    & \geq & e^{ k(\overline{H}(X')-2\epsilon') - n(\overline{H}(X_0'|f(X_0))+\epsilon') } (1-\lambda)^2\nonumber \\
    & \geq & e^{ n.\frac{R}{\overline{H}(X')}(\overline{H}(X')-2\epsilon') - n(\overline{H}(X_0|f(X_0))+\epsilon') } (1-\lambda)^2\nonumber \\
    & \geq & e^{ n.(R- \frac{2\epsilon'}{\overline{H}(X')}-\overline{H}(X_0|f(X_0))+\epsilon') } (1-\lambda)^2 \nonumber \\
    & > & e^{n(\overline{I}(f(X_0 ; F(X_0)))+\epsilon'')} \label{EqNumberOfInverseImages} 
\end{eqnarray}
for suitably chosen $\lambda$, $\epsilon'$ and $\epsilon''$. This is possible to find such  $\lambda$, $\epsilon'$ and $\epsilon''$ since $R - \overline{H}(X_0|f(X_0)) > \overline{I}(f(X_0);F(X_0))$.

Hence, for any $\epsilon''$, if $n$ is large enough, (\ref{EqNumberOfInverseImages}) holds.

 Assume that, thanks to deterministic encoding and decoding of rate $R'$, an ergodic and AMS  random process $Y^l$ (associated to a source $[B,Y]$) is to be obtained from $F^n(X^n)$ (i.e. $R'=\frac{l}{n}\overline{H}(Y)$): 
 $Y^l\rightarrow f^n(X^n)  \rightarrow F^n(X^n)  \rightarrow \widehat{Y^l} $ is a Markov Chain (where $\widehat{Y^l}$ is the "estimate" of $Y^l$ after decoding).

This implies (thanks to Data Processing Inequality) that:
\begin{multline}
I(Y^l ; \widehat{Y^l})  \leq  I(f^n(X^n) ; F^n(X^n)) \nonumber \\
\Rightarrow H(Y^l) - H(Y^l|\widehat{Y^l})  \leq  I(f^n(X^n) ; F^n(X^n)) \nonumber
\end{multline}
Let $P_e(l)=P(\widehat{Y^l} \neq Y^l)$.  By Fano's inequality:
\begin{multline}
H(Y^l) - (H_2(P_e(l))+l.P_e(l).log(|B|))  \leq  I(f^n(X^n) ; F^n(X^n)) \nonumber \\
\Rightarrow \frac{H(Y^l)}{l} \leq \frac{I(f^n(X^n) ; F^n(X^n))}{l} +   \frac{(H_2(P_e(l))+l.P_e(l).log(|B|))}{l}\nonumber \\
\Rightarrow \frac{H(Y^l)}{l} \leq \frac{I(f^n(X^n) ; F^n(X^n))}{n}.\frac{\overline{H}(Y)}{R'} +   \frac{(H_2(P_e(l))+l.P_e(l).log(|B|))}{l}\nonumber
\nonumber
\end{multline}
If the error probability $P_e(l)$ vanishes (i.e. $\lim_{l \to \infty}P_e(l)=0$) then, necessarily:
$$
R' \leq \overline{I}(f(X);F(X))
$$
Hence, given $\epsilon''>0$, for large enough $l$, the number $N_Y(l)$ of typical sequences of the source $[Y,B']$ necessarily verifies
$$
N_Y(l) \leq e^{n + \overline{I}(f(X);F(X)) \epsilon''}
$$
This is nothing else than the classical proof of the converse of the Channel Coding Theorem.

Assuming there exists an injection $\phi$ associating to typical sequences of $f^n(\mathcal{U}(X'^k))$ typical sequences of $f^n(X^n)$, it is always possible to find an injection $\psi$ associating to typical sequences of $f^n(\mathcal{U}(X'^k))$ typical sequences of $Y^l$ and an injection $\delta$ associating to typical sequences of $Y^l$ typical sequences of $Y^l$ $f^n(X^n)$ such that $\psi=\delta \circ \phi$. To obtain an asymptotically perfect estimation of $f^n(\mathcal{U}(X'^k))$, it is necessary to have 
$$
N \leq N_Y(l) \leq e^{n + \overline{I}(f(X);F(X)) \epsilon''}
$$
which is impossible by (\ref{EqNumberOfInverseImages}).

Then, by (\ref{EqNumberOfInverseImages}), the error probability cannot tends to $0$ (the encoding of $X'$ by $X_0$ requires more inverse images $(f^n)^{-1}(y^n)$ or $A_i^n$ than "available").

In this proof, the use of an "ancillary" ergodic AMS source $[B,Y]$ is motivated by the fact that $f^n(\mathcal{U}(X'^k))$ cannot be assumed to be the prefix of an ergodic AMS process $f(\mathcal{U}^*(X'))$.
\end{proof}

\subsection{Reliable computation and noisy encoding \label{SubsectionNoisyEncoding}}

The model of reliable computation assumes that the encoder and the decoder are perfectly reliable. This assumption could be justified by quoting from \cite{Winograd62} {\em ``The computation system [model] was devised for the sole purpose of studying the relation of information theory and reliable automata''}. Moreover, it could be argued that if the complexity of the computation device is of a much greater magnitude than that of the encoder and decoder then the unreliability of the encoder and decoder have almost no impact on the overall reliability of the computation and thus can be neglected. For complex systems, this is quite realistic. 

In any case, it is impossible to overcome the fact that the reliability reached is at the best the reliability of the final decoding device. The only way is to built a intrinsically reliable enough decoder (for example thanks to gate redundancy).  

A noisy encoder is a noisy computation itself and thus can be handled from the point of view of ``cascaded noisy computations''. To allow downstream reliable computation, a noisy encoder has to possess intrinsic performance which can be expressed as follows.

 For any $\epsilon>0$ and $\lambda>0$, it should exist $k(\epsilon,\lambda)$ such that for any $k>k(\epsilon,\lambda)$ and for any $z\in {C'}^k$ there must exist a set $A'_z \subset (g^k)^{-1}(z)$ and a set $A_z \subset \cup_{y\in \mathcal{V}^{-1}(z)}(f^n)^{-1}(y)$ for which:
\begin{eqnarray}
P_{X^n|X'^k}( A_z^c | x' ) & < & \epsilon \text{    }\forall x'\in A'_z \label{EqNoisyEncoding1} \\
P_{X'^k|g(X'^k)}(A'_z | z )& > & 1-\lambda \label{EqNoisyEncoding2}
\end{eqnarray}
This just states that encoding $x' \in (g^k)^{-1}(z)$ by any randomly chosen $x\in\cup_{y\in \mathcal{V}^{-1}(z)}(f^n)^{-1}(y)$ will preserve the input code for the noisy computation and thus allows correction through the decoding. Stated another way, the encoding errors remain compatible with the code of the hookup $[[A,X^n];$ $[A,f^n \times F^n),B\times C]]$.

(\ref{EqNoisyEncoding1}) and (\ref{EqNoisyEncoding2}) characterize $(A_z , A'_z)_{z\in C'^k}$ as a $(M,k,\epsilon,\delta)$-Feinstein code for the hookup  $[[A', X'^{k}] ; [A', (u^k,\mathcal{U} ), A\times A] ]$ where $u$ stands for the perfect encoding function, $\mathcal{U}$ for the noisy one and $M=|{C'}^k|$.

\pagebreak

\section{Example: Computation with noisy inputs \label{SectionExampleComputationWithNoisyInputs}}

In the special case where a computation is noisy due to the input device, it is possible to assert that correcting the  results of a function applied to noisy input will offer a better input rate than correcting inputs before applying the function. This case is simply captured by a noisy computation $[A, f\times F, B\times B)$ where:
\begin{itemize}
\item $f$ is the expected perfect function
\item $[A,F,B]$ is a cascade of a channel $[A,\nu,A]$ and of $[A,f,B]$ (i.e. $F\equiv \nu f$)
\end{itemize}
\begin{figure}[htbp]
	\begin{center}
	\xymatrix{
    		\overset{\txt{Input source }}{ [A,X] }	\ar  [r] ^\nu 	\ar@{.>} [dr] ^f &   \overset{\txt{Noisy input source  }}{[A,Y]} \ar  [r] ^f & \overset{\txt{Noisy result}}{[B,f(Y)]} \\
																		&   \overset{\txt{Expected result}}{[B,f(X)]} 
 	 }
	\end{center}
	\caption{Computation with noisy inputs} \label{FigComputationWithNoisyInputs}
\end{figure}
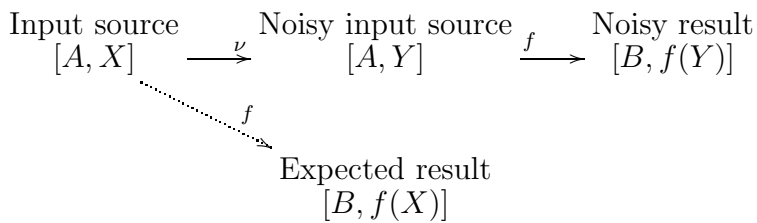
The following lemma asserts that the capacity of the channel $\nu$ is lower or equal than the capacity of the noisy function $f\times (\nu f)$. Then it is possible to obtain a greater input rate for $f\times (\nu f)$ than for $\nu$.
\begin{lemma}\label{LemmaComputationWithNoisyInputs}
Let $[A, f\times F, B\times B]$ be an AMS and ergodic noisy computation such that 
\begin{itemize}
\item $f$ is weakly causal
\item $[A,F,B]$ is a cascade of an AMS and ergodic channel $[A,\nu,A]$ and of $[A,f,B]$ (i.e. $F\equiv \nu f$)
\end{itemize}
Let $C_\nu$ denote the capacity of the channel $[A,\nu,A]$ and $C_f(F)$ denote the capacity of the noisy computation $[A, f\times F, B\times B]$. Then
$$
C_\nu \leq C_f(F)
$$
\end{lemma}
\begin{proof}
Let $X$ and $Y$ two random variables with values in a finite set $A$. Let $f:A \to B$ a measurable function from $A$ to a finite set $B$. Then
\begin{eqnarray}
H(X|Y) & = & \sum_{(a_1,a_2) \in A^2}  P_{X,Y}(a_1,a_2) log\left( \frac{P_Y(a_2)}{P_{X,Y}(a_1,a_2)}\right) \nonumber \\
	& = & \sum_{(b_1,b_2) \in B^2} \sum_{(a_1,a_2) \in (f^{-1}(b_1) \times f^{-1}(b_2)} P_{X,Y}(a_1,a_2) log \left( \frac{P_Y(a_2)}{P_{X,Y}(a_1,a_2)}\right) \nonumber 
\end{eqnarray}
By the log-sum inequality:
\begin{eqnarray}
H(X|Y) 	& \geq & \sum_{(b_1,b_2) \in B^2} \left( \sum_{(a_1,a_2) \in (f^{-1}(b_1) \times f^{-1}(b_2)} P_{X,Y}(a_1,a_2) \right) . \nonumber\\
	&          &  log\left( \frac{\sum_{(a_1,a_2) \in (f^{-1}(b_1) \times f^{-1}(b_2)} P_Y(a_2)}{\sum_{(a_1,a_2) \in (f^{-1}(b_1) \times f^{-1}(b_2)} P_{X,Y}(a_1,a_2)}\right) \nonumber  \\
	& \geq & \sum_{(b_1,b_2) \in B^2}  P_{f(X),f(Y)}(b_1,b_2)  log\left( \frac{P_{f(Y)}(b_2) \sum_{(a_1) \in (f^{-1}(b_1)} 1}{P_{f(X),f(Y)}(b_1,b_2)}\right) \nonumber  \\
	& \geq & \sum_{(b_1,b_2) \in B^2}  P_{f(X),f(Y)}(b_1,b_2)  log\left( \frac{P_{f(Y)}(b_2)}{P_{f(X),f(Y)}(b_1,b_2)} \right) \nonumber \\
	& \geq & H(f(X)|f(Y)) \nonumber
\end{eqnarray}
Let $[A,X]$ and $[A,Y]$ be two AMS and ergodic sources and $f: A^\mathcal{I} \to B^\mathcal{I}$ a weakly causal measurable function. Then for any $n$ large enough, ${H}(X^n|Y^n) \geq {H}(f^n(X^n)|f^n(Y^n))$. Taking the limit when $n$ tends to infinity:
$$
\overline{H}(X|Y) \geq \overline{H}(f(X)|f(Y))
$$
Let $[A, f\times F, B\times B]$ be an AMS and ergodic noisy computation such that 
\begin{itemize}
\item $f$ is weakly causal
\item $[A,F,B]$ is a cascade of an AMS and ergodic channel $[A,\nu,A]$ and of $[A,f,B]$ (i.e. $F\equiv \nu f$)
\end{itemize}
Let $[A,X]$  be an AMS and ergodic source and let and $[A,Y]$ denote the output of the channel $\nu$ when fed by the source $X$.
\begin{eqnarray}
C_\nu - C_f(F) & = & \sup_X (\overline{H}(X) - \overline{H}(X|Y)) -  \sup_X (\overline{H}(X) - \overline{H}(f(X)|f(Y))) \nonumber \\
		& = & \sup_X (\overline{H}(X) - \overline{H}(X|Y)) + \inf_X (-\overline{H}(X) + \overline{H}(f(X)|f(Y))) \nonumber
\end{eqnarray}
Then, for any source $X'$:
$$
C_\nu - C_f(F) \leq \sup_X (\overline{H}(X) - \overline{H}(X|Y)) -\overline{H}(X') + \overline{H}(f(X')|f(Y')))
$$
Thus, for any $\epsilon >0$, there exists a source $X_0$ such that:
$$
C_\nu - C_f(F) \leq \epsilon +  \overline{H}(X_0) - \overline{H}(X_0|Y_0) -\overline{H}(X_0) + \overline{H}(f(X_0)|f(Y_0)))
$$
Since $\overline{H}(X_0|Y_0) \geq \overline{H}(f(X_0)|f(Y_0)) $, for any $\epsilon >0$
$$
C_\nu - C_f(F) \leq \epsilon
$$
Then $C_\nu \leq C_f(F)$

\end{proof}

As an example, assume that: 
\begin{itemize}
	\item $A$ is a set of binary words of some finite length, $f$ is such that, for all $x\in A^\mathcal{I}$, $f(x)=\phi(x_0)\cdots \phi(x_i)\cdots$ where $\phi$ is a Turing-computable function from $A$ to $B=\phi(A)$ ; then $f$ is obviously stationary and weakly causal
	\item the noisy channel $[A,F,C]$ is made of a stationary and ergodic noisy communication channel $[A,\nu,A]$ (e.g. a memoryless channel) followed by a perfect instance of $f$. I.e. $B=C$ and $[A,F,B]$ is a cascade $\nu f$. 
\end{itemize}
This models an infinite sequence of computations of a computable function $\phi$ on inputs acquired through a noisy channel. Then, thanks to the following lemmas (Lemma \ref{LemmaStationaryErgodicNoisyInputFunction} and Lemma \ref{LemmaStationaryErgodicFunctionChannelProduct}, proved in Appendix), Propositions ~\ref{PropositionNoisyComputation}, ~\ref{PropositionCoding} and ~\ref{PropositionConverse}  and Lemma \ref{LemmaComputationWithNoisyInputs} hold.

\begin{lemma}\label{LemmaStationaryErgodicNoisyInputFunction}
If  $[A,\nu,A]$ is a stationary ergodic channel and $f: A^\mathcal{I} \to B^\mathcal{I}$ is a stationary function  (i.e. $f(T_A x) = T_B f(x)$ for all $x\in A^\mathcal{I}$) then the cascade $[A,\nu f, B]$  is a stationary and ergodic channel $[A,F,B]$ ($F \equiv \nu f$).
\end{lemma}

\begin{lemma}\label{LemmaStationaryErgodicFunctionChannelProduct}
If  $f$ is a stationary function  (i.e. $f(T_A x) = T_B f(x)$ for all $x\in A^\mathcal{I}$) and $F$ is stationary and ergodic then $[A,f \times F,B\times C]$ is stationary and ergodic.
\end{lemma}

\pagebreak

\part{Ergodic and AMS noisy computations}\label{PartErgodicAMSNoisyComputations}

\pagebreak

\section{Channel Products}\label{SectionChannelProduct}

In this section,  it is proved that a channel product $[A, \nu_1 \times \nu_2, B\times C]$ is respectively stationary, AMS (w.r.t. a stationary source), recurrent or output weakly mixing if and only if both channels $[A,\nu_1,B]$ and $[A,\nu_2,C]$ are respectively stationary, AMS (w.r.t. a stationary source), recurrent or output weakly mixing.  The section ends with corollaries dedicated to noisy computations.

\subsection{Stationary and AMS Channel Products}\label{SubsectionAMSChannelProducts}

\begin{proposition}\label{PropositionCNSQuasiStationaryChannelProduct}
The channel product $[A, \nu_1 \times \nu_2 , B\times C]$ is stationary w.r.t  a stationary source distribution $\mu$ if and only if $[A, \nu_1 , B]$ and $[A,  \nu_2 ,  C]$ are stationary w.r.t $\mu$.
\end{proposition}
\begin{proof}
If the channel product is stationary w.r.t $\mu$ then the hookup $\mu(\nu_1 \times \nu_2)$ is stationary. Thus the marginals $\mu\nu_1$ and $\mu\nu_2$ are stationary. Then $[A, \nu_1 , B]$ and $[A,  \nu_2 ,  C]$ are stationary w.r.t $\mu$.

Assume that $[A, \nu_1 , B]$ and $[A,  \nu_2 ,  C]$ are stationary w.r.t $\mu$. $\forall E  \in \mathcal{B}_{A^\mathcal{I}}$, $\forall G  \in \mathcal{B}_{B^\mathcal{I}}$ and $\forall H  \in \mathcal{B}_{C^\mathcal{I}}$:
$$
\mu(\nu_1 \times \nu_2)(E\times G \times H)= \int_E \nu_1(x,G).\nu_2(x,H) d\mu
$$
Since $\mu$ is stationary, by a change of variable:
$$
\mu(\nu_1 \times \nu_2)(E\times G \times H)= \int_{T_A^{-1}E} \nu_1(T_A x,G)\nu_2(T_Ax,H) d\mu
$$
Thanks to lemma~\ref{LemmaStationaryChannel} (p. \pageref{LemmaStationaryChannel})
$$
\mu(\nu_1 \times \nu_2)(E\times G \times H)= \int_{T_A^{-1}E} \nu_1(x,T_B^{-1}G)\nu_2(x,T_C^{-1}H) d\mu
$$
Then $\mu(\nu_1 \times \nu_2)(E\times G \times H)=\mu(\nu_1 \times \nu_2)(T_A^{-1} E\times T_B^{-1}G \times T_C^{-1}H)$. Stationarity on the semi-algebra of rectangles implies stationarity on the product $\sigma$-field (\cite{Kakihara99}, Remark 3, p. 77) then $\mu(\nu_1 \times \nu_2)$ is stationary.
\end{proof}

\begin{proposition}\label{PropositionCNSAMSChannelProduct}
The channel product $[A, \nu_1 \times \nu_2 , B\times C]$ is AMS w.r.t  a stationary source distribution $\mu$ if and only if $[A, \nu_1 , B]$ and $[A,  \nu_2 ,  C]$ are AMS w.r.t $\mu$. In this case, $\mu(\nu_1 \times \nu_2) \ll^a \mu (S\overline{\nu}_{1_\mu} \times S\overline{\nu}_{2_\mu})$.
\end{proposition}
Proof of Proposition \ref{PropositionCNSAMSChannelProduct} relies on the following lemma:
\begin{lemma}\label{LemmaProductOfAMSSources}
Let $\mu_1$, $\mu_2$, resp. $\eta_1$, $\eta_2$, be probabilities on a space $(\Omega, \mathcal{B}_\Omega)$, resp. $(\Lambda, \mathcal{B}_\Lambda)$.  Let $\mu_1 \times \eta_1$ and $\mu_2 \times \eta_2$ denote the  product probability on $(\Omega \times \Lambda, \mathcal{B}_{\Omega \times \Lambda})$. Then:
\begin{enumerate}
	\item $\mu_1 \times \eta_1 \ll \mu_2 \times \eta_2$ if and only if $\mu_1 \ll \mu_2$ and $\eta_1 \ll \eta_2$
	\item assuming that $\Omega$ and $\Lambda$ are sequence spaces and that $\mu_2$ and $\eta_2$ are stationary, $\mu_1 \times \eta_1 \ll^a \mu_2 \times \eta_2$ if and only if $\mu_1 \ll^a \mu_2$ and $\eta_1 \ll^a \eta_2$
\end{enumerate}
\end{lemma}
An obvious consequence is that $\mu_1 \times \eta_1$ is AMS if and only if $\mu_2$ and $\eta_2$ are AMS.
\begin{proof}[Proof of Lemma \ref{LemmaProductOfAMSSources}]
$\mu_1$ and $\eta_1$ (resp. $\mu_2$ and $\eta_2$) are marginals of $\mu_1 \times \eta_1$ (resp. of $\mu_2 \times \eta_2$). It is thus obvious that 
\begin{itemize}
\item $\mu_1 \times \eta_1 \ll \mu_2 \times \eta_2 \Rightarrow \mu_1 \ll \mu_2 \text{ and }\eta_1 \ll \eta_2$
\item $\mu_1 \times \eta_1 \ll^a \mu_2 \times \eta_2 \Rightarrow \mu_1 \ll^a \mu_2 \text{ and }\eta_1 \ll^a \eta_2$
\end{itemize}
Assume  that $\mu_1 \ll \mu_2$ and $\eta_1 \ll \eta_2$. Let $E$ be an event of the product $\sigma$-field $\mathcal{B}_{\Omega \times \Lambda}$ such that $(\mu_2 \times \eta_2)(E)=0$. $\mu_2 \times \eta_2$ is the product of the two probabilities $\mu_2$ and $\eta_2$ then:
$$
(\mu_2 \times \eta_2)(E)=\int \eta_2(E_x) d\mu_2
$$
where $E_x=\{y \in B^\mathcal{I} / (x,y) \in E\}$ denotes the section of $E$ at $x$. $(\mu_2 \times \eta_2)(E)=0$ then 
$$
\eta_2(E_x) = 0 \text{  }\mu_2\text{-a.e.}
$$
 $\mu_1 \ll \mu_2$ and $\eta_1 \ll \eta_2$, this implies that $\eta_1(E_x) = 0 \text{  }\mu_1\text{-a.e.}$
Then 
$$
\mu_1\times \eta_1(E)=\int \eta_1(E_x) d\mu_1=0
$$
 Hence $\mu_1 \times \eta_1 \ll \mu_2 \times \eta_2$.

Assume now that $\mu_2$ and $\eta_2$ are stationary and that $\mu_1 \ll^a \mu_2$ and $\eta_1 \ll^a \eta_2$.  From Lemma \ref{LemmaDominanceOnTailSigmaField} (p. \pageref{LemmaDominanceOnTailSigmaField}), this implies that $\mu_{1_\infty} \ll \mu_{2_\infty} $ and $\eta_{1_\infty} \ll \eta_{2_\infty} $. Thus, from above, $\mu_{1_\infty} \times \eta_{1_\infty} \ll \mu_{2_\infty} \times \eta_{2_\infty}$. Thanks to Lemma \ref{LemmaChannelRestriction} (p. \pageref{LemmaChannelRestriction}), for any probabilities $\mu$ and $\eta$,  $\mu_\infty \times \eta_\infty=(\mu \times \eta)_\infty$. Then, by Lemma \ref{LemmaDominanceOnTailSigmaField} (p. \pageref{LemmaDominanceOnTailSigmaField}) since $\mu_2 \times \eta_2$ is stationary,  $\mu_1 \times \eta_1 \ll^a \mu_2 \times \eta_2$.
\end{proof}
\begin{proof}[Proof of Proposition \ref{PropositionCNSAMSChannelProduct}]
If the channel product is AMS w.r.t $\mu$ then $\mu(\nu_1 \times \nu_2)$ is AMS then the marginals $\mu\nu_1$ and $\mu\nu_2$ are AMS. Then $[A, \nu_1 , B]$ and $[A,  \nu_2 ,  C]$ are AMS w.r.t $\mu$.

Assume that $[A, \nu_1 , B]$ and $[A,  \nu_2 ,  C]$ are AMS w.r.t $\mu$. Let  $\overline{\nu}_{1_\mu}$ and $\overline{\nu}_{2_\mu}$ respectively denote the stationary means of $\nu_1$ and $\nu_2$ w.r.t $\mu$. 

By Lemma \ref{LemmaChannelOfStationaryMeans} (p. \pageref{LemmaChannelOfStationaryMeans}) $S\overline{\nu}_{1_\mu}$ and $S\overline{\nu}_{2_\mu}$ are stationary w.r.t $\mu$. Then, by Proposition~\ref{PropositionCNSQuasiStationaryChannelProduct} (p. \pageref{PropositionCNSQuasiStationaryChannelProduct}), the channel product $S\overline{\nu}_{1_\mu} \times S\overline{\nu}_{2_\mu}$ is stationary w.r.t.  $\mu$, i.e., $\mu (S\overline{\nu}_{1_\mu} \times S\overline{\nu}_{2_\mu})$ is stationary.

Thanks to Proposition \ref{LemmaNecessarySufficientConditionAMSChannel} (p. \pageref{LemmaNecessarySufficientConditionAMSChannel}):
$$
\nu_1(x,.) \ll^a S\overline{\nu}_{1_\mu}(x,.) \text{  and  } \nu_2(x,.) \ll^a S\overline{\nu}_{2_\mu}(x,.) \text{  }\mu\text{-a.e.}
$$
Since, $S\overline{\nu}_{1_\mu}(x,.)$ and $S\overline{\nu}_{2_\mu}(x,.)$ are stationary, by Lemma \ref{LemmaDominanceOnTailSigmaField} (p. \pageref{LemmaDominanceOnTailSigmaField}), this implies that
$$
(\nu_1(x,.))_\infty \ll (S\overline{\nu}_{1_\mu}(x,.))_\infty \text{  and  } (\nu_2(x,.))_\infty \ll (S\overline{\nu}_{2_\mu}(x,.))_\infty \text{  }\mu\text{-a.e.}
$$
By Lemma \ref{LemmaProductOfAMSSources} (p. \pageref{LemmaProductOfAMSSources}):
$$
(\nu_1(x,.))_\infty \times (\nu_2(x,.))_\infty \ll (S\overline{\nu}_{1_\mu}(x,.))_\infty \times (S\overline{\nu}_{2_\mu}(x,.))_\infty \text{  }\mu\text{-a.e.}
$$
Thanks to Lemma \ref{LemmaChannelRestriction} (p. \pageref{LemmaChannelRestriction}), $(\nu_1(x,.))_\infty \times (\nu_2(x,.))_\infty = (\nu_1(x,.) \times \nu_2(x,.))_\infty$ and $(S\overline{\nu}_{1_\mu}(x,.))_\infty \times (S\overline{\nu}_{2_\mu}(x,.))_\infty=(S\overline{\nu}_{1_\mu}(x,.) \times S\overline{\nu}_{2_\mu}(x,.))_\infty$, then, by Lemma \ref{LemmaChannelDominanceEquivHookupDominance} (p. \pageref{LemmaChannelDominanceEquivHookupDominance}):
$$
\mu_\infty(\nu_1 \times \nu_2)_\infty \ll \mu_\infty(S\overline{\nu}_{1_\mu} \times S\overline{\nu}_{2_\mu})_\infty
$$
Thus, by Lemma \ref{LemmaDominanceOnTailSigmaField} (p. \pageref{LemmaDominanceOnTailSigmaField}), since $\mu(S\overline{\nu}_{1_\mu} \times S\overline{\nu}_{2_\mu})$ is stationary:
$$
\mu(\nu_1 \times \nu_2) \ll^a \mu(S\overline{\nu}_{1_\mu} \times S\overline{\nu}_{2_\mu})
$$
Hence $\mu(\nu_1 \times \nu_2)$ is AMS.
\end{proof}

\subsection{Recurrent Channel Product}

\begin{proposition}\label{PropositionProductRecurrent}
Let $[A,\nu_1,B]$ and $[A,\nu_2,C]$ be two channels and $\mu$ be the distribution of a recurrent  source $[A,X]$. The channel product $[A,\nu_1 \times \nu_2,B\times C]$ is recurrent w.r.t $\mu$ if and only if $[A,\nu_1,B]$ and $[A,\nu_2,C]$  are both recurrent w.r.t. $\mu$. 
\end{proposition}
\begin{proof}
The "only if " part is obvious. Assume that $[A,\nu_1,B]$ and  $[A,\nu_2,B]$ are recurrent w.r.t $\mu$. By Lemma \ref{LemmaRecurrentChannel2} (p. \pageref{LemmaRecurrentChannel2}), $\nu_1(x,.)$  and $\nu_2(x,.)$ are recurrent $\mu$-a.e. Then, $\forall G \in \mathcal{B}_{B^\mathcal{I}}$, $\forall H \in \mathcal{B}_{C^\mathcal{I}}$:
\begin{multline}
(\nu_1\times \nu_2)(x,G \times H \setminus \cup_{i\geq 1}(T_{BC}^{-i}G\times H)= \\
	(\nu_1\times \nu_2)(x,G \times H \cap  \bigcap_{i\geq 1}(T_{B}^{-i}G \times T_C^{-i} H^c) \\
	+(\nu_1\times \nu_2)(x,G \times H \cap  \bigcap_{i\geq 1}(T_{B}^{-i}G^c  \times T_C^{-i} H) \\
	+ (\nu_1\times \nu_2)(x,G \times H \cap  \bigcap_{i\geq 1}(T_{B}^{-i}G^c  \times T_C^{-i} H^c) \nonumber
\end{multline}
$$
(\nu_1\times \nu_2)(x,G \times H \cap  \bigcap_{i\geq 1}(T_{B}^{-i}G \times T_B^{-i} H^c) \leq \nu_2(x, H \setminus \cup_{i\geq 1} T_C^{-i} H) =0 \text{  }\mu \text{-a.e.}
$$
since $\nu_2(x,.)$ is recurrent $\mu$-a.e.
$$
(\nu_1\times \nu_2)(x,G \times H \cap  \bigcap_{i\geq 1}(T_{B}^{-i}G^c  \times T_C^{-i} H) \leq \nu_1(x, G \setminus \cup_{i\geq 1} T_B^{-i} G=0) \text{  }\mu \text{-a.e.}
$$
since $\nu_1(x,.)$ is recurrent $\mu$-a.e.
\begin{multline}
(\nu_1\times \nu_2)(x,G \times H \cap  \bigcap_{i\geq 1}(T_{B}^{-i}G^c  \times T_C^{-i} H^c) = \\
	\nu_1(x,G \setminus  \cup_{i\geq 1}(T_B^{-i}G)).\nu_2(x,H \setminus  \cup_{i\geq 1}(T_C^{-i}H)) =0 \text{  }\mu \text{-a.e.} \nonumber
\end{multline}
since $\nu_1(x,.)$ and $\nu_2(x,.)$ are recurrent $\mu$-a.e.

Then the product $(\nu_1\times \nu_2)(x,.)$ is recurrent on rectangles $\mu$-a.e. Then, by Lemma \ref{LemmaRecurrentChannel2} (p. \pageref{LemmaRecurrentChannel2}),  $(\nu_1\times \nu_2)(x,.)$ is recurrent  $\mu$-a.e. Thanks to Lemma \ref{LemmaRecurrentChannel2} (p. \pageref{LemmaRecurrentChannel2}), $[A,\nu_1 \times \nu_2,B\times C]$ is recurrent w.r.t $\mu$.
\end{proof}

\subsection{Output Weakly Mixing Channel Product}

It is shown that the product of output weakly mixing stationary channels is output weakly mixing. This property is of interest because it implies ergodicity for a stationary channel thanks to the following lemma which is Lemma 2.7 of \cite{Gray11}.
\begin{lemma}\label{LemmaOWMImpliesErgordicity}
Let $[A,\nu,B]$ be a channel stationary w.r.t. an ergodic and  stationary source distribution $\mu$. If $[A,\nu,B]$ is output weakly mixing then $[A,\nu,B]$ is ergodic w.r.t. $\mu$
\end{lemma}

\begin{proposition}\label{PropositionOWMChannelProduct}
Let $[A,\nu_1,B]$ and $[A,\nu_2,C]$ be   channels stationary w.r.t. a stationary source distribution $\mu$.  $[A,\nu_1 \times \nu_2,B\times C]$ is output weakly mixing w.r.t. $\mu$ if and only if $[A,\nu_1,B]$ and $[A,\nu_2,C]$ are output weakly mixing w.r.t. $\mu$.
\end{proposition}
The proof relies of the following lemmas.
\begin{lemma}\label{LemmaDensityConvergence}
Let $(a_n)$, $(b_n)$, $(c_n)$ and $(d_n)$ be bounded sequences of real numbers such that $\lim_{n \to \infty} \frac{1}{n} \sum_{i=0}^{n-1} |a_i - b_i |=\lim_{n \to \infty} \frac{1}{n} \sum_{i=0}^{n-1} |c_i - d_i |=0$. Then $\lim_{n \to \infty} \frac{1}{n} \sum_{i=0}^{n-1} |a_i.c_i - b_i.d_i |=0$. 
\end{lemma}
\begin{proof}
\begin{multline}
\frac{1}{n} \sum_{i=0}^{n-1} |a_i.c_i - b_i.d_i |= \frac{1}{n} \sum_{i=0}^{n-1} |a_i.c_i - b_i.c_i + b_i.c_i - b_i.d_i | \\
	\leq \frac{1}{n} \sum_{i=0}^{n-1} |c_i|.|a_i - b_i | + \frac{1}{n} \sum_{i=0}^{n-1} |b_i|.| c_i - d_i | \nonumber
\end{multline}
$b_i$ , $c_i$ are bounded then the RHS tends to 0 when $n \to \infty$.
\end{proof}

\begin{lemma}\label{Lemma81Gray09}
If a probability $\mu$ is AMS and if $\lim_{n \to \infty} \frac{1}{n} \sum_{i=0}^{n-1} |\mu(F \cap T_A^{-i}G) - \mu(F).\mu(T_A^{-i}G)| =0$ for any sets $F$ and $G$ of a generating field, then $\mu$ is weakly mixing.
\end{lemma}
\begin{proof}
This is Lemma 8.1 of \cite{Gray09} (page 246).
\end{proof}

\begin{proof}[Proof of Proposition \ref{PropositionOWMChannelProduct}]
Assume that $[A,\nu_1 \times \nu_2,B\times C]$ is output weakly mixing. Then, by definition, $\forall E, E' \in \mathcal{B}_{B^\mathcal{I}\times C^\mathcal{I}}$ 
\begin{multline}
\lim_{n \to \infty} \frac{1}{n} \sum_{i=0}^{n-1} \left | \nu_1\times \nu_2(x,T_{BC}^{-i}E \bigcap E') \right. \\
	- \left. \nu_1\times \nu_2(x,T_{BC}^{-i}E).  \nu_1\times \nu_2(x, E')\right | = 0 \nonumber
\end{multline}
Then choosing $E=G \times C^\mathcal{I}$ and $E'=G' \times C^\mathcal{I}$ where $G, G' \in \mathcal{B}_{B^\mathcal{I}}$ (resp. $E= B^\mathcal{I} \times H$ and $E'=B^\mathcal{I} \times H'$ where $H, H' \in \mathcal{B}_{C^\mathcal{I}}$) proves that $\nu_1$ (resp. $\nu_2$) is output weakly mixing.

Assume that $[A,\nu_1,B]$ and $[A,\nu_2,C]$ are output weakly mixing. Then (all statements are $\mu$-a.e.):
$$
\lim_{n \to \infty} \frac{1}{n} \sum_{i=0}^{n-1} \left | \nu_1(x,T_{B}^{-i}G \cap G')- \nu_1(x,T_{B}^{-i}G).  \nu_1(x, G')\right | =0
$$
$$
\lim_{n \to \infty} \frac{1}{n} \sum_{i=0}^{n-1} \left | \nu_2(x,T_{C}^{-i}H \cap H')- \nu_2(x,T_{C}^{-i}H).  \nu_2(x, H')\right | =0 
$$
Then, by Lemma \ref{LemmaDensityConvergence} (p. \pageref{LemmaDensityConvergence}), with 
$$
a_n=\nu_1(x,T_{B}^{-n}G \cap G'),\  b_n=\nu_1(x,T_{B}^{-n}G).  \nu_1(x, G')
$$
and 
$$
c_n=\nu_2(x,T_{C}^{-n}H \cap H'), \ d_n=\nu_2(x,T_{C}^{-i}H).  \nu_2(x, H')
$$
\begin{multline}
\lim_{n \to \infty} \frac{1}{n} \sum_{i=0}^{n-1} \left | \nu_1\times \nu_2(x,T_{BC}^{-i}(G\times H) \bigcap (G' \times H')) \right. \\
	- \left. \nu_1\times \nu_2(x,T_{BC}^{-i}(G\times H)).  \nu_1\times \nu_2(x, (G'\times H'))\right | = 0  \nonumber
\end{multline}
Thus, $\nu_1 \times \nu_2(x,.)$ is weakly mixing on the semi-algebra of the rectangles $\mathcal{B}_{B^\mathcal{I}}\times \mathcal{B}_{C^\mathcal{I}}$. 

Any element of the algebra generated by the rectangles is a finite union of disjoint rectangles. Let $E$ and $E'$ two elements of this algebra. Then
$$
E=\cup_{k=1}^K G_k \times H_k, \ E'=\cup_{l=1}^L G'_l \times H'_l
$$
The unions are disjoint then:
$$
(\nu_1\times \nu_2)(x, (T_{BC}^{-i} E) \cap E')= \sum_{k=1}^K \sum_{l=1}^L  (\nu_1\times \nu_2)(x, (T_{BC}^{-i} G_k \times H_k) \cap  G'_l \times H'_l)
$$
and
\begin{multline}
(\nu_1\times \nu_2)(x, T_{BC}^{-i} E). (\nu_1\times \nu_2)(x,E')= \\
	\sum_{k=1}^K \sum_{l=1}^L  (\nu_1\times \nu_2)(x, T_{BC}^{-i} G_k).(\nu_1\times \nu_2)(x, G'_l \times H'_l) \nonumber
\end{multline}
Then
\begin{multline}
\frac{1}{n}\sum_{i=0}^{n-1} |(\nu_1\times \nu_2)(x, (T_{BC}^{-i} E) \cap E') - (\nu_1\times \nu_2)(x, T_{BC}^{-i} E). (\nu_1\times \nu_2)(x,E') | \leq  \\
	\sum_{k=1}^K \sum_{l=1}^L \frac{1}{n}\sum_{i=0}^{n-1} |(\nu_1\times \nu_2)(x, (T_{BC}^{-i} G_k \times H_k) \cap  G'_l \times H'_l) - \\
		(\nu_1\times \nu_2)(x, T_{BC}^{-i} G_k).(\nu_1\times \nu_2)(x, G'_l \times H'_l)  |  \nonumber
\end{multline}
This obviously implies that $\nu_1 \times \nu_2(x,.)$ is weakly mixing on the algebra generated by  $\mathcal{B}_{B^\mathcal{I}}\times \mathcal{B}_{C^\mathcal{I}}$.

 Moreover, since $\nu_1\times \nu_2$ is a stationary channel (Proposition \ref{PropositionCNSQuasiStationaryChannelProduct}, p. \pageref{PropositionCNSQuasiStationaryChannelProduct}), $\nu_1 \times \nu_2(x,.)$ is AMS $\mu$-a.e. (Lemma \ref{LemmaStationaryChannelAMSkernel}, p. \pageref{LemmaStationaryChannelAMSkernel}), by Lemma \ref{Lemma81Gray09} (p. \pageref{Lemma81Gray09}),  $\nu_1 \times \nu_2(x,.)$ is  weakly mixing on the $\sigma$-field $\mathcal{B}_{B^\mathcal{I} \times C^\mathcal{I}}$.
\end{proof}

\subsection{Noisy computations}

The next corollary immediately follows  from Proposition \ref{PropositionCNSQuasiStationaryChannelProduct} (p. \pageref{PropositionCNSQuasiStationaryChannelProduct}) and Proposition \ref{PropositionCNSAMSChannelProduct} (p. \pageref{PropositionCNSAMSChannelProduct}).
\begin{corollary}\label{CorollaryQSAMSNoisyComputation}
A noisy computation $[A,f\times F, B\times C]$ is AMS (resp. stationary) w.r.t a stationary source distribution $\mu$ if and only if the deterministic channel $[A,f,B]$ and the random channel $[A,F,C]$ are both AMS (resp. stationary) w.r.t $\mu$.
\end{corollary}

The next corollary immediately follows  from Proposition \ref{PropositionProductRecurrent} (p. \pageref{PropositionProductRecurrent}).
\begin{corollary}\label{CorollaryProductRecurrent}
A noisy computation $[A,f\times F, B\times C]$ is recurrent w.r.t.  a recurrent  source distribution $\mu$ if and only if the deterministic channel $[A,f,B]$ and the random channel $[A,F,C]$  are both recurrent w.r.t. $\mu$. 
\end{corollary}

The next corollary  immediately follows  from Proposition \ref{PropositionOWMChannelProduct} (p. \pageref{PropositionOWMChannelProduct}).
\begin{corollary}\label{CorollaryOWMChannelProduct}
A  noisy computation $[A,f\times F, B\times C]$, stationary w.r.t. a stationary  source distribution $\mu$,  is output weakly mixing w.r.t $\mu$ if and only if  the deterministic channel $[A,f,B]$ and the random channel $[A,F,C]$ are  both output weakly mixing w.r.t. $\mu$.
\end{corollary}

\pagebreak

\section{Sufficient Conditions for Ergodicity of Channel Products and Noisy Computations}\label{SectionSufficientConditionsForErgodicity}

The objective is to identify conditions on channels to get an ergodic channel product. It is obvious that ergodicity of $\nu_1$ and $\nu_2$ is a necessary condition for $\nu_1 \times \nu_2$ to be ergodic. It cannot be a sufficient condition since the product of ergodic probabilities may not be ergodic.

 In this section, it is shown that,  if  the product $[A, S\overline{\nu}_{1_{\overline{\mu}}} \times S\overline{\nu}_{2_{\overline{\mu}}}, B\times C]$  is ergodic w.r.t to the stationary mean $\overline{\mu}$ of an AMS source distribution $\mu$, the  product of the two AMS (w.r.t. $\overline{\mu}$) channels $[A, \nu_1,B]$ and $[A, \nu_2,C]$ is ergodic  w.r.t. $\mu$.

It should be noticed that, thanks to Lemma \ref{LemmaErgodicAMSIsRecurrent} (p. \pageref{LemmaErgodicAMSIsRecurrent}), asymptotic mean stationarity and ergodicity jointly imply recurrence and, thus, ergodic AMS sources and channels are ergodic R-AMS. This leads to the fact that the AMS probabilities, channels and (by Proposition \ref{PropositionCSErgodicProduct} below ) channel products under consideration  are dominated and not only asymptotically dominated by their stationary means.

\subsection{Channel products}
\begin{proposition}\label{PropositionCSErgodicProduct}
Let $[A,X]$ be an AMS source with distribution $\mu$. Let $[A,\nu_1,B]$ and $[A,\nu_2,C]$ be two channels AMS w.r.t the stationary mean $\overline{\mu}$ of  $\mu$. Then $\nu_1$, $\nu_2$ and $\nu_1\times \nu_2$ are AMS w.r.t. $\mu$.

If, in addition,  $\mu$ is ergodic (or equivalently $\overline{\mu}$ is ergodic) and if  $[A,S\overline{\nu_1}_{\overline{\mu}} \times S\overline{\nu_2}_{\overline{\mu}},B\times C]$ is ergodic w.r.t. $\overline{\mu}$ then $[A, \nu_1 \times \nu_2, B\times C]$ is ergodic w.r.t. $\mu$ and w.r.t. $\overline{\mu}$ (where $\overline{\nu_1}_{\overline{\mu}}$ and $\overline{\nu_2}_{\overline{\mu}}$ are respectively the stationary means of $\nu_1$ and of $\nu_2$ w.r.t to $\overline{\mu}$).

In this case 
$$
\overline{(\nu_1 \times \nu_2)}_{\mu}=S\overline{\nu_1}_{\mu} \times S\overline{\nu_2}_{\mu} = \overline{(\nu_1 \times \nu_2)}_{\overline{\mu}}=S\overline{\nu_1}_{\overline{\mu}} \times S\overline{\nu_2}_{\overline{\mu}} \text{   }\overline{\mu}\text{-a.e.}
$$

If  $[A,S\overline{\nu_1}_{\overline{\mu}},B]$ and $[A,S\overline{\nu_2}_{\overline{\mu}}, C]$ are output weakly mixing w.r.t. $\overline{\mu}$ then $[A, \nu_1 \times \nu_2, B\times C]$ is ergodic w.r.t. $\mu$. 
\end{proposition}
\begin{proof}
$\mu$ is AMS then $\mu \ll^a \overline{\mu}$. By Lemma \ref{LemmaDominanceOnTailSigmaField} (p. \pageref{LemmaDominanceOnTailSigmaField}), $\mu_\infty \ll \overline{\mu}_\infty$. Then, by Lemma \ref{Lemma2Fontana} (p. \pageref{Lemma2Fontana}), $\mu_\infty \nu_{1_\infty} \ll \overline{\mu}_\infty\nu_{1_\infty}$. Since $\nu_1$ is AMS w.r.t. $\overline{\mu}$, and thanks to Lemma \ref{LemmaDominanceOnTailSigmaField} (p. \pageref{LemmaDominanceOnTailSigmaField}), $\overline{\mu}_\infty\nu_{1_\infty} \ll \overline{\mu}_\infty \overline{\nu}_{1_{\overline{\mu}_\infty}}$. Then $\mu_\infty \nu_{1_\infty} \ll \overline{\mu}_\infty \overline{\nu}_{1_{\overline{\mu}_\infty}}$. This implies that $\nu_1$ is AMS w.r.t. $\mu$. The same holds for $\nu_2$.

$\nu_1$ and $\nu_2$ are AMS w.r.t. the stationary probability $\overline{\mu}$ then, thanks to Proposition \ref{PropositionCNSAMSChannelProduct} (p. \pageref{PropositionCNSAMSChannelProduct}), $\nu_1 \times \nu_2$ is AMS w.r.t. $\overline{\mu}$ and (from above) w.r.t. $\mu$.

Assume now that, in addition, $\mu$ (or equivalently $\overline{\mu}$) is ergodic. Thanks to Lemma \ref{LemmaErgodicAMSIsRecurrent} (p. \pageref{LemmaErgodicAMSIsRecurrent}), $\mu$ is R-AMS and then $\mu \ll \overline{\mu}$. By Lemma \ref{Lemma2Fontana} (p. \pageref{Lemma2Fontana}), this implies that $\mu(\nu_1 \times \nu_2) \ll \overline{\mu}(\nu_1 \times \nu_2)$. 

Tanks to Proposition \ref{PropositionCNSAMSChannelProduct} (p. \pageref{PropositionCNSAMSChannelProduct}), 
$$
\mu(\nu_1 \times \nu_2)  \ll \overline{\mu}(\nu_1 \times \nu_2) \ll^a \overline{\mu}(S\overline{\nu}_{1_{\overline{\mu}}} \times S\overline{\nu}_{2_{\overline{\mu}}})
$$
Then, if $\overline{\mu}(S\overline{\nu}_{1_{\overline{\mu}}} \times S\overline{\nu}_{2_{\overline{\mu}}})$ is ergodic, $\mu(\nu_1 \times \nu_2)$ and $\overline{\mu}(\nu_1 \times \nu_2) $ are ergodic.

$\overline{\mu}(\nu_1 \times \nu_2) $ is ergodic then $\overline{\mu}\nu_1 $ is ergodic: $\nu_1$ is AMS and ergodic w.r.t $\overline{\mu}$. This implies that $\overline{\mu}\nu_1 \ll \overline{\mu}\overline{\nu}_{1_{\overline{\mu}}}$ i.e.
$$
\nu_1(x,.) \ll \overline{\nu}_{1_{\overline{\mu}}}(x,.) \text{   }\overline{\mu}\text{-a.e.}
$$
But $\overline{\nu}_{1_{\overline{\mu}}}(x,.) \ll S\overline{\nu}_{1_{\overline{\mu}}}(x,.) \text{   }\overline{\mu}\text{-a.e.}$. Then
$$
\nu_1(x,.) \ll S\overline{\nu}_{1_{\overline{\mu}}}(x,.) \text{   }\overline{\mu}\text{-a.e.}
$$
The same holds for $\nu_2$:
$$
\nu_2(x,.) \ll S\overline{\nu}_{2_{\overline{\mu}}}(x,.) \text{   }\overline{\mu}\text{-a.e.}
$$
Then, by Lemma \ref{LemmaProductOfAMSSources} (p. \pageref{LemmaProductOfAMSSources}), since $S\overline{\nu}_{2_{\overline{\mu}}}(x,.)$ and $S\overline{\nu}_{2_{\overline{\mu}}}(x,.)$ are stationary:
$$
\nu_1(x,.) \times \nu_2(x,.) \ll S\overline{\nu}_{1_{\overline{\mu}}}(x,.) \times S\overline{\nu}_{2_{\overline{\mu}}}(x,.)  \text{   }\overline{\mu}\text{-a.e.}
$$
This implies that:
$$
\overline{\mu}(\nu_1 \times \nu_2) \ll^a \overline{\mu}(S\overline{\nu}_{1_{\overline{\mu}}} \times S\overline{\nu}_{2_{\overline{\mu}}})
$$
and thus 
$$
\mu(\nu_1 \times \nu_2)  \ll \overline{\mu}(\nu_1 \times \nu_2) \ll \overline{\mu}(S\overline{\nu}_{1_{\overline{\mu}}} \times S\overline{\nu}_{2_{\overline{\mu}}})
$$

From \cite{Gray09}, an  AMS probability is  ergodic  if and only if its stationary mean is ergodic. Then $\overline{\mu}\overline{(\nu_1 \times \nu_2)}_{\overline{\mu}}$ and  $\overline{\mu}(S\overline{\nu_1}_{\overline{\mu}} \times S\overline{\nu_2}_{\overline{\mu}})$  are two ergodic stationary probabilities on the same space. They are thus equal or mutually singular. Since they both dominate the same probability (i.e. $\mu(\nu_1 \times \nu_2)$), they are equal. Then 
$$
\overline{(\nu_1 \times \nu_2)}_{\overline{\mu}}=S\overline{\nu_1}_{\overline{\mu}} \times S\overline{\nu_2}_{\overline{\mu}} \text{   }\overline{\mu}\text{-a.e.}
$$

Similarly, $\overline{\mu}\ \overline{\nu_1}_{\mu}$, $\overline{\mu}S\overline{\nu_1}_{\mu}$, $\overline{\mu}\ \overline{\nu_1}_{\overline{\mu}}$, $\overline{\mu}S\overline{\nu_1}_{\overline{\mu}}$ are stationary and ergodic and they dominate the same probability $\mu\nu_1$ (which is AMS ergodic and thus recurrent). They are thus equal. The same holds for $\nu_2$. Then 
$$
\overline{\nu_1}_{\mu}=S\overline{\nu_1}_{\mu}=\overline{\nu_1}_{\overline{\mu}}=S\overline{\nu_1}_{\overline{\mu}} \text{   }\overline{\mu}\text{-a.e.}
$$
and
$$
\overline{\nu_2}_{\mu}=S\overline{\nu_2}_{\mu}=\overline{\nu_2}_{\overline{\mu}}=S\overline{\nu_2}_{\overline{\mu}} \text{   }\overline{\mu}\text{-a.e.}
$$

Finally, if  $[A,S\overline{\nu_1}_{\overline{\mu}},B]$ and $[A,S\overline{\nu_2}_{\overline{\mu}}, C]$ are output weakly mixing w.r.t. $\overline{\mu}$ then, by Proposition \ref{PropositionOWMChannelProduct} (p. \pageref{PropositionOWMChannelProduct}),  $[A,S\overline{\nu_1}_{\overline{\mu}} \times S\overline{\nu_2}_{\overline{\mu}},B\times C]$ is output weakly mixing w.r.t. $\overline{\mu}$ and thus ergodic w.r.t. $\overline{\mu}$. This implies, from above, that $[A, \nu_1 \times \nu_2, B\times C]$ is ergodic w.r.t. $\mu$. 
\end{proof}

\subsection{Noisy computations}

The next corollary  immediately follows  from  Proposition \ref{PropositionCSErgodicProduct} (p. \pageref{PropositionCSErgodicProduct}).
\begin{corollary}\label{CorollaryErgodicNoisyComputation}
Let $[A,f\times F,B\times C]$ be a noisy computation AMS w.r.t.  the stationary mean $\overline{\mu}$ of an ergodic AMS source distribution $\mu$.

 If the channel product $[A,S\overline{f}_{\overline{\mu}} \times S\overline{F}_{\overline{\mu}},B\times C]$ is ergodic w.r.t. $\overline{\mu}$ then the noisy computation $[A,f\times F,B\times C]$ is ergodic w.r.t. $\mu$ and w.r.t $\overline{\mu}$.

If  $[A,S\overline{f}_{\overline{\mu}},B]$ and $[A,S\overline{F}_{\overline{\mu}},B]$ are output weakly mixing w.r.t. $\overline{\mu}$ then the noisy computation $[A,f\times F,B\times C]$ is ergodic w.r.t. $\mu$.
\end{corollary}

\pagebreak

\section{Noisy computations and computable functions} \label{SectionNoisyComputationsAndComputableFunctions}

In this section, a way to instantiate the noisy computation model to the case of computable functions is proposed. First, it is shown that  a function on finite length symbol sequences (thus which can be a computable function) $f$ induces a deterministic channel $f^*$ (equivalently a measurable function on infinite length sequences). This deterministic channel $f^*$ has the obvious property that, for any input sequence $x$, the output sequence $f^*(x)$ is made the concatenation of a finite length sequence and of an infinite sequence of blank symbols. 

Secondly, it is shown that any measurable function $g$ on infinite length sequences (i.e., deterministic channel) possesses  stability properties (AMS, mixing) with respect to sources which produce, almost surely,  input sequences $x$ such that $g(x)$ is made the concatenation of a finite length sequence and of an infinite sequence of blank symbols. This allows to conclude that the deterministic channel $f^*$, induced by a computable function $f$, is AMS and mixing w.r.t any AMS and ergodic source $\mu$.

Then, given an AMS and ergodic source distribution $\mu$, for any random channel $[A,F,C]$ AMS w.r.t. $\overline{\mu}$ such that $[A, \overline{F}_{\overline{\mu}}, C]$ is output weakly mixing w.r.t. $\overline{\mu}$, the noisy computation $[A, f^*\times F, B\times C]$ is AMS and ergodic w.r.t $\mu$.
In other words, for a computable function $f$, it is sufficient to focus on stability properties of the noisy realization $F$ under concern in order to apply the coding results given above.

This section ends with very simple examples.

\subsection{Linking the noisy computation model to a formal computational model and its noisy version} \label{SubSectionInducedDeterministicChannel}

When dealing with a function $f$ computable according to a formal computational model such as Turing machines, Moore automata, \ldots, it will be necessary to associate to $f$, which is a function from a set of finite length sequences $A^+$ to a set of finite length sequences $B^+$, a function $f^*$ from $A^\mathbb{N}$ to $B^\mathbb{N}$ defining a deterministic channel $[A,f^*,B]$. A finite length sequence of $u \in A^+$, on which $f$ is defined, will be associated to an infinite length sequence $x$ by concatenating an infinite sequence of blank symbols to $u$. This association is obviously a measurable function $\phi$. This can be viewed as  writing  a finite length string on a semi-infinite tape with cells initially containing blank symbols. Reading the output of such a computing formal device can be viewed as extracting the prefix $u$ of an infinite length sequence made of $u$ concatenated to an infinite length sequence of blank symbols. This can also be associated to a measurable function $\psi$.

Given two measurable finite alphabets $A$ and $B$, let $(A^+,\mathcal{P}(A^+))$ and $(B^+,\mathcal{P}(B^+))$ be the corresponding measurable spaces of finite length sequences,  $\mathcal{P}(A^+)$ being the set of subsets of $A^+$. If $u$ is a finite string of symbols and $v$ is a finite or countable string of symbols, $uv$ denotes the concatenation of $u$ and $v$, $|u|$ is the length of $u$. If $s$ is a symbol, $s^\infty$ denotes the countable string $s\cdots s \cdots$. $A^+$ and $B^+$ are assumed to contain an empty string denoted $\epsilon$ for both sets ($|\epsilon|=0$).

Let $L$ be a subset of $A^+$ such that $\forall x \in L, x_{|x|} \neq \alpha$ and $f: L \to B^+$ be a function. By convention, $f(\epsilon)=\epsilon$. Let $\alpha \in A$ called the blank symbol of $A$ and $\beta \in B$ called the blank symbol of $B$.

Let $\psi^A_\alpha : A^\mathbb{N} \to A^+$ be the function such that, $\forall x=x_1\cdots x_n\cdots \in A^\mathbb{N}$,
$$
\psi^A_\alpha(x)= \left\{ 
		\begin{array}{ll}
			x_1\cdots x_n 	&	\text{if }x=x_1\cdots x_n \alpha^\infty \text{ and } x_1\cdots x_n  \in L \setminus \{\epsilon\}\\
			\epsilon			& 	\text{else }
		\end{array}
	\right.
$$

Let $\phi^A_\alpha : A^+ \to A^\mathbb{N}$ be the function such that, $\forall x=x_1\cdots x_{|x|}\in A^+$,

$$
\phi^A_\alpha(x)=x\alpha^\infty
$$

Let $\psi^B_\beta : B^\mathbb{N} \to B^+$ be the function such that, $\forall y=y_1\cdots y_n\cdots \in B^\mathbb{N}$,
$$
\psi^B_\beta(y)= \left\{ 
		\begin{array}{ll}
			y_1\cdots y_n 	&	\text{if }y=y_1\cdots y_n \beta^\infty , y_n \neq \beta\\
			\epsilon			& 	\text{else }
		\end{array}
	\right.
$$

Let $\phi^B_\beta : B^+ \to B^\mathbb{N}$ be the function such that, $\forall y=y_1\cdots y_{|y|}\in B^+$,
$$
\phi^B_\beta(y)=y\beta^\infty
$$

\begin{lemma}\label{LemmaFiniteToInfiniteFunction}
The function $f^*= \phi^B_\beta \circ f \circ \psi^A_\alpha$ is $(\mathcal{B}_{A^\mathbb{N}},\mathcal{B}_{B^\mathbb{N}})$-measurable and thus induces a deterministic channel $[A,f^*,B]$. Moreover:
\begin{itemize}
	\item if $x= x_0 \ldots x_n \alpha^\infty$ where $x_0 \ldots x_n \in L$ then $f^*(x) = f(x_0 \ldots x_n)\beta^\infty$ , else $f^*(x)=\beta^\infty$
	\item for any $x_0 \ldots x_n \in L$, $f(x_0 \ldots x_n)=\psi^B_\beta \circ f^* \circ \phi^A_\alpha (x_0 \ldots x_n)$
	\item for any $x_0 \ldots x_n \notin L$, $\psi^B_\beta \circ f^* \circ \phi^A_\alpha (x_0 \ldots x_n)=\epsilon$
	\item $\forall x \in A^\mathbb{N}, \exists N_x \text{ such that  } \forall n > N_x, (f^*(x))_n=\beta$
\end{itemize}
\end{lemma}

\begin{figure}[htbp]
\xymatrix{
(A^+,\mathcal{P}(A^+))  \ar [r]^f \ar@<1ex> [d]^{\phi^A_\alpha}				&	(B^+,\mathcal{P}(B^+)) \ar@<1ex> [d]^{\phi^B_\beta} \\
(A^\mathbb{N}, \mathcal{B}_{A^\mathbb{N}}) \ar [r]^{f^*} \ar@<1ex> [u]^{\psi^A_\alpha}	& 	(B^\mathbb{N}, \mathcal{B}_{B^\mathbb{N}}) \ar@<1ex> [u]^{\psi^B_\beta}
}
\caption{A function $f$ on finite sequences and the induced deterministic channel $[A,f^*,B]$}
\end{figure}
\begin{proof}
Obviously a function $f:L \to B^+$ is $(\mathcal{P}(A^+),\mathcal{P}(B^+))$-measurable, $\phi^A_\alpha$ is $(\mathcal{P}(A^+),\mathcal{B}_{A^\mathbb{N}})$-measurable and $\phi^B_\beta$ is $(\mathcal{P}(B^+),\mathcal{B}_{B^\mathbb{N}})$-measurable.

 It remains to prove that $\psi^A_\alpha$ is $(\mathcal{B}_{A^\mathbb{N}}, \mathcal{P}(A^+),)$-measurable (the proof also applies to $\psi^B_\beta$).

Let $x_1\cdots x_n \in A^+$. Then
$$
\psi_\alpha^{-1}(x_1\cdots x_n)=\{ x_1\cdots x_n \alpha^\infty \}=c_1^n(x_1\cdots x_n) \bigcap \cap_{i>n}c_i^i(\alpha)
$$
or
$$
\psi_\alpha^{-1}(x_1\cdots x_n)= \emptyset
$$
where $c_1^n(x_1\cdots x_n)=\{a \in A^\mathbb{N}, a_i=x_i, i=1,\cdots, n\}$ is the thin cylinder induced by $(x_1\cdots x_n)$.

$\psi_\alpha^{-1}(x_1\cdots x_n)$ is thus a countable intersection of thin cylinders which are measurable. Then $\psi_\alpha^{-1}(x_1\cdots x_n)$ is measurable. Any subset $E$ of $A^+$ is countable then, $\forall E \neq \{ \alpha \}$, $\psi_\alpha^{-1}(E)=\cup_{x \in E} \psi^{-1}(x)$ is measurable. 

The last statements of the proposition are immediate from the definition of the functions $\phi^A_\alpha$, $\phi^B_\beta$, $\psi^A_\alpha$ and  $\psi^B_\beta$
\end{proof}

Let $D_\beta(f^*)$ be the set:
$$
D_\beta(f^*)=\{ x\in A^{\mathbb{N}} / \exists N_x \in \mathbb{N} \text{ such that } \forall n> N_x, (f^*(x))_n = \beta \}
$$
From Lemma \ref{LemmaFiniteToInfiniteFunction} (p. \pageref{LemmaFiniteToInfiniteFunction}), $D_\beta(f^*)=A^{\mathbb{N}}$. Thus, for any probability $\mu$ on $(A^{\mathbb{N}},\mathcal{B}_{A^{\mathbb{N}}})$, $\mu(D_\beta(f^*))=1$.

Associated to a formal computational model (e.g., circuits made of binary gates from a complete basis), a noisy computational model can be considered in order to capture unreliability of an actual implementation of a function (e.g., circuits made of noisy binary gates, \cite{VonNeumann56}). It is also relevant to consider a noisy implementation of a computable function as a channel taking infinite length sequences as inputs and producing infinite length sequences as outputs. An input finite sequence $u$ can be considered as an infinite one i.e.  $u\alpha^\infty$ (see above). It the noisy device halts on an input, the finite output sequence $v$ can also be considered as the infinite one $v\beta^\infty$. If the noisy device doest not halt on the input $u\alpha^\infty$ then its output is an infinite sequence $z$. 

The choice of a peculiar noisy computational model is in fact the choice of probabilistic assumptions intended to capture the nature of faults of a technology in use. A noisy implementation of a function, in such a noisy computational model, can also be described as a random channel whose probabilistic properties are a consequence of the probabilistic assumptions at the base of the noisy computational model.  For example, in the case of circuits based on noisy boolean gates, if the gates behave independently and have constant probability of failure, a circuit made of such noisy gates will behave as a stationary memoryless random channel.

To summarize, the realization $F$ of a computable function $f$  built according to a noisy computational model induces a noisy computation $[A, f^*\times F , B\times C]$ where the characteristics of the random channel $[A, F, C]$ depends on the noisy computational model. Moreover, for any input source with distribution $\mu$, $\mu(D_\beta(f^*))=1$.

\subsection{Stability properties of a deterministic channel $f^*$ w.r.t a source $\mu$ such that $\mu(D_\beta(f^*))=1$}

It is shown below that, if $f^*$ is a $(\mathcal{B}_{A^\mathbb{N}}, \mathcal{B}_{B^\mathbb{N}})$-measurable function, the corresponding deterministic channel $[A,f^*,B]$ possesses some stability properties w.r.t probabilities $\mu$ such that $\mu(D_\beta(f^*))=1$ ($\beta$ is a fixed symbol of $B$).

In the following, the shift $T_B$ is assumed non-invertible: $\forall G \in \mathcal{B}_{B^\mathbb{N}}$, $T_B^{-1}(G) = B \times G$.

\begin{lemma}\label{LemmaStabilityofDeterministicChannel}
 Let $f^*:A^{\mathbb{N}} \rightarrow B^{\mathbb{N}}$ be a $(\mathcal{B}_{A^{\mathbb{N}}},\mathcal{B}_{B^{\mathbb{N}}})$-measurable function such that  $D_\beta(f^*)$ is non-empty. Let $\{ f^*(x,.), x\in A^{\mathbb{N}}\}$ denote the probability kernel defining the deterministic channel $[A,f^*,B]$. Then
\begin{enumerate}
	\item \label{LemmaStabilityStationarity} for any $x\in D_\beta(f^*)$, $(f^*(x,.))_\infty$ is a stationary probability on  the space $(B^{\mathbb{N}}, (\mathcal{B}_{B^{\mathbb{N}}})_\infty, T_B)$
	\item \label{LemmaStabilityAMSwrtStationary} if $\mu$ is a stationary probability on  $(A^{\mathbb{N}}, \mathcal{B}_{A^{\mathbb{N}}},T_A)$ such that $\mu(D_\beta(f^*))=1$ then the channel $[A,f^*,B]$ is AMS w.r.t. $\mu$
	\item \label{LemmaStabilityAMSwrtAMS} if $\mu$ is a R-AMS probability on   $(A^{\mathbb{N}}, \mathcal{B}_{A^{\mathbb{N}}},T_A)$ such that $\overline{\mu}(D_\beta(f^*))=1$ then the channel $[A,f^*,B]$ is AMS w.r.t. $\mu$
	\item \label{LemmaStabilityRecurrence} if $\mu$ is a recurrent probability on   $(A^{\mathbb{N}}, \mathcal{B}_{A^{\mathbb{N}}})$ such that $\mu(D_\beta(f^*))=1$ then the channel $[A,f^*,B]$ is recurrent w.r.t. $\mu$
	\item \label{LemmaStabilityRAMS} if $\mu$ is an R-AMS probability on   $(A^{\mathbb{N}}, \mathcal{B}_{A^{\mathbb{N}}})$ such that $\overline{\mu}(D_\beta(f^*))=1$ then the channel $[A,f^*,B]$ is R-AMS w.r.t. $\mu$
	\item \label{LemmaStabilityAMS} if $\mu$ is an AMS probability on   $(A^{\mathbb{N}}, \mathcal{B}_{A^{\mathbb{N}}},T_A)$ such that $\overline{\mu}(D_\beta(f^*))=1$ then the channel $[A,f^*,B]$ is AMS w.r.t. $\mu$
\end{enumerate}
\end{lemma}
\begin{proof}
\begin{enumerate}
	\item Let  $x\in D_\beta(f)$. Then $f^*(x)=y_1\cdots y_{N_x} \beta \cdots \beta \cdots$. 

	Let $G\in (\mathcal{B}_{B^{\mathbb{N}}})_\infty$. Then, $\forall i \geq 0$, there exists $G_i \in \mathcal{B}_{B^{\mathbb{N}}}$ such that $G=T_B^{-i} G_i$. Since $T_B$ is the non-invertible shift, $T_B^{-1}G_i=B\times G_i$. Then $G=B^i \times G_i$. 

	In particular, there exists $G_{N_x} \in \mathcal{B}_{B^{\mathbb{N}}}$ such that $G=B^{N_x}\times G_{N_x}$ and $T_B^{-1}G=B^{N_x+1}\times G_{N_x}$. In fact $G_{N_x+1}=B\times G_{N_x}$ i.e. $G_{N_x+1}=T_B^{-1} G_{N_x}$.

	$f^*$ defines a deterministic channel then $f^*(x,G)=1$ if and only if $f^*(x)=y_1\cdots y_{N_x} \beta \cdots \beta \cdots \in G=B^{N_x}\times G_{N_x}$. $y_1\cdots y_{N_x} \in B^{N_x}$ is obviously always true, then $f^*(x,G)=1$ if and only if $\beta^\infty= \beta \cdots \beta \cdots \in  G_{N_x}$.	Moreover $\beta^\infty \in G_{N_x} \Rightarrow \beta^\infty \in B\times G_{N_x}$.Thus $f^*(x,G)=1$ implies $\beta^\infty \in B^{N_x+1}\times G_{N_x}=G_{N_x+1}$ which implies that $f^*(x) \in T_B^{-1}G$ then $f^*(x,T_B^{-1}G)=1$.

	If $f^*(x,G)=0$ then $f^*(x)=y_1\cdots y_{N_x} \beta \cdots \beta \cdots \notin G=B^{N_x}G_{N_x}$. This implies that $\beta^\infty \notin G_{N_x}$ and thus $\beta^\infty \notin B\times G_{N_x}$. Hence $f^*(x) \notin T_B^{-1}G$ which is equivalent to $f^*(x,T_B^{-1}G)=0$.

	Thus, in any case,$\forall G\in (\mathcal{B}_{B^{\mathbb{N}}})_\infty$, $f^*(x,G)=f^*(x,T_B^{-1}G)$.

	\item by \ref{LemmaStabilityStationarity}), $\forall x\in D_\beta(f^*)$, $f^*(x,.)_\infty$ is stationary thus R-AMS. Then, there exits a stationnary probability $\eta$ such that $f^*(x,.)_\infty \ll \eta_\infty$ (e.g. let $\eta'$ be a stationary probability on $(B^\mathbb{N}, \mathcal{B}_{B^\mathbb{N}})$, set $\eta=\eta'$ on  $\mathcal{B}_{B^\mathbb{N}} \setminus (\mathcal{B}_{B^\mathbb{N}})_\infty$ and $\eta =f^*(x,.)_\infty$ on $(\mathcal{B}_{B^\mathbb{N}})_\infty$) . Then, by Lemma \ref{LemmaDominanceOnTailSigmaField} (p. \pageref{LemmaDominanceOnTailSigmaField}), $f^*(x,.)$ is AMS. Since $\mu(D_\beta(f^*))=1$, $f^*(x,.)$ is AMS $\mu$-a.e. Thanks to Lemma \ref{LemmaNecessarySufficientConditionAMSChannel2} (p. \pageref{LemmaNecessarySufficientConditionAMSChannel}), this implies that $f^*$ is AMS w.r.t. $\mu$.

	\item $\mu$ is R-AMS thus $\mu$ is dominated by its stationary mean: $\mu \ll \overline{\mu}$. Then, by Lemma \ref{Lemma2Fontana} (p. \pageref{Lemma2Fontana}), $\mu f^* \ll \overline{\mu}f^*$. By \ref{LemmaStabilityAMSwrtStationary}),  $f^*$ is AMS w.r.t $\overline{\mu}$: $\overline{\mu}f^* \ll^a \overline{\mu}\overline{f^*}_{\overline{\mu}}$. Hence $\mu f^*   \ll^a \overline{\mu}\overline{f^*}_{\overline{\mu}}$ i.e. $\mu f^*$ is AMS. $f^*$ is thus AMS w.r.t. $\mu$.

	\item  $\forall G \in \mathcal{B}_{B^\mathbb{N}}$, $T_B^{-i}G=B^i \times G$. Let $ G \in \mathcal{B}_{B^\mathbb{N}}$ be such that $T_B^{-1}G \subset G$. Then $\forall i$, $T_B^{-i}G=B^i \times G \subset G$. This implies that $\forall i$ there exists $G_i \in \mathcal{B}_{B^\mathbb{N}}$ such that $G=B^i \times G_i$, then $G \in (\mathcal{B}_{B^\mathbb{N}})_\infty$. Thus, by \ref{LemmaStabilityStationarity}), $f^*(x, G)=f^*(x,T_B^{-1}G)$. Then $f^*(x,.)$ is incompressible, equivalently recurrent.

	\item obvious from \ref{LemmaStabilityAMSwrtAMS}) and \ref{LemmaStabilityRecurrence})

	\item $\mu$ is AMS then $\mu \ll^a \overline{\mu}$. Then by Lemma \ref{LemmaDominanceOnTailSigmaField} (p. \pageref{LemmaDominanceOnTailSigmaField}), $\mu_\infty \ll \overline{\mu}_\infty$. By Lemma \ref{Lemma2Fontana} (p. \pageref{Lemma2Fontana}), $\mu_\infty f^*_\infty \ll \overline{\mu}_\infty f^*_\infty$. By \ref{LemmaStabilityAMSwrtStationary}), $\overline{\mu} f^*$ is AMS thus $\overline{\mu}_\infty f^*_\infty \ll \overline{\mu}_\infty \overline{f^*}_{\overline{\mu}_\infty}$. This implies that   $\mu_\infty f^*_\infty \ll  \overline{\mu}_\infty \overline{f^*}_{\overline{\mu}_\infty}$ and thus  (since is $\overline{\mu} \overline{f^*}_{\overline{\mu}}$ is stationary), by Lemma \ref{LemmaDominanceOnTailSigmaField} (p. \pageref{LemmaDominanceOnTailSigmaField}),  $\mu f^* \ll^a \overline{\mu} \overline{f^*}_{\overline{\mu}}$. Hence, $f^*$ is AMS w.r.t. $\mu$
\end{enumerate}

\end{proof}

\begin{lemma}\label{LemmaQSMeanIsOWMStationarySource}
Let $\mu$ be the distribution of a stationary and ergodic  source $[A,X]$.  Let $f^*:A^{\mathbb{N}} \rightarrow B^{\mathbb{N}}$ be a $(\mathcal{B}_{A^{\mathbb{N}}},\mathcal{B}_{B^{\mathbb{N}}})$-measurable function  such that $\mu(D_\beta(f^*))=1$. Since, by Lemma \ref{LemmaStabilityofDeterministicChannel} (p. \pageref{LemmaStabilityofDeterministicChannel}), the stationary mean $[A, \overline{f^*}_{\mu}, B]$ exists, then
\begin{enumerate}
	\item $\forall G \in \mathcal{B}_{B^\mathbb{N}}$, $\overline{f^*}_{\mu}(x,G) = 1$  if and only if $\beta^\infty \in G$ and $\overline{f^*}_{\mu}(x,G) = 0$  if and only if $\beta^\infty \notin G$, $\mu$-a.e.
	\item $\forall G, G' \in \mathcal{B}_{B^\mathbb{N}}$, $\overline{f^*}_{\mu}(x, G\cap G') = \overline{f^*}_{\mu}(x, G).\overline{f^*}_{\mu}(x, G')$ $\mu$-a.e.
	\item $\overline{f^*}_{\mu}(x, .)$ is stationary $\mu$-a.e. and  consequently $\overline{f^*}_{\mu}(x, .)=S\overline{f^*}_{\mu}(x, .)$ $\mu$-a.e. and $[A,\overline{f^*}_{\mu},B]=[A,S\overline{f^*}_{\mu},B]$
	\item  $\overline{f^*}_{\mu}(x, .)$ is strongly  mixing $\mu$-a.e.
\end{enumerate}
\end{lemma}
\begin{proof}
From Lemma \ref{LemmaStabilityofDeterministicChannel} (p. \pageref{LemmaStabilityofDeterministicChannel}),  the deterministic channel $[A,f^*,B]$ is R-AMS w.r.t. $\mu$. Then the channels $[A,\overline{f^*}_{\mu},B]$ and $[A,S\overline{f^*}_{\mu},B]$ exist.
\begin{enumerate}
	\item 

Since $\overline{\mu f^*}=\mu \overline{f^*}_\mu$, $\forall E \in \mathcal{B}_{A^\mathbb{N}}$ and $\forall G \in \mathcal{B}_{B^\mathbb{N}}$
\begin{equation}\label{Eq1}
	\int_E \overline{f^*}_\mu(x,G) d\mu = \lim_{n \to \infty} \frac{1}{n} \sum_{i=0}^{n-1} \int f^*(x, T_B^{-i}G) 1_E(T_A^i x) d\mu
\end{equation}
$\forall i \in \mathbb{N}$, $f^*(x, T_B^{-i}G) = 1_G(T_B^i f^*(x))$ and  $\forall x \in D_\beta(f^*)$, $f^*(x)=y_1\cdots y_{N_x} \beta^\infty$. Then $\forall x\in D_\beta(f^*)$ and $\forall i>N_x$, 
$$
f^*(x, T_B^{-i}G) = 1 \text{   if and only if  }\beta^\infty \in G
$$
and 
$$
f^*(x, T_B^{-i}G) = 0 \text{  if and only if  } \beta^\infty \notin G
$$

Moreover 
$$
\forall i,j > N_x, f^*(x, T_B^{-i}G) = f^*(x, T_B^{-j}G)
$$ 

For $n> N_x+1$, let 
$$
\Delta_n(x,E,G)= \frac{1}{n} \sum_{i=0}^{n-1} f^*(x,T_B^{-i}G) 1_E(T_A^i x)
$$
 Then
\begin{multline}
\Delta_n(x,E,G)=\\
	\frac{1}{n} \sum_{i=0}^{N_x} f^*(x,T_B^{-i}G) 1_E(T_A^i x) +  \frac{1}{n} \sum_{i=N_x+1}^{n-1} f^*(x,T_B^{-i}G) 1_E(T_A^i x) \nonumber
\end{multline}

$\forall x\in D_\beta(f^*)$, $\forall n>N_x$ and $\forall G$ such that $\beta^\infty \notin G$
$$
\Delta_n(x,E,G)= \frac{1}{n} \sum_{i=0}^{N_x} f^*(x,T_B^{-i}G) 1_E(T_A^i x) 
$$
Thus 
$$
\forall x\in D_\beta(f^*), \forall G \text{  such that  }\beta^\infty \notin G,\text{      }\lim_{n \to \infty} \Delta_n(x,E,G)=0
$$

$\forall x\in D_\beta(f^*)$, $\forall n>N_x$ and $\forall G$ such that $\beta^\infty \in G$
$$
\Delta_n(x,E,G)= \frac{1}{n} \sum_{i=0}^{N_x} f^*(x,T_B^{-i}G) 1_E(T_A^i x) +  \frac{1}{n} \sum_{i=N_x+1}^{n-1} 1_E(T_A^i x) 
$$
Thus 
$$
\lim_{n \to \infty} \Delta_n(x,E,G)=\lim_{n \to \infty}\frac{1}{n} \sum_{i=N_x+1}^{n-1} 1_E(T_A^i x) = \lim_{n \to \infty}\frac{1}{n} \sum_{i=0}^{n-1} 1_E(T_A^i x)
$$
 Since $\mu$ is stationary and ergodic:
$$
\lim_{n \to \infty}\frac{1}{n} \sum_{i=0}^{n-1} 1_E(T_A^i x)  = \mu(E) \text{  }\mu\text{-a.e.}
$$
Thus, since  $\mu( D_\beta(f^*))=1$, $\forall G$ such that  $\beta^\infty \in G$,
$$
\lim_{n \to \infty} \Delta_n(x,E,G)=\mu(E)  \text{  }\mu\text{-a.e.}
$$

Thanks to the Bounded Convergence Theorem and (\ref{Eq1})
$$
\int_E \overline{f^*}_\mu(x,G) d\mu = \int \lim_{n \to \infty} \Delta_n(x,E,G) d\mu
$$
$\mu(D_\beta(f^*))=1$, then
$$
\int_E \overline{f^*}_\mu(x,G) d\mu = \int_{D_\beta(f^*))} \lim_{n \to \infty} \Delta_n(x,E,G) d\mu
$$
Thus, if $\beta^\infty \in G$, $\forall E \in \mathcal{B}_{A^\mathbb{N}}$
$$
\int_E \overline{f^*}_\mu(x,G) d\mu = \int_{D_\beta(f^*))}  \mu(E) d\mu = \mu(E)\mu(D_\beta(f^*))=\mu(E)
$$
Then $ \overline{f^*}_\mu(x,G) =1 $ $\mu$-a.e.

If  $\beta^\infty \notin G$, $\forall E \in \mathcal{B}_{A^\mathbb{N}}$
$$
\int_E \overline{f^*}_\mu(x,G) d\mu  = 0
$$
Then $ \overline{f^*}_\mu(x,G) =0 $ $\mu$-a.e.

	\item obvious from 1).

	\item Let $x \in D_\beta(f^*)$. Let $G \in \mathcal{B}_{A^\mathbb{N}}$. Since $T_B^{-1} G = B \times G$, $\beta^\infty \in G \Leftrightarrow \beta^\infty \in T_B^{-1} G$. Thus, from  1)
$$
\overline{f^*}_\mu(x,G)=1 \Leftrightarrow \overline{f^*}_\mu(x,T_B^{-1} G)=1
$$
And $\beta^\infty \notin G \Leftrightarrow \beta^\infty \notin T_B^{-1} G$. Thus
$$
\overline{f^*}_\mu(x,G)=0 \Leftrightarrow \overline{f^*}_\mu(x,T_B^{-1} G)=0
$$
Then, $\forall x \in D_\beta(f^*)$, $\overline{f^*}_\mu(x,.)$ is stationary and thus is equal to its stationary mean $S\overline{f^*}_\mu(x,.)$. Since $\mu(D_\beta(f^*))=1$, $\overline{f^*}_\mu(x,.)=S\overline{f^*}_\mu(x,.)$ $\mu$-a.e.

	\item 
By 2), $|\overline{f^*}_\mu(x, G \cap T_B^{-n}G')-\overline{f^*}_\mu(x, G).\overline{f^*}_\mu(x, T_B^{-n}G')|=0$ $\mu$-a.e. Then, obviously, $\overline{f^*}_\mu(x,.)$ is strongly mixing, $\mu$-a.e.
\end{enumerate}

\end{proof}

\begin{proposition}\label{PropositionQSMeanIsOWMRAMSSource}
Let $\mu$ be the distribution of an AMS and ergodic  source $[A,X]$.  Let $f^*:A^{\mathbb{N}} \rightarrow B^{\mathbb{N}}$ be a $(\mathcal{B}_{A^{\mathbb{N}}},\mathcal{B}_{B^{\mathbb{N}}})$-measurable function  such that $\overline{\mu}(D_\beta(f^*))=1$. Then
\begin{itemize}
	\item the channel $[A,f^*,B]$ is R-AMS and ergodic w.r.t. $\mu$
	\item $[A,\overline{f^*}_{\mu},B]=[A,\overline{f^*}_{\overline{\mu}},B]=[A,S\overline{f^*}_{\mu},B]=[A,S\overline{f^*}_{\overline{\mu}},B]$ mod. $\overline{\mu}$ and mod. $\mu$
	\item $\overline{f^*}_{\mu}(x, .)$ is an R-AMS and strongly  mixing probability on $(B^\mathbb{N},\mathcal{B}_{B^\mathbb{N}})$ $\overline{\mu}$-a.e. and $\mu$-a.e.
\end{itemize}
\end{proposition}
\begin{proof}
$\mu$ is AMS and ergodic thus, by Lemma \ref{LemmaErgodicAMSIsRecurrent} (p. \pageref{LemmaErgodicAMSIsRecurrent}), also recurrent. Since $\overline{\mu}(D_\beta(f^*))=1$, thanks to Lemma \ref{LemmaStabilityofDeterministicChannel} (p. \pageref{LemmaStabilityofDeterministicChannel}), $f^*$ is R-AMS w.r.t. $\mu$ and w.r.t. $\overline{\mu}$.

$\overline{\mu}$ is stationary and ergodic, then, by Lemma \ref{LemmaQSMeanIsOWMStationarySource} (p. \pageref{LemmaQSMeanIsOWMStationarySource}), $\overline{f^*}_{\overline{\mu}}(x,.)$ is strongly (then weakly) mixing w.r.t. $\overline{\mu}$. By Lemma \ref{LemmaOWMImpliesErgordicity} (p. \pageref{LemmaOWMImpliesErgordicity}), this implies that $\overline{\mu}\overline{f^*}_{\overline{\mu}}$ is ergodic.

Meanwhile $\mu \ll \overline{\mu}$. Then, by Lemma \ref{LemmaChannelDominanceEquivHookupDominance} (p. \pageref{LemmaChannelDominanceEquivHookupDominance}), $\mu f^* \ll \overline{\mu} f^*$. Since $f^*$ is R-AMS w.r.t. $\overline{\mu}$, $\overline{\mu} f^* \ll \overline{\mu}\overline{f^*}_{\overline{\mu}}$. This implies that
$$
\mu f^* \ll \overline{\mu}\overline{f^*}_{\overline{\mu}}
$$
Then $\mu f^*$ is ergodic i.e. $f^*$ is ergodic w.r.t. $\mu$.

From Proposition \ref{PropositionCSErgodicProduct} (p. \pageref{PropositionCSErgodicProduct}) and its proof, $\overline{f^*}_\mu=S\overline{f^*}_\mu=\overline{f^*}_{\overline{\mu}}=S\overline{f^*}_{\overline{\mu}}$.

This implies that $\overline{f^*}_{\mu}(x, .)$ is strongly  mixing probability on $(B^\mathbb{N},\mathcal{B}_{B^\mathbb{N}})$ $\overline{\mu}$-a.e. and $\mu$-a.e.
\end{proof}

\subsection{An AMS and output weakly mixing random channel makes a noisy computation AMS and ergodic.}

Proposition  \ref{PropositionQSMeanIsOWMRAMSSource} (p. \pageref{PropositionQSMeanIsOWMRAMSSource}) straightforwardly leads to the following corollary which gives a sufficient condition of ergodicity for noisy computations  relying only on properties of the input source and of the random channel $[A,F,C]$.

\begin{corollary}\label{CorollaryErgodicAMSNoisyComputation}
Let $[A, f^*\times F, B\times C]$ be a noisy computation and $\mu$ be the distribution of an AMS and ergodic source. If
\begin{itemize}
	\item  $\overline{\mu}(D_\beta(f^*))=1$
	\item and $[A, F , C]$ is AMS w.r.t. $\overline{\mu}$
	\item and $[A, \overline{F}_{\overline{\mu}}, C]$ is output weakly mixing w.r.t. $\overline{\mu}$
\end{itemize}
then the noisy computation $[A, f^*\times F, B\times C]$ is AMS and ergodic w.r.t $\mu$.
\end{corollary}
\begin{proof}
From Lemma \ref{LemmaStabilityofDeterministicChannel} (p. \pageref{LemmaStabilityofDeterministicChannel}), $f^*$ is R-AMS w.r.t. $\overline{\mu}$. Since $F$ is AMS w.r.t. $\overline{\mu}$, by Proposition \ref{PropositionCNSAMSChannelProduct} (p. \pageref{PropositionCNSAMSChannelProduct}), the noisy computation $[A, f^*\times F, B\times C]$ is AMS w.r.t. $\overline{\mu}$.

From Proposition \ref{PropositionQSMeanIsOWMRAMSSource} (p. \pageref{PropositionQSMeanIsOWMRAMSSource}), $[A,\overline{f^*}_{\mu},B]$ is output weakly mixing w.r.t. $\overline{\mu}$. Since $[A, \overline{F}_{\overline{\mu}}, C]$ is output weakly mixing w.r.t. $\overline{\mu}$, by Corollary \ref{CorollaryErgodicNoisyComputation} (p. \pageref{CorollaryErgodicNoisyComputation}), the noisy computation $[A, f^*\times F, B\times C]$ is ergodic w.r.t $\mu$

\end{proof}

Let $f:L\to B^+$ be a computable function. Let $[A, f^*, B]$ be the induced deterministic channel  (c.f. \ref{SubSectionInducedDeterministicChannel}) and $\mu$ be an AMS and ergodic source. Since $D_\beta(f^*)=A^\mathbb{N}$, $\mu(D_\beta(f^*))=\overline{\mu}(D_\beta(f^*))=1$. For any random channel $[A,F,C]$ AMS w.r.t. $\overline{\mu}$ such that $[A, \overline{F}_{\overline{\mu}}, C]$ is output weakly mixing w.r.t. $\overline{\mu}$, the noisy computation $[A, f^*\times F, B\times C]$ is AMS and ergodic w.r.t $\mu$.

In other words, for a function on finite length sequences (thus a computable function) $f$, it is sufficient to focus on stability properties of the noisy realization $F$ under concern in order to apply the coding results given above.

\subsection{Examples}

\subsubsection{Design error}

Let $f: L_f \to B^+$ be a Turing computable function. A design error (e.g., software bug) in the realization of $f$ can be viewed as a discrepancy between the algorithm specified by $f$ and the algorithm obtained through the (imperfect) design process. The obtained algorithm is also a Turing computable function  $g: L_g \to  B^+$. Indeed, $g$ is based on the same computational model as $f$. 

According to section \ref{SubSectionInducedDeterministicChannel}), $f$ and $g$ both determine deterministic channels $[A,f^*,B]$ and $[A,g^*,B]$. The computation of $f$ by $g$ is thus modeled by the noisy computation $[A, f^* \times g^* , B\times B]$. 

By Proposition \ref{PropositionQSMeanIsOWMRAMSSource} (p. \pageref{PropositionQSMeanIsOWMRAMSSource}), $[A,g^*,B]$ is AMS w.r.t  $\mu$ and  $[A,\overline{g^*_{\overline{\mu}}},B]$ is output weakly mixing w.r.t $\overline{\mu}$ for any AMS ergodic probability $\mu$.
According to  Corollary \ref{CorollaryErgodicAMSNoisyComputation} (p. \pageref{CorollaryErgodicAMSNoisyComputation}),  the noisy computation $[A, f^* \times g^* , B\times B]$ is AMS and ergodic w.r.t. any AMS and ergodic probability $\mu$.
 
\subsection{Noisy circuit}

Let $A=\{0,1\}^k$ and $B=\{0,1\}^{k'}$. Let $C$ be a circuit made of boolean gates and implementing a function $\phi : A \to B$. Let $f : A^+ \to B^+$ be the function such that $f(x_0 x_1 \cdots x_{n-1}) = \phi(x_0) \phi( x_1) \cdots \phi(x_{n-1})$. $f$ induces a deterministic channel $[A,f^*,B]$.

Let $\tilde{C}$ be a noisy realization of $C$ made of noisy boolean gates. Each noisy gate produces a correct output with probability $1-\xi$. Errors of different gates are assumed to be independent and the noisy circuit is assumed to be memoryless. 
Then successive uses of $\tilde{C}$ induces a memoryless and stationary random channel $[A,F,B]$. The computation of $f$ by successive uses of $\tilde{C}$ is thus modeled by the noisy computation $[A, f^*\times F, B\times B]$.

The channel $[A,F,B]$ is stationary and memoryless thus, for any AMS and ergodic source with distribution $\mu$, $\overline{F_{\overline{\mu}}}=F$. Moreover, being stationary and memoryless,  $[A,F,B]$ is output weakly mixing w.r.t. $\overline{\mu}$. Then, according to  Corollary \ref{CorollaryErgodicAMSNoisyComputation} (p. \pageref{CorollaryErgodicAMSNoisyComputation}),  the noisy computation $[A, f^* \times F , B\times B]$ is AMS and ergodic w.r.t. any AMS and ergodic probability $\mu$.

Assuming that $\tilde{C}$ is exactly $C$ made from noisy gates (they have the same structure), $\tilde{C}$ and $C$ have the same size (number of gates). Let apply the model of reliable computation of Section \ref{SectionJointSourceComputationCodingTheorem} to $f^*$ ($g=f^*$) and to the noisy computation $[A, f^*\times F, B\times B]$. Then, for large enough $k$ and $n$, it is possible to obtain $f(x_0 x_1 \cdots x_{k-1}) = \phi(x_0) \phi( x_1) \cdots \phi(x_{k-1})$ from a sequence of $n$ uses (or copies) of $\tilde{C}$ with an arbitrarily small error probability under the condition that:
$$
R=\frac{k}{n} \overline{H}(X') < C_{f^*}(F)
$$
where $\overline{H}(X')$ is the entropy of the input source.

Then instead of $k$ uses or copies of $C$ (i.e., a circuit of size $k.size(C)$), the reliable computation relies on $n$ uses or copies of $\tilde{C}$ (i.e., a circuit of size $n.size(\tilde{C})$). In addition to the encoder and decoder sizes, the "cost of reliability" is given by an increasing factor of circuit size $\Lambda$ which is 
$$
\Lambda = \frac{n.size(\tilde{C})}{k.size(C)}=\frac{n}{k}
$$
Assuming that the rate $R$ is close to capacity $C_{f^*}(F)$ ($C_{f^*}(F)$ assumed strictly positive) then 
$$
\frac{C_{f^*}(F)- \epsilon  }{\overline{H}(X')} \leq \frac{k}{n}< \frac{C_{f^*}(F)}{\overline{H}(X')}
$$
Then, for any $\epsilon >0$, it is possible to build a noisy circuit $F^n$ (i.e., $n$ copies of $\tilde{C}$) such that
$$
\frac{\overline{H}(X')}{C_{f^*}(F)} < \Lambda \leq \frac{\overline{H}(X')}{C_{f^*}(F)- \epsilon} 
$$
and giving an arbitrarily low error probability for large $k$ and $n$. 

This very simple example also complies with the general reliable computation model given in \cite{Romashchenko06}. Although the model of reliable computation  given in the present paper and the one of \cite{Romashchenko06} have different goals, they possess the common characteristic to clearly identify an encoding phase, a (noisy) computation phase and a decoding phase. Connections between these two approaches are certainly worth to be further investigated.

\pagebreak 

\section{Further work} \label{SectionFurtherWork}

In addition to try to extend other classical results from the channel case to the computation case (e.g., related to error probability), some specific questions regarding noisy computation could be addressed.

Relevant models of noisy Turing machines remain to be explored and analyzed regarding stability properties (ergodicity, asymptotic mean stationarity, ...). To the author's best knowledge, very few models of noisy Turing machines have been proposed and studied (\cite{Asarin05}, \cite{CapuniGacs12}). The examples given above are very simple and involve a noisy computing device $F$ which is stationary and memoryless.  If $F$ is a Turing machine whose behavior is altered by (e.g.) random perturbations of the transition function, stationarity and memorylessness will not hold anymore. 

Instantiation of the noisy and reliable computation models to peculiar formal computational models and their noisy counterpart deserves some work, one goal being to assess key parameters such as the cost of reliability. In particular, connections with the domain of reliable circuits built from noisy gates could be of significant interest. As sketch in the very simple example given above, it might be possible to obtain asymptotics of parameters such as the blow-up of circuits. 

The question of noisy encoding  has only been skimmed over but deserves in depth analysis through the point of view of cascaded noisy computations. This has probably some links with asymptotic properties  of restoring organs (\cite{Gacs05}).

\pagebreak 

\appendix

\section{Proof of  Lemma \ref{LemmaMaximalFeinsteinFamily} \label{SectionProofsOfLemmas}}

Proof of Lemma \ref{LemmaMaximalFeinsteinFamily} is based on the following lemma.

\begin{lemma}\label{LemmaMaximalFeinsteinFamily1}
Let $A$ and $C$ be standard alphabets and $B$ be a countable standard alphabet. Let $[A,f \times F,B\times C]$ be a noisy computation. Let $[A,X]$ be a source. Let $\epsilon>0$ and $\lambda>0$. Let $\tilde{B}$ be a measurable subset of $B^{\mathcal{I}}$. Assume there exist a set $\left\{ y_k \in \tilde{B} , k=1\cdots M \right\}$, a collection of measurable disjoint sets $\Gamma_k, k=1,...,M$ members of $\mathcal{B}_{C^{\mathcal{I}}}$ and a collection of measurable sets $A_k \subset f^{-1}(y_k), k=1,\cdots, M$ members of $\mathcal{B}_{A^{\mathcal{I}}}$ such that:
$$
\forall k=1, \cdots, M \text{  } P_{X|f(X)}(A_k|y_k)>1-\lambda \text{  and  } \forall x\in A_k \text{  }  P_{F(X)|X}\left( \Gamma_k^c | x  \right) \leq \epsilon 
$$
If the set $\left\{ y_k \in \tilde{B} , k=1\cdots M \right\}$ is maximal, meaning there do not  exist any other $y_{M+1}\in \tilde{B}$, $\Gamma_{M+1}$ and $A_{M+1}$ complying with the given properties, then $\forall y_0 \in \tilde{B}\setminus \{y_1,\cdots,y_M\} $ there exits $A_0 \subset f^{-1}(y_0)$ such that:
\begin{equation}
P_{X|f(X)}(f^{-1}(y_0) \setminus A_0|y_0)>\lambda \label{EqGreaterLambda}
\end{equation}
and
\begin{equation}
\forall x\in f^{-1}(y_0) \setminus A_0 \text{  }  P_{F(X)|X}\left( \bigcup_{k=1}^M \Gamma_k | x  \right) > \epsilon \label{EqGreaterEpsilon}
\end{equation}
\end{lemma}
\begin{proof}[Proof of Lemma \ref{LemmaMaximalFeinsteinFamily1}]
If $\tilde{B}\setminus \{y_1,\cdots,y_M\}=\emptyset$, the result is trivial.

Let $y_0  \in \tilde{B}\setminus \{y_1,\cdots,y_M\}$. Assume that there is no $A_0 \subset f^{-1}(y_0)$ for which the two statements (\ref{EqGreaterLambda}) and (\ref{EqGreaterEpsilon}) both hold. This implies that for {\em any} subset $A_0$ of $f^{-1}(y_0)$:
\begin{equation}
	P_{X|f(X)}(f^{-1}(y_0) \setminus A_0|y_0) \leq  \lambda \label{EqLowerLambda}
\end{equation}
or
\begin{equation}
\exists x\in f^{-1}(y_0) \setminus A_0 \text{  }  P_{F(X)|X}\left( \bigcup_{k=1}^M \Gamma_k | x  \right)  \leq \epsilon \label{EqLowerEpsilon}
\end{equation}
$A_0=\{ x\in f^{-1}(y_0) /  P_{F(X)|X}\left( \bigcup_{k=1}^M \Gamma_k | x  \right)  \leq \epsilon \}$ is a subset of $f^{-1}(y_0)$ for which (\ref{EqLowerEpsilon}) does not hold, then (\ref{EqLowerLambda}) hold: $P_{X|f(X)}(f^{-1}(y_0) \setminus A_0|y_0) \leq  \lambda$. Then 
$$
P_{X|f(X)}(A_0|y_0) > 1- \lambda \text{  and  } \forall x\in  A_0 \text{  }  P_{F(X)|X}\left( \bigcup_{k=1}^M \Gamma_k | x  \right)  \leq \epsilon
$$
Thus setting $\Gamma_0=C_\mathcal{I} \setminus \bigcup_{k=1}^M \Gamma_k$, $y_0$ can be added to the set $\{y_1,\cdots,y_M\}$, contradicting the maximality of this set.
\end{proof}

\begin{proof}[Proof of Lemma \ref{LemmaMaximalFeinsteinFamily}]
Since $f(X)\rightarrow X \rightarrow F(X)$ is a Markov Chain (from \cite{Gray11}, Chapter 2):
\begin{eqnarray}
P_{F(X)|fX)}\left( \bigcup_{k=1}^M \Gamma_k | y_0  \right) & = &  \int P_{F(X)|X}\left( \bigcup_{k=1}^M \Gamma_k | x  \right) dP_{X|f(X)}(x|y_0) \nonumber \\
								  & \geq & \int_{f^{-1}(y_0) \setminus A_0} P_{F(X)|X}\left( \bigcup_{k=1}^M \Gamma_k | x  \right) dP_{X|f(X)}(x|y_0) \nonumber
\end{eqnarray}
Then, by Lemma \ref{LemmaMaximalFeinsteinFamily1}
\begin{equation}\label{EqLemmaMaximalFeinsteinFamily}
P_{F(X)|fX)}\left( \bigcup_{k=1}^M \Gamma_k | y_0  \right) >  \lambda \epsilon 
\end{equation}
proving statement \ref{Statement1LemmaMaximalFeinsteinFamily}).
\begin{eqnarray}
P_{F(X)}\left( \bigcup_{k=1}^M \Gamma_k  \right)  & = &  \int_{\{y_1,\cdots,y_M\}} P_{F(X)|f(X)}\left( \bigcup_{k=1}^M \Gamma_k  | y \right) dP_{f(X)} \nonumber \\
							&  & + \int_{\{y_1,\cdots,y_M\}^c} P_{F(X)|f(X)}\left( \bigcup_{k=1}^M \Gamma_k  | y \right) dP_{f(X)} \nonumber
\end{eqnarray}
Since $f(X)\rightarrow X \rightarrow F(X)$ is a Markov chain:
\begin{multline}
P_{F(X)|f(X)}\left( \bigcup_{k=1}^M \Gamma_k  | y \right)=\int P_{F(X)|X}\left( \bigcup_{k=1}^M \Gamma_k  | x \right) dP_{X| f(X)} \nonumber \\
		\geq \int_{A_k} P_{F(X)|X}\left(  \bigcup_{k=1}^M \Gamma_k  | x \right) dP_{X| f(X)} \nonumber
\end{multline}
 By \ref{Statement1LemmaMaximalFeinsteinFamily}), $P_{F(X)|f(X)}\left( \bigcup_{k=1}^M \Gamma_k  | y \right) \geq \lambda\epsilon$ for any $y \notin \{y_1,\cdots,y_M\}$. Then 
\begin{eqnarray}
P_{F(X)}\left( \bigcup_{k=1}^M \Gamma_k  \right)	& > &  \int_{\{y_1,\cdots,y_M\}} \int_{A_k} P_{F(X)|X}\left(  \Gamma_k  | x \right) dP_{X| f(X)} dP_{f(X)} \nonumber \\
							& &  + \lambda \epsilon \int_{\tilde{B}\setminus \{y_1,\cdots,y_M\}}  dP_{f(X)} \nonumber \\
							& > &  \int_{\{y_1,\cdots,y_M\}} (1-\epsilon) P_{X| f(X)}(A_k|y_k)  dP_{f(X)} \nonumber \\
							& & + \lambda \epsilon \int_{\tilde{B}\setminus \{y_1,\cdots,y_M\}}  dP_{f(X)} \nonumber \\
							& > &  \int_{\{y_1,\cdots,y_M\}} (1-\epsilon) (1-\lambda)  dP_{f(X)} \nonumber \\
							& & + \lambda \epsilon \int_{\tilde{B}\setminus \{y_1,\cdots,y_M\}}  dP_{f(X)} \nonumber
\end{eqnarray}
$$
\Rightarrow P_{F(X)}\left( \bigcup_{k=1}^M \Gamma_k  \right) >min( (1-\lambda)(1-\epsilon) , \lambda\epsilon).P_{f(X)}(\tilde{B})
$$
proving statement \ref{Statement2LemmaMaximalFeinsteinFamily}).

Finally $y_1$, $A_1= f^{-1}(y_1)$ and $C_{\mathcal{I}}$ fulfill the properties. Thus $M\geq 1$.
\end{proof}

\begin{proof}[Proof of Lemma \ref{LemmaStationaryErgodicNoisyInputFunction}]

A cascade of stationary channels is a stationary channel (see \cite{Gray11}). Then it remains to prove that $F=\nu f$ is ergodic.
Let  $[A,X]$ be a stationary and ergodic source.  Let $O \in \mathcal{B}_{A^\mathcal{I} \times B^\mathcal{I} }$ be an invariant event: $(T_A\times T_B)^{-1} O = O$. 

Let $\tilde{O}=\left\{ (x,y) \in A^\mathcal{I} \times A^\mathcal{I} / (x, f(y)) \in O\right\}$.

$\tilde{O}_x= f^{-1}(O_x)$ since :
\begin{eqnarray}
\tilde{O}_x	& = &	\left\{ y\in A^\mathcal{I} / (x,y) \in \tilde{O} \right\} \nonumber \\
		& = &	\left\{ y\in A^\mathcal{I} / (x,f(y)) \in O \right\} \nonumber \\
		& = &	\left\{ y\in A^\mathcal{I} / f(y) \in O_x \right\} \nonumber
\end{eqnarray}

This implies $P_{XF(X)}(O) = P_{XY}(\tilde{O})$ since:
\begin{eqnarray}
P_{XF(X)}(O)	& = &	\int P_{F(X)|X}(O_x|x) dP_{X} \nonumber \\
		& = &	\int \int P_{f(Y)|Y}(O_x|y) dP_{Y|X} dP_{X} \nonumber \\
		& = &	\int \int 1_{O_x}(f(y)) dP_{Y|X} dP_{X} \nonumber \\
		& = & \int P_{Y|X}(f^{-1}(O_x)|x)  dP_{X} \nonumber \\
		& = & \int P_{Y|X}(\tilde{O}_x)|x)  dP_{X} \nonumber \\
		& = &  P_{XY}(\tilde{O}) \nonumber
\end{eqnarray}
$O$ is $(T_A\times T_B)$-invariant then $\tilde{O}$ is $(T_A\times T_A)$-invariant since
\begin{multline}
(T_A\times T_A)^{-1}\tilde{O} = 	\\
		\left\{ (x,y) \in A^\mathcal{I} \times A^\mathcal{I} / (T_A x, T_A y) \in \tilde{O} \right\} = \nonumber \\
						\left\{ (x,y) \in A^\mathcal{I} \times A^\mathcal{I} / (x, f(T_A y)) \in O \right\} \nonumber
\end{multline}
By stationarity of $f$:
\begin{multline}
(T_A\times T_A)^{-1}\tilde{O}  =	\\
		\left\{ (x,y) \in A^\mathcal{I} \times A^\mathcal{I} / (x, T_B f(y)) \in O \right\} =  \\
					 \left\{ (x,y) \in A^\mathcal{I} \times A^\mathcal{I} /   (x,f(y)) \in (T_A \times T_B)^{-1}(O) \right\} \nonumber 
\end{multline}
By invariance of $O$:
\begin{eqnarray}
(T_A\times T_C)^{-1}\tilde{O}	& = &	\left\{ (x,y) \in A^\mathcal{I} \times A^\mathcal{I} / (x, f(y)) \in O \right\} \nonumber \\
					& = &	\tilde{O} \nonumber
\end{eqnarray}
$P_{XY}$ is ergodic then $P_{XF(X)}(O) = P_{XY}(\tilde{O})$ is either 1 or 0. Thus $(X,F(X))$ is an ergodic process.
\end{proof}

\begin{proof}[Proof of Lemma \ref{LemmaStationaryErgodicFunctionChannelProduct}]

Let $[A,X]$ be a stationary source.  $f$ is a stationary function  then $f$ induces a stationary deterministic channel $[A,f,B]$.   $[A,f,B]$ is thus an ergodic channel (a deterministic stationary channel is ergodic, see \cite{Gray11}). $[A,F,C]$ is a stationary and ergodic channel. Then, all these assumptions imply that, for any stationary source $[A,X]$, $(X,f(X))$ and  $(X,F(X))$ are stationary and ergodic processes. To show that $[A, f\times F, B \times C]$ is stationary and ergodic, it remains to prove that $(f(X),F(X))$ is stationary and ergodic.

$f$ being a function, $F(X)\rightarrow X \rightarrow f(X)$ is a Markov chain, then $\forall G \in \mathcal{B}_{B^\mathcal{I}}$:
$$
P_{f(X)|XF(X)}(G|xz) = P_{f(X)|X}(G|x)  \text{  }P_{XF(X)}\text{-a.e.}
$$
then
$$
P_{f(X)|XF(X)}(T_B^{-1}G|xz)  =  P_{f(X)|X}(T_B^{-1}G|x)   \text{  }P_{XF(X)}\text{-a.e.}
$$
By stationarity of the deterministic channel $[A,f,B]$:
$$
P_{f(X)|XF(X)}(T_B^{-1}G|xz)  =  P_{f(X)|X}(G|T_A x)   \text{  }P_{XF(X)}\text{-a.e.}
$$
Using once more the fact that $F(X)\rightarrow X \rightarrow f(X)$ is a Markov chain:
$$
P_{f(X)|XF(X)}(T_B^{-1}G|xz)  =  P_{f(X)|XF(X)}(G|T_A xT_Bz)   \text{  }P_{XF(X)}\text{-a.e.}
$$
Then the channel $(X,F(X))\rightarrow f(X)$ is stationary. Since $(X,F(X))$ is a stationary process, this implies that $(X,f(X),F(X))$ is a stationary process and so is $(f(X),F(X))$.
 
Assume in addition that $[A,X]$ is ergodic. Let $O \in \mathcal{B}_{B^\mathcal{I} \times C^\mathcal{I} }$ be an invariant event: $(T_B\times T_C)^{-1} O = O$. Let $\tilde{O}=\left\{ (x,z) \in A^\mathcal{I} \times C^\mathcal{I} / (f(x),z) \in O\right\}$.

$\tilde{O}_z= f^{-1}(O_z)$ since :
\begin{eqnarray}
\tilde{O}_z	& = &	\left\{ x\in A^\mathcal{I} / (x,z) \in \tilde{O} \right\} \nonumber \\
		& = &	\left\{ x\in A^\mathcal{I} / (f(x),z) \in O \right\} \nonumber \\
		& = &	\left\{ x\in A^\mathcal{I} / f(x) \in O_z \right\} \nonumber
\end{eqnarray}

This implies $P_{f(X)F(X)}(O) = P_{XF(X)}(\tilde{O})$ since:
\begin{eqnarray}
P_{f(X)F(X)}(O)	& = &	\int P_{f(X)|F(X)}(O_z/z) dP_{F(X)} \nonumber \\
		& = &	\int P_{X|F(X)}(f^{-1}(O_z)/z) dP_{F(X)} \nonumber \\
		& = &	\int P_{X|F(X)}(\tilde{O}_z/z) dP_{F(X)} \nonumber
\end{eqnarray}
$O$ is $(T_B\times T_C)$-invariant then $\tilde{O}$ is $(T_A\times T_C)$-invariant since
\begin{multline}
(T_A\times T_C)^{-1}\tilde{O} = 	\\
		\left\{ (x,z) \in A^\mathcal{I} \times C^\mathcal{I} / (T_A x, T_C z) \in \tilde{O} \right\} = \nonumber \\
						\left\{ (x,z) \in A^\mathcal{I} \times C^\mathcal{I} / (f(T_A x), T_C z) \in O \right\} \nonumber
\end{multline}
By stationarity of $f$:
\begin{multline}
(T_A\times T_C)^{-1}\tilde{O}  =	\\
		\left\{ (x,z) \in A^\mathcal{I} \times C^\mathcal{I} / (T_B f(x), T_C z) \in O \right\} =  \\
					 \left\{ (x,z) \in A^\mathcal{I} \times C^\mathcal{I} /   (f(x), z) \in (T_B \times T_C)^{-1}(O) \right\} \nonumber 
\end{multline}
By invariance of $O$:
\begin{eqnarray}
(T_A\times T_C)^{-1}\tilde{O}	& = &	\left\{ (x,z) \in A^\mathcal{I} \times C^\mathcal{I} / (f(x), z) \in O \right\} \nonumber \\
					& = &	\tilde{O} \nonumber
\end{eqnarray}
$P_{X F(X)}$ is ergodic then $P_{f(X)F(X)}(O) = P_{XF(X)}(\tilde{O})$ is either 1 or 0. Thus $(f(X),F(X))$ is an ergodic process.
\end{proof}

\pagebreak

\bibliographystyle{elsarticle-harv}

\bibliography{AMSChannels,NoisyComputations}

\end{document}